\documentclass{article}

\usepackage{arxiv}

\usepackage{amsmath}
\usepackage{amssymb}
\usepackage{amsthm}
\usepackage{subcaption}
\usepackage{tikz}
\usepackage{booktabs}
\usepackage{braket}
\usepackage{color}
\usepackage{enumitem}
\usepackage{array}

\newtheorem{prop}{Proposition}
\newtheorem{definition}{Definition}
\newtheorem*{problem}{Problem}
\newtheorem{example}{Example}
\usepackage{dsfont}
\definecolor{Green}{rgb}{0.0, 0.5, 0.0}

\usepackage{graphicx,wrapfig}      
\usepackage{natbib}        

\newcommand{\uu}{\mathbf{u}}
\newcommand{\ww}{\mathbf{w}}
\newcommand{\x}{\mathbf{x}}
\newcommand{\y}{\mathbf{y}}

\newcommand{\R}{{\mathbb{R}}}

\newcommand{\N}{{\mathbb{N}}}

\newcommand{\filtration}{\mathds{F}}

\definecolor{amber}{rgb}{1.0, 0.49, 0.0}

\makeatletter
\newcommand{\safenewtheorem}[2]{\@ifundefined{#1}{\newtheorem{#1}{#2}}{}}
\makeatother
\theoremstyle{plain}
\safenewtheorem{thm}{Theorem}
\safenewtheorem{lem}{Lemma}
\safenewtheorem{cor}{Corollary}
\theoremstyle{definition}
\safenewtheorem{assum}{Assumption}
\safenewtheorem{scase}{Special Case}
\theoremstyle{remark}
\safenewtheorem{rem}{Remark}
\makeatletter
\@ifundefined{pf}{\newenvironment{pf}{\par\noindent\textit{Proof.}\ }{\par\medskip}}{}
\makeatother
\providecommand{\qed}{\ensuremath{\square}}

\mathchardef\mathcommachar="602C
\begingroup\catcode`,=\active
\gdef,{\mathcommachar\penalty0\relax}
\endgroup
\mathcode`\,="8000
\newif\ifshowrevisions
\showrevisionsfalse
\ifshowrevisions
  \definecolor{hlcolor}{rgb}{0,0,1}
\else
  \definecolor{hlcolor}{rgb}{0,0,0}
\fi
\let\hl\relax
\newcommand{\hl}[2][]{{\color{hlcolor}#2}}

\title{Compositional Approximately Bisimilar Abstractions of Stochastic Interconnected Systems
\thanks{ This work was partly supported by the ARTPARK and the ANR PIA funding: ANR-20-IDEES-0002.}
}

\author{
 S. Saharsh \\
  Department of Cyber-Physical Systems \\
  IISc, Bengaluru, India\\
  \texttt{saharsh2021@iisc.ac.in} 
  \And
 Sadek Belamfedel Alaoui \\
  College of Computing, University Mohammed VI Polytechnic\\
  Benguerir, Morocco \\
  \texttt{{ sadek.belamfedel@um6p.ma}} 
  \And
 Abdalla Swikir \\
  Mohamed Bin Zayed University of Artificial Intelligence (MBZUAI)\\
  Abu Dhabi, UAE, and Omar Al-Mukhtar University (OMU), Albaida, Libya \\
  \texttt{abdalla.swikir@mbzuai.ac.ae} 
  \And
  Pushpak Jagtap \\
  Department of Cyber-Physical Systems\\
  IISc, Bengaluru, India\\
  \texttt{pushpak@iisc.ac.in}
  \And
 Adnane Saoud \\
  College of Computing, University Mohammed VI Polytechnic\\
  Benguerir, Morocco \\
  \texttt{{adnane.saoud@um6p.ma}} 
}

\begin{document}
\maketitle

\begin{abstract}
This paper addresses the challenge of constructing (alternating) bisimilar symbolic models for discrete-time, stochastic, interconnected systems. By leveraging a refined over-approximation of the stochastic reachable sets and the $q$-th moment input-to-state stability ($\delta$-ISS-$M_q$) property, we establish an approximate bisimulation relation between stochastic systems and their deterministic (noise-free) counterparts. This approach significantly simplifies the analysis of the approximate (alternating) bisimulation between the global interconnected stochastic system and its symbolic model. We derive formal conditions under which local bisimulation relations are preserved when composing subsystems into a global interconnected system. 
The engineering relevance of the proposed compositional framework is demonstrated through temperature regulation case studies involving a fixed ten-room building and a circular cascade interconnection scaled from six to fourteen rooms. The results show that, beyond a small number of subsystems, the proposed method achieves substantially faster abstraction computation times than standard monolithic approaches, with the performance gap widening as the number of interconnected subsystems grows and the monolithic approach failing to converge beyond ten rooms.
\end{abstract}

\vspace{0.5em}
\noindent\textbf{Keywords:} Stochastic systems, Symbolic control, Compositional abstraction, Interconnected systems

\section{Introduction}

In recent years, formal synthesis methods have emerged as a rigorous paradigm for the design of stochastic control systems, attracting significant attention for their ability to automatically generate controllers that are correct-by-construction \cite{lavaei2022automated}. These approaches are particularly valuable for handling complex specifications, such as those expressed in linear temporal logic (LTL) or represented by automata on infinite strings \cite{baier2008principles}. Despite these advances, the automated synthesis of controllers for interconnected stochastic systems remains inherently challenging.

To address this complexity, the utilization of finite abstractions has become a standard methodology. This process involves abstracting the continuous dynamics into a finite set of states, performing analysis and synthesis on this abstract model using algorithmic techniques from the formal methods community \cite{baier2008principles, majumdar2024symbolic}, and mapping the controller back to the concrete system with guaranteed error bounds. The existing literature encompasses various abstraction techniques for stochastic systems, including finite bisimilar abstractions for randomly switched systems \cite{zamani2014symbolic}, incrementally stable switched systems \cite{zamani2015symbolic}, and continuous control systems \cite{zamani2014approximately}. Alternative approaches focus on constructing infinite abstractions for jump-diffusion systems using small-gain conditions \cite{zamani2016approximations}, utilizing stochastic reachability analysis for discrete-time hybrid systems \cite{ABATE20082724}, or proposing symbolic models based on input set discretization \cite{zamani2017towards} and omega-regular objectives \cite{dutreix2022abstraction}.

However, monolithic abstraction procedures frequently encounter the curse of dimensionality, rendering them computationally intractable for large-scale interconnected systems. To circumvent this, compositional approaches have been advocated, wherein a global abstraction is synthesized from local subsystem abstractions by exploiting the inherent modularity and interconnection topology of the underlying dynamics. This compositional paradigm imparts modularity, enabling the reconfiguration of interconnection topologies without necessitating a global re-computation, thus maintaining computational efficiency when integrating auxiliary subsystems.

Early contributions in the compositional domain predominantly addressed deterministic dynamics. Notable works include the employment of simulation functions and small-gain conditions \cite{pola2016symbolic}, partial feedback linearizability \cite{hussien2017abstracting}, joint dissipativity properties \cite{zamani2017compositional}, and alternating simulation functions \cite{swikir2019compositional}. Furthermore, \cite{saoud2021compositional} introduced a generalized framework based on ``approximate composition'' that obviates the need for standard small-gain conditions. More recently, compositional methodologies have been extended to the stochastic domain. Studies such as \cite{awan2023abstractions, lavaei2020compositional} leverage small-gain theorems and stochastic simulation functions, while others address finite Markov decision processes for switched systems \cite{lavaei2020automatica} or employ dissipativity-type properties \cite{Lavaei2022}.

Crucially, however, these existing stochastic frameworks typically establish an approximate (alternating)  {simulation} relation. The primary impetus for this research lies in the derivation of approximate (alternating)  {bisimulation} relations for the class of discrete-time stochastic interconnected systems. Existing compositional frameworks predominantly rely on simulation relations, which characterize a unidirectional behavioral inclusion essentially establishing that the abstract model over-approximates the reachable sets of the concrete dynamics. While such inclusion suffices for safety verification, it inherently introduces conservatism in control synthesis; specifically, the inability to synthesize a controller in the abstract domain does not strictly preclude its existence in the concrete system.
\hl{In contrast, the proposed (alternating) bisimulation relation enforces a bidirectional behavioral equivalence between the stochastic interconnected system and its deterministic abstraction. This stronger relational property is significant because, within the general theory of (alternating) bisimulation for symbolic control \cite{tabuada2009verification}, such bidirectional equivalence enables a control policy synthesized over the finite abstraction to be refined to the concrete stochastic domain with quantifiable error bounds, and underlies decidability results for the existence of valid control policies with respect to a broad class of specifications. In this paper, we rather focus on establishing this bisimulation relation compositionally for stochastic interconnected systems. We demonstrate the practical use of the constructed relation through a safety-controller case study (Section \ref{sec: case study}).}

In the broader context of controller synthesis, compositional methods allow for efficient distributed synthesis without requiring operations on the full state space. Early works utilized tiling \cite{le2016distributed} and small-gain theorems \cite{dallal2015compositional}, often employing assume-guarantee reasoning \cite{saoud2018composition, saoud2021assume} to decompose global control problems into independent local problems, as demonstrated for multi-periodic sampled-data systems in \cite{saoud2020contract}. While methods using disturbance bisimulation relations \cite{mallik2018compositional} and overlapping symbolic models \cite{meyer2017compositional} have attempted to reduce conservatism, they often restrict the feasible control action space or lack formal quantification via alternating bisimulation relations, a goal we achieve in this work. 

This paper extends the preliminary results presented in \cite{sadek2022compositional}. In this version, we extend the results to stochastic systems and introduce a novel approach to compute a less conservative upper bound for the distance between noisy and noise-free trajectories using the $q^{th}$-moment of the input-to-state stability concept. We provide detailed proofs, new examples, and a comprehensive case study illustrating compositional abstraction for discrete-time stochastic systems. Our contributions focus on four main aspects:
\begin{itemize}
    \item \hl{We build on the notion of approximate composition from \cite{saoud2021compositional} and extend it to establish compositional abstraction under approximate stochastic (alternating) bisimulation for interconnected stochastic systems, strengthening the existing one-directional simulation results into a bidirectional (alternating) bisimulation guarantee, which, by the general theory of approximate bisimulation \cite{tabuada2009verification}, supports controller refinement with quantified error.}
    \item We quantify the conservatism between discrete-time stochastic subsystems and their noiseless counterparts using the $\delta$-ISS-M$_q$ property via an approximate bisimulation relation. By transitivity, this establishes a guaranteed approximate bisimulation connection between the global interconnected stochastic system and the abstractions of the noiseless systems.
    \item We derive a tighter upper bound on the $q^{th}$-moment distance between the noisy and noise-free trajectories than the one obtained directly from the $\delta$-ISS-M$_q$ property, under a mild moment inequality on the diffusion term of the discrete-time stochastic system.
    \item Finally, we perform a comparative analysis against monolithic abstraction techniques to demonstrate the reduction in computation time for \hl{large-scale interconnected systems}. We validate the effectiveness and scalability of our approach through a case study involving temperature regulation in a ten-room building.
\end{itemize}
\hl{
In summary, the proposed framework differs from existing compositional
abstraction results in the following respects. Compared with~\cite{saoud2021compositional}, which establishes only approximate
(alternating) simulation between approximate and conservative compositions of
deterministic systems, we prove the converse direction under $\delta$-ISS,
yielding bisimulation.
Compared with the stochastic simulation-function framework
of~\cite{lavaei2020compositional}, which gives probabilistic output bounds under
a max small-gain condition, we give a $q$-th moment guarantee under approximate
composition and a $\delta$-ISS-type condition, and the relation is bidirectional.
Compared with~\cite{zamani2014symbolic,haesaert2017verification}, our
construction is compositional.
Our Lyapunov-free bound on the distance between noisy and noise-free trajectories covers $q=1$, which the Lyapunov-based
bound of~\cite[Lemma~3.8]{zamani2014symbolic} (stated for $q\ge 2$) does not.
}

\section{Preliminaries and Problem statement}
    
\textbf{Notations:}
The sets of positive integers, non-negative integers, real numbers, positive real numbers, and non-negative real numbers are denoted by $\mathbb{N}$, $\mathbb{N}_{0}$, $\mathbb{R}$, $\mathbb{R}^{+}$, and $\mathbb{R}_{0}^{+}$, respectively. 
The null vector of dimension $N \in \mathbb{N}$ is denoted by $\mathbf{0}_{N}:=(0, \ldots, 0)$. 
Given a set $X$, the identity map on $X$ is denoted by $\mathrm{id}_X(s)=s$ for all $s \in X$. A map $\mathbf{d}_{X}: X \times X \rightarrow \mathbb{R}_{0}^{+}$ is called a {pseudometric} if the following conditions hold for any $x_{1}, x_{2}, x_{3} \in X$: 
(i) $x_{1}=x_{2} \Rightarrow \mathbf{d}_{X}\left(x_{1}, x_{2}\right)=0$; 
(ii) $\mathbf{d}_{X}\left(x_{1}, x_{2}\right)=\mathbf{d}_{X}\left(x_{2}, x_{1}\right)$; and 
(iii) $\mathbf{d}_{X}\left(x_{1}, x_{3}\right) \leqslant \mathbf{d}_{X}\left(x_{1}, x_{2}\right)+\mathbf{d}_{X}\left(x_{2}, x_{3}\right)$. 
Given a relation $\mathcal{R} \subseteq A \times B$ \hl{that relates elements $a \in A$ with elements $b \in B$, we denote the relation by $a\mathcal{R}b$ whenever $(a, b) \in \mathcal{R}$.} The relation $\mathcal R$ can equivalently be viewed as the set-valued map $\mathcal{R}: A \rightrightarrows B$ defined 
by $b \in \mathcal{R}(a)$ if and only if $(a, b) \in \mathcal{R}$. The inverse relation is denoted by $\mathcal{R}^{-1} := \{(b,a) \in B \times A \mid (a,b) \in \mathcal{R}\}$. \hl{We denote the indicator function on the set $A$ as $\mathbf{1}_A$. A measurable space is a pair $(X, \mathcal{F})$ consisting of a sample space $X$ and $\sigma$-algebra $\mathcal{F}$. A Polish space $X$ is a separable, completely metrizable topological space. Unless otherwise specified, we equip $X$ with its Borel $\sigma$-algebra $\mathcal{B}(X)$. A probability measure $\mathbb P$ on measurable space $(X, \mathcal{F})$ defines probability space by the triplet $(X, \mathcal{F}, \mathbb P)$, where $x \sim \mathbb P$ represents a realization drawn according to $\mathbb P$. Boldface symbols (e.g., $\x$) denote random variables defined on $(\Omega, \mathcal{F}, \mathbb{P})$, and non-bold symbols (e.g., $x$) denote their realizations. The space of real-valued random variables with finite $q^{th}$ moment is denoted as $L^q$.} For a vector $x \in \mathbb{R}^n$ and a matrix $A \in \mathbb{R}^{n \times m}$, $\|x\|$ and $\|A\|$ denote the infinity norm and the induced infinity norm, respectively. 
For a discrete-time signal $\mathbf{x}:\mathbb{N}_0 \rightarrow X \subseteq \mathbb{R}^n$, the truncated supremum norm is denoted by $\|\mathbf{x}\|_{k} := \sup _{j \in \{0, \ldots, k\}} \|\mathbf{x}(j)\|$. A continuous function $\alpha: \mathbb{R}_{0}^{+} \rightarrow \mathbb{R}_{0}^{+}$ is of class $\mathcal{K}$ if $\alpha(0)=0$ and it is strictly increasing. If the function $\alpha$ is of class $\mathcal{K}$ and it is also unbounded, then it is of class $\mathcal{K}_{\infty}$. 
A function $\sigma: \mathbb{N}_{0} \rightarrow \mathbb{R}_{0}^{+}$ is of class $\mathcal{L}$ if it is strictly decreasing and $\lim _{k \rightarrow \infty} \sigma(k)=0$. 
A function $\beta: \mathbb{R}_{0}^{+} \times \mathbb{N}_{0} \rightarrow \mathbb{R}_{0}^{+}$ is of class $\mathcal{K} \mathcal{L}$ if $\beta(\cdot, k)$ is of class $\mathcal{K}$ for any fixed $k \in \mathbb{N}_0$, and $\beta(s, \cdot)$ is of class $\mathcal{L}$ for any fixed \hl{$s \in \mathbb{R}^+$}.

\subsection{Discrete-time stochastic control systems (DT-SCS)}
Let $(\Omega, \mathcal{F}, \mathbb{P})$ be a probability space endowed with a filtration $\filtration = (\mathcal{F}_k)_{k\in\mathbb{N}_0}$ satisfying the usual conditions of completeness and right continuity \cite[p.48]{karatzas1991brownian}. Let the sequence $(\ww(k))_{k \in \mathbb{N}_0}$ be a $p$-dimensional $\filtration$-adapted discrete-time stochastic process modeling the noise.

\begin{definition}
	\label{Def:DiscreteTimeSCS}
	A discrete-time stochastic control system is defined by the tuple
	$
	\Sigma = \bigl( \mathcal{X}, \mathcal{U}^{ext}, \mathcal{U}^{int}, \ww, f, g \bigr),
	$
	where:
	\begin{itemize}
		\item $\mathcal{X} \subseteq \mathbb{R}^{n}$ is the {Polish state space}, with Borel measurable space $(\mathcal{X}, \mathcal{B}(\mathcal{X}))$;
		\item $\mathcal{U}^{ext}$ is the set of external control inputs;
		\item $\mathcal{U}^{int}$ is the set of internal control inputs;
		\item $\ww := \{ \ww(k) \ | \ \ww(k): \Omega \rightarrow \mathcal{V}_{\ww}, k \in \mathbb{N}_{0} \}$ is the sequence of $p$-dimensional independent and identically distributed (i.i.d.) random variables taking values in the \hl{bounded set $\mathcal{V}_{\ww} \subset \mathbb{R}^{p}, \mathbf{0}_p \in \mathcal{V}_{\ww}$,} with mean $\varpi \in \mathbb{R}^p$ and covariance $\Lambda \in \mathbb{R}^{p \times p}$;
		\item the function $
		f: \mathcal{X}\times\mathcal{U}^{ext}\times\mathcal{U}^{int} \rightarrow \mathcal{X}$
		is a drift term;
		\item the function $
		g: \mathcal{X} \rightarrow \mathbb{R}^{n \times p}$
		is a diffusion term.
	\end{itemize}
    Then, a stochastic process $\x: \mathbb{N}_0 \rightarrow \mathcal{X}$ is said to be a  {solution process} of $\Sigma$ if there exist $ \uu^{ext}:\mathbb{N}_0 \rightarrow \mathcal{U}^{ext}$ and $\uu^{int}:\mathbb{N}_0 \rightarrow \mathcal{U}^{int}$ for any $k \in \mathbb{N}_{0}$ satisfying:
    \begin{equation}\label{dt_sys}   \x(k+1)=f\bigl(\x(k),\uu^{ext}(k),\uu^{int}(k)\bigr) +  g\bigl(\x(k)\bigr)\ww (k).
    \end{equation}
\end{definition}
Similarly, we define the noise-free version $\bar{\Sigma}$ of the DT-SCS in \eqref{dt_sys} as
\begin{align}\label{nfe} 
{\x}(k+1) = f({\x}(k), \uu^{ext}(k), \uu^{int}(k)).
\end{align}
We use $\x_{x\varsigma}(k)$ to represent the value of the solution process $\x_{x\varsigma}$ of $\Sigma$ at the $k^{th}$ instant, starting from the initial state $x\in\mathcal{X}$ under $\varsigma := (\uu^{ext}, \uu^{int})$, where the external and internal input signals $ \uu^{ext}:\mathbb{N}_0 \rightarrow \mathcal{U}^{ext}$ and $\uu^{int}:\mathbb{N}_0 \rightarrow \mathcal{U}^{int}$.
In the sequel, we simplify notation by using
$f_k\bigl(\x,\uu^{ext},\uu^{int}\bigr) \equiv f\bigl(\x(k),\,\uu^{ext}(k),\,\uu^{int}(k)\bigr),
g_k\bigl(\x\bigr) \equiv g\bigl(\x(k)\bigr).$

The stability of DT-SCS is characterized by incremental Input-to-State Stability (ISS) in the $q^{th}$ moment, which is stated as follows.
\subsection{Incremental Input-to-State Stability in $q^{th}$-moment for DT-SCS}

We introduce a discrete-time version of the notion of incremental input-to-state stability in the $q^{th}$-moment ($\delta$-ISS-$M_{q}$), originally proposed in \cite{zamani2014symbolic, jagtap2017backstepping} for continuous-time stochastic systems.

\begin{definition} \label{DefD-ISS-M-q} 
    Consider the stochastic control system $\Sigma$ given in \eqref{dt_sys}. The system $\Sigma$ is said to be globally incrementally input-to-state stable in the $q^{th}$-moment ($\delta$-ISS-$M_{q}$), with $q\in \N$, if there exists a class $\mathcal{K} \mathcal{L}$ function $\beta$, a class $\mathcal{K}_{\infty}$ function $\gamma$, and a constant $\psi \in \mathbb{R}^{+}_{0}$ such that for any $k \in \mathbb{N}_{0}$, any pair of initial states $x_{1}, x_{2} \in \mathcal{X}$, and any input signals $\mathbf{u}_{1}^{int}, \mathbf{u}_{2}^{int}: \mathbb{N}_0\rightarrow \mathcal{U}^{int}$ and $\mathbf{u}_{1}^{ext}, \mathbf{u}_{2}^{ext} : \mathbb{N}_0\rightarrow \mathcal{U}^{ext}, \varsigma_1 = (\mathbf{u}_{1}^{ext}, \mathbf{u}_{1}^{int}), \varsigma_2 = (\mathbf{u}_{2}^{ext}, \mathbf{u}_{2}^{int})$, the following inequality holds:
    \begin{align}\label{D-ISS-M-q}
        \mathbb{E}&\left[\| \x_{x_1\varsigma_1}(k)- \x_{x_2\varsigma_2}(k) \|^{q}\right] \nonumber \\
        &\leqslant \beta\left(\|x_{1}-x_{2}\|^{q}, k\right) + \gamma^{int}\left( \left\|\mathbf{u}_{1}^{int}-\mathbf{u}_{2}^{int}\right\|_{k-1}\right)\nonumber\\
        &\quad + \gamma^{ext}\left( \left \|\mathbf{u}_{1}^{ext}-\mathbf{u}_{2}^{ext}\right\|_{k-1}\right) +\psi, 
    \end{align}
    where the second and third terms on the right-hand side of \eqref{D-ISS-M-q} are defined as $0$ for $k=0$.
\end{definition}

\begin{rem}
    It can be readily verified that the $\delta$-ISS-$M_{q}$ property for the stochastic control system $\Sigma$ reduces to the standard $\delta$-ISS property \cite{bayer2013discrete,sadek2022compositional} in the case of the noise-free control system $\bar{\Sigma}$ (where $q=1$ and $\psi=0$). In such cases, \eqref{D-ISS-M-q} simplifies to:
    \begin{align}\label{D-ISS}
        &\|\x_{x_1\varsigma_1}(k)- \x_{x_2\varsigma_2}(k)\| \nonumber \\
        &\quad \leqslant \beta\left(\|x_{1} - x_{2}\|, k\right) +\gamma^{int}\left( \left\|\mathbf{u}_{1}^{int}-\mathbf{u}_{2}^{int}\right\|_{k-1}\right) \nonumber \\
        &\quad + \gamma^{ext}\left( \left \|\mathbf{u}_{1}^{ext}-\mathbf{u}_{2}^{ext}\right\|_{k-1}\right).
    \end{align}
\end{rem}

We derive the explicit maps $\beta$, $\gamma^{ext}$, and $\gamma^{int}$ that characterize the $\delta$-ISS-$M_{q}$ property for both linear and nonlinear implementations of the DT-SCS $(\Sigma)$ in Appendix \ref{ISSl} and Appendix \ref{ISSnl}, respectively.

\subsection{DT-SCS as a gMDP}
\label{sec: dtss-ts}
\hl{We can model DT-SCS as a general Markov decision process (gMDPs), using the definition in \cite{wang2026correct}, as below.
\begin{definition}\label{def: gmdp}
A general Markov Decision Process (gMDP) is a tuple $ M=\left(X,X^{0}, U^{{ext}}, U^{int}, \mathbb{T}, Y, H\right) $, where Polish state space $X$ is the set of states, $X^{0} \subseteq X$ is the set of initial states, Polish input spaces $U^{ext}$ and $U^{{int}}$ are the set of external and internal inputs, respectively, stochastic kernel $\mathbb{T} : X \times U^{ext} \times U^{int} \times \mathcal{B}(X) \rightarrow [0, 1]$ is a map that assigns each state-input tuple $(x, u^{ext}, u^{int}), x \in X, u^{ext} \in U^{ext}, u^{int} \in U^{int}$ a probability measure $\mathbb{T}(.\vert x, u^{ext}, u^{int})$ over the measurable successor state space $(X, \mathcal{B}(X))$, $Y$ is the set of outputs, and $H: X \rightarrow Y$ is the measurable output map. 
\end{definition}
}
A DT-SCS as in \eqref{dt_sys} is represented by gMDP as below. 
\begin{equation}\label{dtss-trans}
    M(\Sigma) = \left(X,\, X^{0},\, {U}^{ext},\, {U}^{int},\,\mathbb{T},\, Y,\, H\right),
\end{equation}
where the set of states $X = \mathcal{X}$, a Polish state space, is associated with the Borel measurable space $({X}, \mathcal{B}({X}))$, the set of initial states is $X^{0} \subseteq {X}$, and the Polish input spaces, is denoted as input sets ${U}^{ext} = \mathcal{U}^{ext}, {U}^{int} = \mathcal{U}^{int}$. We define the stochastic kernel $\mathbb T : X \times U^{ext} \times U^{int} \times \mathcal{B}(X) \to [0,1]$ by
\begin{equation}\label{kernel}
    \mathbb T(A \mid x, u^{ext}, u^{int}) := \mathbb{P}\big(f(x,u^{ext},u^{int}) + g(x)\ww \in A\big), \qquad \forall A \in \mathcal{B}(X),
\end{equation}
i.e., $\mathbb T(\cdot \mid x, u^{ext}, u^{int})$ is the probability distribution induced on $X$ by the measurable map $w \mapsto f(x,u^{ext},u^{int}) + g(x)w$, where $w$ denotes a realization of the random variable $\ww$, with $w\in\mathcal{V}_{\ww}$.
Consequently, we denote the set of admissible inputs for a state $x \in X$ by $U_{M}(x)$ and define it formally as follows.
\begin{equation}\label{adm_M}
    U_{M}(x) := \big\{ (u^{ext},u^{int}) \in U^{ext}\times U^{int} \;\big|\; \mathbb{T}(X \mid x, u^{ext},u^{int}) = 1 \big\}.
\end{equation}
This ensures that a pair of inputs is allowed at state $x$ only if the next state of \eqref{dt_sys} is guaranteed almost surely to remain within the set of states $X$ for all possible noise realizations. The output set is $Y = X$, and $H$ is the identity map $\mathrm{id}_X(\cdot)$. 

Next, we define the notion of a transition system for modeling the noise-free version of DT-SCS.
\subsection{Noise-free system as Transition system}
We define the notion of  {transition systems} adopted from \cite{tabuada2009verification}, which permits us to represent non-stochastic control systems and their abstractions in a unified way.

\begin{definition}\label{DefTransitionSys}
	A transition system is a tuple $ S=\left(X, X^{0}, U^{{ext}}, U^{int}, \Delta, Y, H\right) $, where $X$ is the set of states (possibly infinite), $X^{0} \subseteq X$ is the set of initial states, $U^{ext}$ and $U^{{int}}$ are the set of external and internal inputs (possibly infinite), respectively, $\Delta \subseteq X \times U^{ext} \times U^{int} \times X$ is the transition relation, $Y$ is the set of outputs, and $H: X \rightarrow Y$ is the output map. 
\end{definition}

The set of admissible inputs for $x\in X$ is denoted by $U_{S}(x):= \{({u}^{ext}, {u}^{int}) \in U^{ext} \times U^{int} \mid  \Delta(x, u^{ext}, {u}^{int}) \neq \emptyset \}.$ We denote by $x^{\prime} \in \Delta\left(x, u^{{ext }}, {u}^{int}\right)$ an alternative representation for a transition $\left(x, {u}^{ext},{u}^{int}, x^{\prime}\right) \in \Delta$, where state $x^{\prime}$ is called a $\left({u}^{ext}, {u}^{int}\right)$-successor (or simply successor) of state $x$. A transition system $S$ is said to be:
\begin{itemize}
	\item[$\bullet$] \textbf{ {pseudometric}}, if the state set $X$, input sets $U^{ext},U^{int}$ and the output set $Y$ are equipped with pseudometrics $\mathbf{d}_X:X\times X\rightarrow \mathbb{R}_0^+$, $\mathbf{d}_{U^{int}}: U^{int} \times U^{int} \rightarrow \mathbb{R}_{0}^{+}$, $\mathbf{d}_{U^{ext}}: U^{ext} \times U^{ext} \rightarrow \mathbb{R}_{0}^{+}$ and $\mathbf{d}_Y: Y \times Y \rightarrow \mathbb{R}_{0}^{+}$;
	\item[$\bullet$] \textbf{ {finite}}, if $X$, $ U^{{int}} $, $ U^{{ext}} $ and $Y$ are finite sets;
    \item[$\bullet$] \textbf{ {deterministic}}, if there is at most one $\left({u}^{ext}, {u}^{int}\right)$-successor, for any $x \in X$ and $\left(u^{{ext}}, u^{{int}}\right) \in$ $U^{{ext}} \times U^{{int}}$.
\end{itemize}


We represent the transition system associated to the noise-free version $\bar{\Sigma}$ in \eqref{nfe} by
\begin{equation}\label{dtss-nf-trans}
    S(\bar{\Sigma}) = \left(X,\, X^{0},\, {U}^{ext},\, {U}^{int}, {\Delta},\, Y,\, H\right),
\end{equation}
where the set of states $X = \mathcal{X}$,
the set of initial states is $X^{0} \subseteq {X}$, and the input sets are ${U}^{ext} = \mathcal{U}^{ext}, {U}^{int} = \mathcal{U}^{int}$. We define the transition relation $\Delta$ such that $x' \in \Delta(x, u^{ext}, u^{int}) \iff x' \in f(x, u^{ext}, u^{int})$.
For the transition system in \eqref{dtss-nf-trans}, we denote the set of admissible inputs for a state $x \in X$ by $U_{S}(x)$ and define it formally as follows.
\begin{equation}\label{adm_S}
    U_{S}(x) := \big\{ (u^{ext},u^{int}) \in U^{ext}\times U^{int} \;\big|\; \Delta(x, u^{ext},u^{int}) \subseteq X,\ \Delta(x, u^{ext},u^{int}) \neq \emptyset \big\}.
\end{equation}
This ensures that a pair of inputs is allowed at state $x$ only if the next state of \eqref{nfe} is guaranteed to remain within the set of states $X$. The output set is $Y = X$, and $H$ is the identity map $\mathrm{id}_X(\cdot)$.
\begin{rem}\label{SS: rem/gMDP-TS}
The gMDP $M(\Sigma)$'s or its noise-free version $S(\bar\Sigma)$'s solution matches, i.e.,
$M(\Sigma) = S(\bar\Sigma)$, precisely when the process noise
vanishes, i.e., $g(x) \equiv 0$ for all $x \in X$ or $\mathcal{V}_{\ww} = \{0\}$, so that $\mathbb T(A \mid x, u^{ext}, u^{int}) = \mathbf{1}_A\big(f(x,u^{ext},u^{int})\big), \forall A \in \mathcal{B}(X),$ i.e., $\mathbb T(\cdot \mid x, u^{ext}, u^{int})$ reduces to the Dirac measure.
Similarly, we model Transition system $S(\bar \Sigma )$ as gMDP $M(S(\bar \Sigma)) := (X, X^{0}, U^{ext}, U^{int}, \mathbb{T}_{\Delta}, Y, H)$, where the stochastic kernel is defined as
\begin{equation}\label{SS: eqn/rem1}
    \mathbb T_\Delta(A \mid x, u^{ext}, u^{int}) = \mathbf{1}_A\big(\Delta(x, u^{ext}, u^{int})\big), \qquad \forall A \in \mathcal{B}(X),
\end{equation}
i.e., the Dirac measure concentrated at the reachable successor subspace.
\end{rem}

\subsection{Approximate (alternating) stochastic bisimulation relations}

In this subsection, we establish the stochastic (alternating) bisimulation relation between stochastic gMDPs, as defined in \eqref{dtss-trans}. We further define the corresponding (alternating) bisimulation relation between a stochastic gMDP and the deterministic transition system, as well as between transition systems. Consider two gMDPs $M_{1}=(X_{1}, X_{1}^{0}, U_{1}^{ext}, U_{1}^{int}, \mathbb{T}_{1}, Y_{1}, H_{1})$ and $M_{2}=(X_{2}, X_{2}^{0}, U_{2}^{ext}, U_{2}^{int}, \mathbb{T}_{2}, Y_{2}, H_{2})$, two transition systems $S_{1}=(X_{1}, X_{1}^{0}, U_{1}^{ext}, U_{1}^{int}, \Delta_{1}, Y_{1}, H_{1})$ and $S_{2}=(X_{2}, X_{2}^{0}, U_{2}^{ext}, U_{2}^{int}, \Delta_{2}, Y_{2}, H_{2})$. 
We assume that 
\begin{equation}\label{assum_sp}
    Y_1,Y_2\subseteq Y, \qquad U_j^{ext}\subseteq U^{ext}, \qquad U_j^{int}\subseteq U^{int}, \quad j\in\{1,2\},
\end{equation}where $(Y,\mathbf{d}_Y)$, $(U^{ext},\mathbf{d}_{U^{ext}})$, $(U^{int},\mathbf{d}_{U^{int}})$ are pseudometric spaces.
\hl{For defining equivalence type relations between gMDPs, we need coupling of probability measures $\mathbb{T}_1, \mathbb{T}_2$ over Polish spaces. We borrow the notion of $\hl{\varrho}-$lifting from \cite[Definition 5]{haesaert2017verification} for coupling of the measures, rewritten as below.
\begin{definition}\label{DefinitionLifting}
    Let $X_1, X_2$ be two sets with associated probability spaces $(X_1, \mathcal{B}(X_1), \Gamma), (X_2, \mathcal{B}(X_2), \Theta)$, and let $\mathcal{R} \subseteq X_1 \times X_2$ be a relation such that $\mathcal{R} \in \mathcal{B}(X_1 \times X_2)$. For $\hl{\varrho} \in [0,1]$, we denote by $\bar{\mathcal{R}}_{\hl{\varrho}} \subseteq \mathbb{P}(X_1, \mathcal{B}(X_1)) \times \mathbb{P}(X_2, \mathcal{B}(X_2))$ the corresponding lifted relation (stated as $\Gamma \bar{\mathcal{R}}_{\hl{\varrho}} \Theta$), if there exists a probability space $\mathbb{S} = (X_1 \times X_2, \mathcal{B}(X_1 \times X_2), \mathbb{W})$ satisfying
    \begin{itemize}
        \item[(i)] for all $\mathbf{X}_1 \in \mathcal{B}(X_1): \mathbb{W}(\mathbf{X}_1 \times X_2) = \Gamma(\mathbf{X}_1)$;
        \item[(ii)] for all $\mathbf{X}_2 \in \mathcal{B}(X_2): \mathbb{W}({X}_1 \times \mathbf{X}_2) = \Theta(\mathbf{X}_2)$;
        \item[(iii)] for the probability space $\mathbb{S}$ and for a random pair $(s, t) \sim \mathbb{W}$, it holds that $s\mathcal{R}t$ with probability at least $1 - \hl{\varrho}$ or $\mathbb{W}(\mathcal{R}) \geq 1 - \hl{\varrho}$. 
    \end{itemize}
\end{definition}}
\hl{Intuitively, $\bar{\mathcal{R}}_{\hl{\varrho}}$ lifts the point-to-point relation $\mathcal{R}$ to probability distributions by requiring the existence of a joint coupling $(s, t) \sim \mathbb{W}$ (with marginal distributions $\Gamma$ and $\Theta$) such that $s\mathcal{R}t$ holds with probability at least $1-\hl{\varrho}$.}
We take advantage of the above defined $\hl{\varrho}-$lifting relation to define approximate stochastic simulation relations, which are as follows. 
\begin{definition}\label{DefinitionBisimulation}
    For $\varepsilon, \mu \in \R_0^{+}, \hl{\hl{\varrho} \in [0,1]}$, $M_{1}$ is said to be $(\varepsilon, \mu, \hl{\varrho})$-approximately \hl{stochastically} simulated by $M_{2}$, if \hl{\eqref{assum_sp} holds} and there exists a relation $\mathcal{R} \subseteq X_{1} \times X_{2}$, for which there exists a \hl{Borel measurable coupling kernel $\mathbb{W}_{\mathbb{T}}(. \vert u^{ext}_1, u^{int}_1, x_1, x_2)$ on $X_1 \times X_2$} satisfying:
    \begin{itemize}
        \item[(i)] $\forall x_{1}^{0} \in X_{1}^{0}, \exists x_{2}^{0} \in X_{2}^{0}$ such that $\left(x_{1}^{0}, x_{2}^{0}\right) \in \mathcal{R}$; 
        \item[(ii)] $\forall\left(x_{1}, x_{2}\right) \in \mathcal{R}, \mathbf{d}_{Y}\left(H_{1}\left(x_{1}\right), H_{2}\left(x_{2}\right)\right)\leqslant \varepsilon$;
        \item[(iii)] {$\forall\left(x_{1}, x_{2}\right) \in \mathcal{R}, \forall\left({u}_{1}^{ext}, u_{1}^{int}\right) \in U_{M_{1}}(x_{1}),$  $ \exists\left({u}_{2}^{ext}, u_{2}^{int}\right) \in U_{M_{2}}\left(x_{2}\right)$ with  $\max (\mathbf{d}_{{U}^{ext}}(u_{1}^{ext}, {u}_{2}^{ext}), \mathbf{d}_{{U}^{int}}({u}_{1}^{int}, {u}_{2}^{int})) \leqslant \mu$ and \hl{$\mathbb{T}_1(. \vert x_1, u^{ext}_1, u^{int}_1) \bar{\mathcal{R}}_{\hl{\varrho}} \mathbb{T}_2(. \vert x_2, u^{ext}_2, u^{int}_2)$, with lifted probability measure $\mathbb{W}_{\mathbb{T}}(. \vert u^{ext}_1, u^{int}_1, x_1, x_2)$.} }
    \end{itemize}
    Moreover, $M_{1}$ is said to be $(\varepsilon,\mu, \hl{\hl{\varrho}})$-approximately stochastically bisimilar to $M_{2}$, if $M_{1}$ is $(\varepsilon,\mu, \hl{\hl{\varrho}})$-approximately stochastically simulated by $M_{2}$, and $M_{2}$ is $(\varepsilon,\mu, \hl{\varrho})$-approximately stochastically simulated by $M_{1}$. Stochastic simulation and bisimulation relations are denoted, respectively, by $M_{1} \preccurlyeq^{\varepsilon, \mu}_{\hl{\varrho}} M_{2}$ and $M_{1} \approx^{\varepsilon, \mu}_{\hl{\varrho}} M_{2}$.
\end{definition}
Now, we consider different variants of the stochastic bisimulation relation defined above under two special cases.
\hl{
\begin{scase}(Stochastic Bisimulation between gMDP $M_1$ and Transition system $S_2$)
\label{SS: scase/bisim1}
Consider the deterministic transition system $S_2$ as the gMDP $M(S_2)$ (Remark~\ref{SS: rem/gMDP-TS}), whose kernel
$\mathbb{T}_2(\cdot\vert x_2,u^{ext}_2,u^{int}_2)$ is the Dirac measure on
$\Delta_2(x_2,u^{ext}_2,u^{int}_2)$. The relation is now between a random variable and a point, where $(\x_1,x_2)\in\mathcal{R}$ relates
a (random) state $\x_1$ of $M_1$ to a deterministic state $x_2$ of $S_2$.
$M_1$ is $(\varepsilon,\mu,\varrho)$-approximately stochastically simulated by $M(S_2)$ if there
exists such a relation $\mathcal{R}$ with $q \in \N$, for any realization $x_1 \sim \x_1,$ and a Borel measurable coupling kernel
$\mathcal{W}_{\mathbb{T}}(\cdot\vert u^{ext}_1,u^{int}_1,x_1,x_2)$ satisfying:
\begin{enumerate}
\item[(i)] $\forall x_1^0\in X_1^0$, $\exists x_2^0\in X_2^0$ such that $(x_1^0,x_2^0)\in\mathcal{R}$;
\item[(ii)] $\forall(\mathbf{x}_1,x_2)\in\mathcal{R}$,
$\mathbb{E}\big[\mathbf{d}_Y\big(H_1(\mathbf{x}_1),H_2(x_2)\big)^q\big]^{\frac1q}\leqslant\varepsilon$;
\item[(iii)] $\forall(\mathbf{x}_1,x_2)\in\mathcal{R}$, $\forall(u^{ext}_1,u^{int}_1)\in U_{M_1}({x}_1)$,
$\exists(u^{ext}_2,u^{int}_2) \in U_{M(S_2)}(x_2)$ with 
$\max\big(\mathbf{d}_{U^{ext}}(u^{ext}_1,u^{ext}_2),\mathbf{d}_{U^{int}}(u^{int}_1,u^{int}_2)\big)\leqslant\mu$
and $\mathbb{T}_1(\cdot\vert{x}_1,u^{ext}_1,u^{int}_1)\,\overline{\mathcal{R}}_\varrho\,\mathbb{T}_2(\cdot\vert x_2,u^{ext}_2,u^{int}_2)$
with lifted measure $\mathcal{W}_{\mathbb{T}}$.
\end{enumerate}
Since membership of the successor pair $(\x_1', x_2'), \x_1' = f(\x_1, \uu^{ext}, \uu^{int}) + g(\x_1)\ww, x_1' \sim \mathbb T(.|x_1, u_1^{ext}, u_1^{int}), x_2' \sim \mathbb T(.|x_2, u_2^{ext}, u_2^{int}),$ in $\mathcal{R}$ is a deterministic property (the moment
bound holds or fails), $\varrho\in\{0,1\}$. Bisimulation requires the converse simulation, and is denoted
$M_1\approx^{\varepsilon,\mu}_{\varrho}M(S_2)$.
\end{scase}}
\hl{Note that in contrast with Definition \ref{DefinitionBisimulation}, here $\mathcal R$ collects pairs $(\x_1,x_2)$ with $\x_1$ an $L^q$ random
variable and $x_2\in X_2$, where the successor of $\x_1$ is generated by
$\mathbb T_1$ given the past, and condition (ii) becomes the deterministic moment
condition $\mathbb{E}\big[\mathbf{d}_Y(H_1(\x_1),H_2(x_2))^q\big]^{1/q}\le\varepsilon$;
hence $\varrho=0$ whenever this bound holds at the successor pair.}
\hl{\begin{scase} (Bisimulation between Transition systems $S_1, S_2$) :
 As Transition systems do not require equivalence relations between probability measures, we further simplify condition (iii) in Definition \ref{DefinitionBisimulation} as below.
 \begin{itemize}
     \item[(iii)] $\forall\left(x_{1}, x_{2}\right) \in \mathcal{R}, \forall\left({u}_{1}^{ext}, u_{1}^{int}\right) \in U_{S_{1}}(x_{1}), \forall x_{1}^{\prime} \in \Delta_{1}\left(x_{1}, {u}_{1}^{ext}, u_{1}^{int}\right), \exists\left({u}_{2}^{ext}, u_{2}^{int}\right) \in U_{S_{2}}\left(x_{2}\right)$ with $\max (\mathbf{d}_{{U}^{ext}}(u_{1}^{ext}, {u}_{2}^{ext}), \mathbf{d}_{{U}^{int}}({u}_{1}^{int}, {u}_{2}^{int})) \leqslant \mu$ and $\exists x_{2}^{\prime} \in \Delta_{2}\left(x_{2}, {u}_{2}^{ext}, {u}_{2}^{int}\right)$ such that
         $\left(x_{1}^{\prime}, x_{2}^{\prime}\right) \in \mathcal{R}$.
 \end{itemize}
 Together with conditions (i) and (ii), we define the bisimulation relation (not stochastic) between Transition systems, denoted as $S_{1} \approx^{\varepsilon, \mu} S_{2}$. \label{SS: scase/bisim2}
\end{scase}}
In verification problems, approximate (bi)simulation relations are commonly utilized. Specifically, the concept of approximate alternating (bi)simulation relations, as introduced in \cite{tabuada2009verification}, is particularly relevant when the objective is to synthesize controllers. These relations play a crucial role in ensuring correct behavior and correctness guarantees while addressing compositional controller synthesis scenarios.

\begin{definition}\label{DefinitionAlternatingBisimulation}
For $\varepsilon, \mu \in \R_0^{+}, \hl{\hl{\varrho} \in [0,1]}$, $M_{1}$ is said to be $(\varepsilon, \mu, \hl{\varrho})$-approximately \hl{stochastically} alternatingly simulated by $M_{2}$, if \hl{\eqref{assum_sp} holds} and there exists a relation $\mathcal{R} \subseteq X_{1} \times X_{2}$, for which there exists a \hl{Borel measurable coupling kernel $\mathbb{W}_{\mathbb{T}}(. \vert u^{ext}_2, u^{int}_2, x_1, x_2)$ on $X_1 \times X_2$} satisfying:
    \begin{itemize}
        \item[(i)] $\forall x_{1}^{0} \in X_{1}^{0}, \exists x_{2}^{0} \in X_{2}^{0}$ such that $\left(x_{1}^{0}, x_{2}^{0}\right) \in \mathcal{R}$;
        \item[(ii)] $\forall\left(x_{1}, x_{2}\right) \in \mathcal{R}, \mathbf{d}_{Y}\left(H_{1}\left(x_{1}\right), H_{2}\left(x_{2}\right)\right) \leqslant \varepsilon$;
        \item[(iii)] $\forall\left(x_{1}, x_{2}\right) \in \mathcal{R}, \forall\left({u}_{2}^{ext}, u_{2}^{int}\right) \in U_{M_{2}}(x_{2}),$  $ \exists\left({u}_{1}^{ext}, u_{1}^{int}\right) \in U_{M_{1}}\left(x_{1}\right)$ with  $\max (\mathbf{d}_{{U}^{ext}}(u_{1}^{ext}, {u}_{2}^{ext}), \mathbf{d}_{{U}^{int}}({u}_{1}^{int}, {u}_{2}^{int})) \leqslant \mu$ and \hl{$\mathbb{T}_1(. \vert x_1, u^{ext}_1, u^{int}_1) \bar{\mathcal{R}}_{\hl{\varrho}} \mathbb{T}_2(. \vert x_2, u^{ext}_2, u^{int}_2)$, with lifted probability measure $\mathbb{W}_{\mathbb{T}}(. \vert u^{ext}_2, u^{int}_2, x_1, x_2)$.}
    \end{itemize}
    Moreover, $M_{1}$ is said to be $(\varepsilon,\mu, \hl{\varrho})$-approximately stochastically alternatingly bisimilar to $M_{2}$, if $M_{1}$ is $(\varepsilon,\mu, \hl{\varrho})$-approximately stochastically alternatingly simulated by $M_{2}$, and $M_{2}$ is $(\varepsilon,\mu, \hl{\varrho})$-approximately stochastically alternatingly simulated by $M_{1}$. The stochastic alternating simulation and stochastic alternating bisimulation relations are denoted, respectively, by $M_{1} \widehat{\preccurlyeq}_{\hl{\varrho}}^{\varepsilon, \mu} M_{2}$ and $M_{1} \widehat{\approx}_{\hl{\varrho}}^{\varepsilon, \mu} M_{2}$.
\end{definition}
Similarly, as with the stochastic bisimulation relation, we consider the different variants of Definition \ref{DefinitionAlternatingBisimulation} depending on the type of system models (gMDP or Transition system).
\hl{\begin{scase} (Stochastic Alternating Bisimulation between gMDP $M_1$and Transition system $S_2$) : We can represent a transition system $S_2$ using gMDP $M(S_2)$, as in Remark \ref{SS: rem/gMDP-TS}, same as in Special Case \ref{SS: scase/bisim1}. Once both the systems are modeled as gMDP, we follow conditions (ii) relaxed as in Special Case \ref{SS: scase/bisim1}, (i), (iii), as in Definition \ref{DefinitionAlternatingBisimulation} to define the stochastic alternating bisimulation as $M_{1} \widehat\approx^{\varepsilon, \mu}_{\hl{\varrho}} M(S_{2}).$ \label{SS: scase/abisim1}
\end{scase}}
\hl{\begin{scase} (Alternating Bisimulation between Transition systems $S_1, S_2$) :
 As Transition systems do not require equivalence relations between probability measures, we further simplify condition (iii) in Definition \ref{DefinitionAlternatingBisimulation} as below.
 \begin{itemize}
     \item[(iii)] For all $(x_1,x_2)\in\mathcal{R}$ and all
$({u}_2^{ext},{u}_2^{int})\in U_{S_2}(x_2)$,
there exists
$(u_1^{ext},u_1^{int})\in U_{S_1}(x_1)$
with $\max\!\left(
  \mathbf{d}_{U^{ext}}(u_1^{ext},u_2^{ext}),
  \mathbf{d}_{U^{int}}(u_1^{int},u_2^{int})
\right)
\le \mu$
such that for all
$x_1'\in\Delta_1(x_1,u_1^{ext},u_1^{int})$, $\exists x_{2}^{\prime} \in\Delta_{2}\left(x_{2}, {u}_{2}^{ext}, {u}_{2}^{int}\right)$ such that $\left(x_{1}^{\prime}, x_{2}^{\prime}\right) \in \mathcal{R}$.
 \end{itemize}
 Together with conditions (i) and (ii), we define the alternating bisimulation relation (not stochastic) between Transition systems denoted as $S_{1} \widehat\approx^{\varepsilon, \mu} S_{2}$. \label{SS: scase/abisim2}
\end{scase}}
Contrary to the concepts of approximate (bi)-simulation relation introduced in \cite{tabuada2009verification} and \cite{girard2007approximation}, the concept of approximate (bi)-simulation relation introduced in Definition \ref{DefinitionBisimulation} is more relaxed since it allows a mismatch on the choice of inputs for the transition systems, and requires successor matching in probability. In particular, when $\mu=0, \hl{\hl{\varrho} = 0},$ and $\mathbf{d}_{U^{int}}$ is a metric, the relation proposed in Definition \ref{DefinitionBisimulation} reduces to the notion of approximate bisimulation introduced in \cite{girard2007approximation}, and when $\mu=\infty, \hl{\hl{\varrho} = 0}$, it covers the approximate bisimulation relation given in \cite{tabuada2009verification}. Furthermore, the concept of approximate alternating bisimulation of Definition \ref{DefinitionAlternatingBisimulation} recovers the one in \cite{pola2009symbolic} by taking $\mu=\infty, \hl{\hl{\varrho} = 0}$.

To gather all the necessary ingredients to present our main results, the following two propositions are required. These properties primarily demonstrate the ordering and transitivity properties of the introduced relationships. 

\begin{prop}
    \label{prop1}
    
    \hl{Given three gMDPs $M_{1}, M_{2}$ and $M_{3}$ with shared output space. For any $\mu, \mu^{\prime} \in \R_0^{+}$, $\varepsilon, \varepsilon^{\prime} \in \R_0^{+}$, and $\hl{\varrho}, \hl{\varrho}^{\prime} \in [0,1]$, the following statements hold:
    \begin{itemize}
        \item if $M_{1} \preccurlyeq^{\varepsilon, \mu}_{\hl{\varrho}}  M_{2}$, $M_{2} \preccurlyeq^{\varepsilon', \mu'}_{\hl{\varrho}'} M_{3}$, then $M_{1} \preccurlyeq^{\varepsilon+\varepsilon^{\prime}, \mu+\mu^{\prime}}_{\hl{\varrho} + \hl{\varrho}'}  M_{3}$;
        \item if $M_{1} \approx^{\varepsilon, \mu}_{\hl{\varrho}} M_{2}$, $M_{2} \approx^{\varepsilon', \mu'}_{\hl{\varrho}'} M_{3}$, then $M_{1} \approx^{\varepsilon+\varepsilon^{\prime}, \mu+\mu^{\prime}}_{\hl{\varrho} + \hl{\varrho}'} M_{3}$;
        \item if $M_{1} \widehat{\preccurlyeq}^{\varepsilon, \mu}_{\hl{\varrho}}  M_{2}$, $M_{2} \widehat{\preccurlyeq}^{\varepsilon', \mu'}_{\hl{\varrho}'} M_{3}$, then $M_{1} \widehat{\preccurlyeq}^{\varepsilon+\varepsilon^{\prime}, \mu+\mu^{\prime}}_{\hl{\varrho} + \hl{\varrho}'}  M_{3}$;
        \item if $M_{1} \widehat{\approx}^{\varepsilon, \mu}_{\hl{\varrho}} M_{2}$, $M_{2} \widehat{\approx}^{\varepsilon', \mu'}_{\hl{\varrho}'} M_{3}$, then $M_{1} \widehat{\approx}^{\varepsilon+\varepsilon^{\prime}, \mu+\mu^{\prime}}_{\hl{\varrho} + \hl{\varrho}'} M_{3}$.
    \end{itemize}}
\end{prop}
\begin{pf}
Items 1 and 2 follow from \cite[Theorem 5, Corollary 1]{haesaert2017verification} with respect to $(\varepsilon,\hl{\varrho})$ for the transitivity property, taking $\mathcal R_{13}:=\{(x_1,x_3)\mid \exists x_2:(x_1,x_2)\in\mathcal R_{12},\,(x_2,x_3)\in\mathcal R_{23}\}$ and composing the coupling kernels $\mathbb W_{\mathbb T_{12}}$ and $\mathbb W_{\mathbb T_{23}}$ through $X_2$. For the input deviation, let $(x_1,x_3)\in\mathcal R_{13}$ with witness $x_2$ and $(u_1^{ext},u_1^{int})\in U_{M_1}(x_1)$. Condition (iii) for $\mathcal R_{12}$ yields $u_2\in U_{M_2}(x_2)$ within distance $\mu$, and condition (iii) for $\mathcal R_{23}$ then yields $u_3\in U_{M_3}(x_3)$ within distance $\mu'$ of $u_2$. By the triangle inequality on $\mathbf d_{U^{ext}}$ and $\mathbf d_{U^{int}}$, both distances between $u_1$ and $u_3$ are at most $\mu+\mu'$. Item 2 follows by applying item 1 to $M_1\preccurlyeq M_2\preccurlyeq M_3$ and $M_3\preccurlyeq M_2\preccurlyeq M_1$.
Items 3 and 4 follow analogously, with the order of input selection reversed: given $u_3\in U_{M_3}(x_3)$, select $u_2$ through $\mathcal R_{23}$ and then $u_1$ through $\mathcal R_{12}$, and compose the coupling kernels in the same order. \hfill $\qed$
\end{pf}
Using Remark \ref{SS: rem/gMDP-TS}, the same transitivity property holds for transition systems, as they can be modeled as gMDP using \eqref{SS: eqn/rem1}. All four items can be stated with $\hl{\varrho} = 0$ for the three pseudometric transition systems $S_1, S_2, S_3$ as below. We remove the coupling measure $\hl{\varrho}$ as the relations are unaffected by the noise process.
\begin{itemize}
        \item if $S_{1} \preccurlyeq^{\varepsilon, \mu}  S_{2}$, $S_{2} \preccurlyeq^{\varepsilon', \mu'} S_{3}$, then $S_{1} \preccurlyeq^{\varepsilon+\varepsilon^{\prime}, \mu+\mu^{\prime}}  S_{3}$;
        \item if $S_{1} \approx^{\varepsilon, \mu} S_{2}$, $S_{2} \approx^{\varepsilon', \mu'} S_{3}$, then $S_{1} \approx^{\varepsilon+\varepsilon^{\prime}, \mu+\mu^{\prime}} S_{3}$;
        \item if $S_{1} \widehat{\preccurlyeq}^{\varepsilon, \mu}  S_{2}$, $S_{2} \widehat{\preccurlyeq}^{\varepsilon', \mu'} S_{3}$, then $S_{1} \widehat{\preccurlyeq}^{\varepsilon+\varepsilon^{\prime}, \mu+\mu^{\prime}}  S_{3}$;
        \item if $S_{1} \widehat{\approx}^{\varepsilon, \mu} S_{2}$, $S_{2} \widehat{\approx}^{\varepsilon', \mu'} S_{3}$, then $S_{1} \widehat{\approx}^{\varepsilon+\varepsilon^{\prime}, \mu+\mu^{\prime}} S_{3}$.
    \end{itemize}
\begin{prop}
    \label{prop2}
    \hl{Given two gMDPs $M_{1}$ and $M_{2}$. For any $\mu^{\prime} \geq \mu \geqslant 0$, $\varepsilon^{\prime} \geqslant \varepsilon \geqslant 0$, and $\hl{\varrho}, \hl{\varrho}' \in [0,1], \hl{\varrho}' \geq \hl{\varrho}$, the following statements hold:
    \begin{itemize}
        \item if $M_{1} \preccurlyeq^{\varepsilon, \mu}_{\hl{\varrho}} M_{2}$ then $M_{1} \preccurlyeq^{\varepsilon^{\prime}, \mu^{\prime}}_{\hl{\varrho}'} M_{2}$;
        \item if $M_{1} \widehat{\preccurlyeq}_{\hl{\varrho}}^{\varepsilon, \mu} M_{2}$ then $M_{1} \widehat{\preccurlyeq}^{\varepsilon^{\prime}, \mu^{\prime}}_{\hl{\varrho}'} M_{2}$.
    \end{itemize}}
\end{prop}
\begin{pf}
  \hl{For proof, the same relation $\mathcal{R}$ holds, since conditions (i)--(iii) of Definition~\ref{DefinitionBisimulation} are monotone in $\varepsilon,\mu,\varrho$ ($\varepsilon \le \varepsilon'$, $\mu \le \mu'$, $1-\varrho \ge 1-\varrho'$); the second statement follows identically.} \hfill $\qed$
\end{pf}
Similarly, using Remark \ref{SS: rem/gMDP-TS}, we can rewrite the above proposition for noise-free transition systems:
\begin{itemize}
        \item if $S_{1} \preccurlyeq^{\varepsilon, \mu} S_{2}$ then $S_{1} \preccurlyeq^{\varepsilon^{\prime}, \mu^{\prime}} S_{2}$;
        \item if $S_{1} \widehat{\preccurlyeq}^{\varepsilon, \mu} S_{2}$ then $S_{1} \widehat{\preccurlyeq}^{\varepsilon^{\prime}, \mu^{\prime}} S_{2}$.
\end{itemize}
    
In the following subsection, we define the interconnections among DT-SCS modeled either as gMDPs, in order to state our main problem of preserving the (alternating) bisimulation relation between an interconnected stochastic system and its noise-free counterpart.

\subsection{Interconnections among DT-SCS and   $\Xi$-compatibility composition}
	
An interconnection of discrete-time stochastic transition systems as defined in \eqref{dtss-trans} consists of a collection of $ N (\in \mathbb{N}) $ gMDPs $\left\{M_{i}(\Sigma)\right\}_{i \in I}$, indexed by the set of vertices $ I = \left\lbrace 1,\dots, N \right\rbrace $. The interconnection is defined by a binary connectivity relation $ \mathcal{I} \subseteq I \times I $, where each vertex $ i \in I $ is labeled with the gMDP $ M_{i}(\Sigma) $. For each $ i \in I $, we define $ \mathcal{N}(i) = \left\lbrace j \in I \mid (j, i) \in \mathcal{I} \right\rbrace $ as the set of neighboring components from which the incoming edges originate. The $ i^{th} $ subsystem is described by $M_{i}(\Sigma) =\left(X_{i}, X_{i}^{0}, U_{i}^{{ext}}, U_{i}^{{int}}, \mathbb{T}_{i}, Y_{i}, H_{i}\right)$, where $H_{i}$ is the identity map $ H_{i}(x)=x $. 

Consider the collection $\left\{M_{i}(\Sigma)\right\}_{i \in I}$. For each $i \in I$, let the Borel measurable space $(X_i, \mathcal{B}(X_i))$ be induced by the noise sequence $\ww_i(k): \Omega \rightarrow \mathcal{V}_{\ww}$, a $p$-dimensional $\filtration$-adapted discrete-time stochastic process as defined in Definition \ref{Def:DiscreteTimeSCS}. We assume that for all $i \in I$, the product of neighbor output sets $\prod_{j \in \mathcal{N}(i)} Y_{j}$ and the internal input set $U_{i}^{{int}}$ are subsets of the same pseudometric space equipped with a pseudometric $\mathbf{d}_{U_{i}^{int}}$ such that

\begin{equation}\label{Eq9}
    \mathbf{d}_{U_{i}^{{int}}}\left(u_{i}^{1, {int}}, u_{i}^{2, {int}}\right) = \max _{j \in \mathcal{N}(i)}\left\{\mathbf{d}_{Y_{j}}\left(y_{j}^{1}, y_{j}^{2}\right)\right\},
\end{equation}
where $u_{i}^{l, {int}}$ corresponds to the tuple of neighbor outputs $(y_{j}^{l})_{j \in \mathcal{N}(i)}$ for $l \in \{1,2\}$.

\begin{definition}\label{DefComposition}
    Let $\Xi:=\left(\mu_{1}, \ldots, \mu_{N}\right) \in\left(\mathbb{R}_{0}^{+}\right)^{N}$. We say that the collection of independent local gMDPs $\left\{M_{i}(\Sigma)\right\}_{i \in I}$ is compatible for the $\Xi$-approximate composition with respect to $\mathcal{I}$ if for each $i \in I$ and for every tuple of neighbor outputs, denoted as random variable $\y_{\mathcal{N}(i)} \in \prod_{j \in \mathcal{N}(i)} Y_{j}$, there exists a \hl{deterministic} internal input $u_{i}^{int} \in U_{i}^{int}$ such that:\hl{
    \begin{equation}\label{inp_ms}
        \mathbb E[\mathbf{d}_{U_{i}^{{int}}}\left(u_{i}^{{int}}, \y_{\mathcal{N}(i)}\right)^q]^{\frac 1q} \leqslant \mu_{i},
    \end{equation}}
    where $U_{i}^{{int}}, Y_i, i \in I$, share the same pseudometric space.
   \hl{ We denote the $\Xi$-approximate composed system by $\left\langle M_{i}(\Sigma)\right\rangle_{i \in I}^{\Xi, \mathcal{I}}$. It is defined by the tuple $\left(X, X^{0}, U^{ext}, \mathbb{T}_{\Xi}, Y, \right. $ $\left. H\right)$, where:
    \begin{itemize}
        \item $X=\prod_{i \in I} X_{i}$; $X^{0}=\prod_{i \in I} X_{i}^{0}$; $U^{ext}=\prod_{i \in I} U_{i}^{ext}$; $Y=\prod_{i \in I} Y_{i}$;
        \item $H(x)=\left(H_{1}\left(x_{1}\right), \ldots, H_{N}\left(x_{N}\right)\right)$;
        \item The stochastic kernel $\mathbb{T}_{\Xi}(\cdot \mid x, u^{ext})$ of the interconnected gMDP is well-defined if and only if, for every subsystem $i \in I$, there exists an internal input $u_i^{int} \in U_i^{int}$ such that
\begin{enumerate}
    \item the measurable selection $u_i^{int}$ exists such that input mismatch \eqref{inp_ms} holds; 
    \item $(u_i^{ext}, u_i^{int}) \in U_{M_i}(x_i)$;
    \item the local successor state is generated according to the stochastic kernel
$x_i' \sim \mathbb{T}_i(\cdot \mid x_i, u_i^{ext}, u_i^{int}).$
\end{enumerate}
Consequently, the stochastic kernel of the interconnected gMDP is given by the product kernel 
where, for every measurable set $A=\prod_{i\in I}A_i$ with
$A_i\in\mathcal{B}(X_i)$,

\begin{equation}\label{prodkernel}
    \mathbb{T}_{\Xi}(A \mid x, u^{ext}) = \prod_{i \in I} \mathbb{T}_{i}(A_i \mid x_i, u_i^{ext}, u_i^{int}).
\end{equation}
This construction assumes that the subsystem transitions are conditionally independent given the current state and the external and internal inputs.
\end{itemize}}
    Here, $x = (x_1, \dots, x_N)$ and $x' = (x_1', \dots, x_N')$. The stochastic kernel $\mathbb{T}_i$ is based on the solution processes $\x_i$ of the subsystems, as defined in \eqref{dt_sys}.
\end{definition}
\begin{rem}
    We considered an interconnected system with discrete-time stochastic subsystems, modeled as gMDPs $\left\{M_{i}(\Sigma)\right\}_{i \in I}$ for defining $\Xi$-approximate composition. However, Definition \ref{DefComposition} applies analogously to:
    \begin{itemize}
        \item The interconnected system $\left\{S_{i}(\bar{\Sigma})\right\}_{i \in I}$, composed of noise-free counterparts modeled as transition systems. The transition relation is defined by the deterministic map ${\Delta}_{i}$ rather than by the probability measure $\mathbb{T}_i$, where for any $x_i \in X_i$, the successor is as per \eqref{nfe}.
        \item We denote the $\Xi$-approximately composed transition system by $\left\langle S_{i}(\bar \Sigma)\right\rangle_{i \in I}^{\Xi, \mathcal{I}}$. It is defined by the tuple $\left(X, X^{0}, U^{ext}, \Delta_{\Xi}, Y, \right. $ $\left. H\right)$, where a transition $x^{\prime} \in \Delta_{\Xi}\left(x, {u}^{ext}\right)$ exists if and only if for all $i \in I$, there exists an internal input $u_{i}^{int} \in U_{i}^{int}$ such that:
        \begin{enumerate}
            \item $\mathbf{d}_{U_{i}^{int}}\left(u_{i}^{int}, y_{\mathcal{N}(i)}\right) \leqslant \mu_{i}$, where $y_{\mathcal{N}(i)} = (H_j(x_j))_{j \in \mathcal{N}(i)}$;
            \item $(u_{i}^{ext}, u_{i}^{int}) \in U_{S_{i}}(x_{i})$;
            \item $x_{i}^{\prime} \in \Delta_{i}\left(x_{i}, u_{i}^{ext}, u_{i}^{int}\right)$.
        \end{enumerate}
    \end{itemize}
\end{rem}

For an illustration of the notion of approximate composition, we refer the reader to the following example with transition systems.

\begin{example}
    An example of the interconnection of three transition systems $S_1, S_2$ and $S_3$ is presented in Fig. \ref{fig:interconnectedsys}. The connectivity relation is defined by $\mathcal{I}= \{(1,2), (2,3), (3,1)\}$, implying neighbor sets $\mathcal{N}(1)=\{3\}, \mathcal{N}(2)=\{1\}, \mathcal{N}(3)=\{2\}$. Fig. \ref{fig:interconnectedsys} illustrates the case of a composition with an approximate composition parameter vector $\Xi=(\mu_1,\mu_2,\mu_3)$.

    \begin{figure}[h]
        \centering
            \centering
            \includegraphics[scale = 0.24]{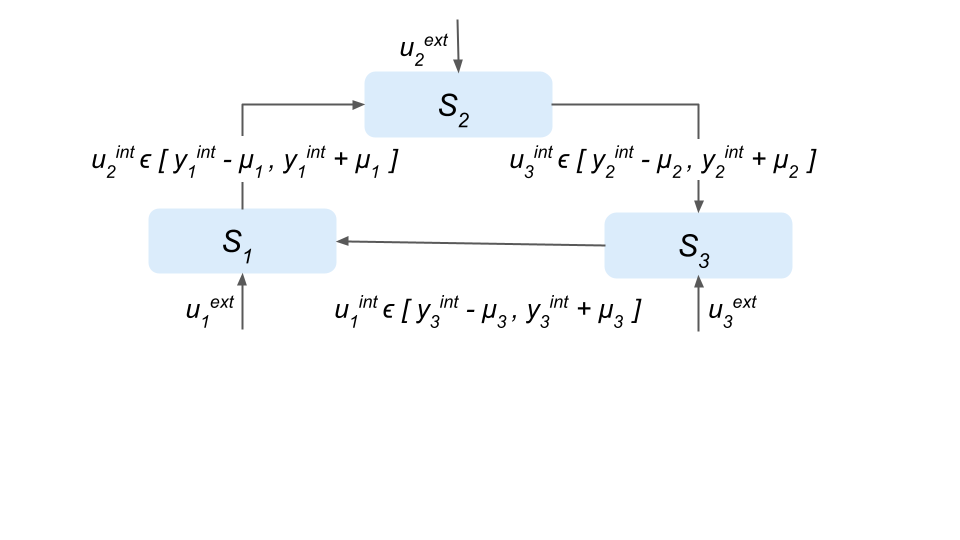}
            \vspace{-50pt}
            \caption{Example of composition with approximate parameter $\Xi = [\mu_1,\ \mu_2,\ \mu_3]$.}
        \label{fig:interconnectedsys}
    \end{figure}
\end{example}

We equip the composed output space with the metric such that for any \hl{random variable} $\y^{(j)} \in Y$ with $\y^{(j)}=(\y_{1}^{(j)}, \ldots, \y_{N}^{(j)})$ for $j \in\{1,2\}$, \hl{for product distribution,
\begin{align}\label{Eq10}
	\mathbf{d}\left(\y^{(1)}, \y^{(2)}\right)=\max _{i \in I}\left\{\mathbb E\big[\mathbf{d}_{Y_{i}}\big(\y_{i}^{(1)}, \y_{i}^{(2)}\big)\big]\right\}.
\end{align}
When the output $y^{(j)}$ represents points for a noise-free system, the metric resolves into a point-wise maximum metric.}
Similarly, we equip the composed external and internal input spaces with the pseudometric, shown for external inputs:
For any $u^{j,ext} \in U^{ext}$ with $u^{j,ext}=(u_{1}^{j,ext}, \ldots, u_{N}^{j,ext})$ for $j \in\{1,2\}$,

\hl{\begin{equation}\label{Eq11}
    \mathbf{d}_{U^{ext}}(u^{1,ext}, u^{2,ext}) = \max _{i \in I}\left\{\mathbb E\big[\mathbf{d}_{U_{i}^{ext}}(u_{i}^{1,ext}, u_{i}^{2,ext})\big]\right\}.
\end{equation}}

Using the interconnected gMDPs for stochastic system as well as interconnected transition system for its noise-free counterpart as defined in this section, we state our main problem as follows.

\subsection{Problem formulation}

This work addresses the problem of establishing a formal behavioral equivalence, in the sense of stochastic bisimulation or alternating bisimulation, between a large-scale interconnected discrete-time stochastic control system and a corresponding noise-free interconnected abstraction. The interconnected system of noise-free abstractions is denoted as $\{\hat{S}_{i}(\bar{\Sigma})\}_{i \in I}$, where each subsystem is defined as $\hat{S}_{i}=(\hat{X}_{i}, \hat{X}_{i}^{0}, \hat{U}_{i}^{ext}, \hat{U}_{i}^{int}, \hat{\Delta}_{i}, \hat{Y}_{i}, \hat{H}_{i})$. Existing approaches typically rely on monolithic abstraction techniques, which scale poorly and do not exploit the modularity of the system. We therefore formulate the problem as follows.

\begin{problem}
    Given a large-scale $\Xi$-compatible interconnected discrete-time stochastic control system $\left\langle M_{i}(\Sigma)\right\rangle_{i \in I}^{\Xi, \mathcal{I}}$ (as per Definition \ref{DefComposition}), construct a scalable and computationally efficient deterministic abstraction $\langle \hat{S}_{i}(\bar{\Sigma})\rangle_{i \in I}^{\hat{\Xi}, \mathcal{I}}$ that preserves the stochastic (alternating) bisimulation relation (Definitions \ref{DefinitionBisimulation} and \ref{DefinitionAlternatingBisimulation}) with the concrete interconnected stochastic system, such that desired behavioral guarantees on the deterministic abstractions can be formally transferred to the global interconnected stochastic system.
\end{problem}

To address this, we propose a method of compositional abstraction that preserves the (alternating) bisimulation relation from the subsystem level to the global interconnected level using a bottom-up approach. We leverage the $q^{th}$-moment incremental input-to-state stability property to first establish a stochastic (alternating) bisimulation relation between each stochastic subsystem $M(\Sigma)$ and its noise-free counterpart $S(\bar{\Sigma})$, as detailed in the following section. Since the subsequent analysis focuses on individual subsystems, we omit the index $i$ for brevity.

The $\delta$-ISS-$M_q$ analysis provides the necessary bounds to quantify the deviation between the solution process of the stochastic system and the trajectory of its corresponding noise-free system, which we address in the following section.
	
\section{On the distance between the noisy and noise-free system trajectories}\label{sec5}

In this section, we aim to quantify the distance between the solutions of the DT-SCS $\Sigma$ and those of the corresponding noise-free control system $\bar{\Sigma}$. Although a conservative upper bound can be readily obtained under the $\delta$-ISS-$M_q$ property assumption, we establish a novel bound that yields a tighter estimate than that provided by the direct application of $\delta$-ISS-$M_q$.

As illustrated in Fig.~\ref{fig:distquant}, given any noisy solution process $\x$ of \eqref{dt_sys} and any corresponding noise-free trajectory $\bar{\x}$ of \eqref{nfe}, the $q^{\text{th}}$ moment of the maximum distance between any pair of noisy processes is denoted by $d_1$, and that between any pair of noise-free trajectories is denoted by $d_2$. These can be upper-bounded using \eqref{D-ISS-M-q} and \eqref{D-ISS}, respectively. The objective of this subsection is to derive a tighter (less conservative) upper bound on the distance between the noisy and noise-free solutions, namely on $\mathbb{E}[\|\x(k) - \bar{\x}(k)\|^q], \forall k \in \mathbb{N}_{0}$. This is particularly relevant for the first moment $q=1$, which is not typically handled in the existing literature using $\delta$-ISS Lyapunov functions (as highlighted in \cite{zamani2014symbolic}). We achieve this by constraining the magnitude of the noise affecting the system at a given time step via the following assumption.

\begin{figure}[h]
    \centering
    
    \includegraphics[scale=0.3]{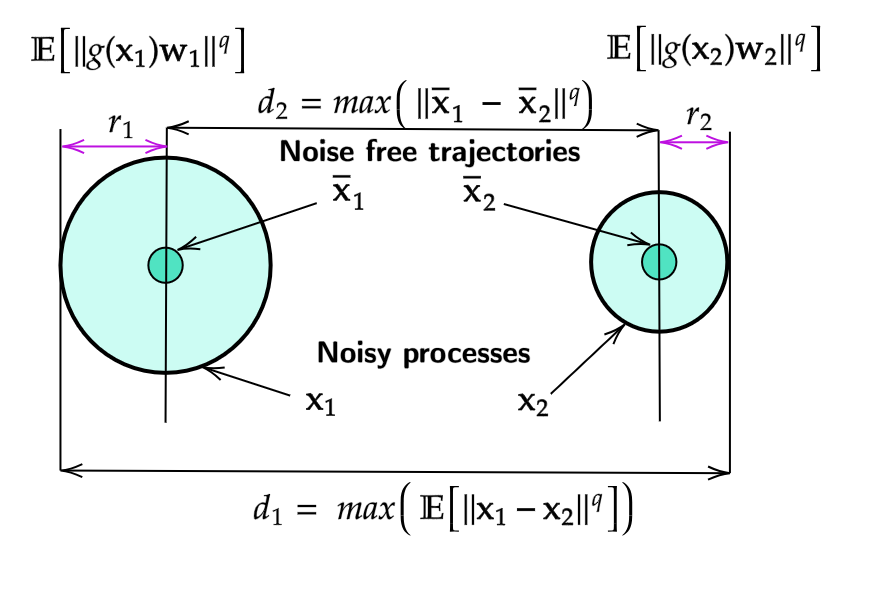}
    \caption{Geometric interpretation of the trajectory bounds. The noisy and noise-free solutions are represented by a horizontal cross-section of a cylindrical tube at time index $k$. The smaller dark green circle denotes the set of noise-free solutions ($\bar{\x}_1, \bar{\x}_2$), while the larger light green circle illustrates the noisy solutions ($\x_1,\x_2$), capturing the spread caused by the diffusion term $g(\x)\ww$. The radii $r_1$ and $r_2$ represent the noise effect, while $d_1$ and $d_2$ indicate the maximum deviations between the noisy solutions and the noise-free solutions, respectively.}
    \label{fig:distquant}
\end{figure}

\begin{assum}\label{Assum8}
    Consider the solution trajectories $\x_{x_1\varsigma_1}$, $\varsigma_1 := (\uu^{ext}_1, \uu^{int}_1)$, and $\x_{x_2\varsigma_2}$, $\varsigma_2 := (\uu^{ext}_2, \uu^{int}_2)$, of the DT-SCS $\Sigma$, starting from any initial states $x_1 \in \mathcal{X}$ and $x_2 \in \mathcal{X}$, respectively, such that $\|x_1 - x_2\| \leqslant \vartheta' < \vartheta$ (where $\vartheta, \vartheta' \in \mathbb{R}^+_0$).
    Assume the system is driven by the same internal inputs $\mathbf{u}_1^{int}=\mathbf{u}_2^{int}$, external inputs $\mathbf{u}_1^{ext}=\mathbf{u}_2^{ext}$, and different independent noise realization $\ww_k := \ww(k)$, $\ww'_k := \ww'(k), k \in \mathbb{N}_{0}$, and satisfies the $\delta$-ISS-$M_q$ stability property with functions $\beta, \psi$ as defined in Definition \ref{DefD-ISS-M-q}. The drift function $f(.)$ is affine in state with coefficient $L_{f}^x <1, L_f^x\in \mathbb{R}^{+}$.
    We assume that there exists a scalar $\lambda \in (0,1)$ such that the $q^{th}$-moment of the convex combination of the diffusion terms corresponding to the two noisy solutions is upper bounded by a function ${\bar{\alpha}}(k): \mathbb{N}_{0} \rightarrow \mathbb{R}_{0}^{+}$, satisfying:
    \begin{align}\label{mainkey}
        &\mathbb{E}\left[\|(1 - \lambda)g_{k}(\x_{x_1\varsigma_1})\ww_{k} + \lambda g_{k}(\x_{x_2\varsigma_2})\ww'_{k}\|^q\right] \leqslant {\bar{\alpha}}(k+1), \nonumber \\
        &({\bar{\alpha}}(k+1) + (L_f^x)^q{\bar{S}_{\alpha}}(k)) < (4^{1-q} - \lambda) \beta(\|x_1-x_2\|^q, k+1) \nonumber \\
        &+ (\lambda - 1) (\beta(\|x_1-x_2\|, k+1))^{q} + (4^{-(q-1)} - \lambda)\psi,
    \end{align}
where
\begin{equation}\label{Sbar}
    \bar{S}_{\alpha}(k) = \frac{1}{L_f^x}\,\frac{1}{(1-L_f^x)^{q-1}}\sum_{j=0}^{k-1}(L_f^x)^{k-j}\,\bar \alpha(j+1).
\end{equation}
\end{assum}
Note that $\bar{S}_{\alpha}(0) = 0$ and the Assumption \ref{Assum8} at coincident initial states reduces to a single noise-moment inequality, satisfied since $\psi \in \mathbb R^{+}$ (due to different independent noise) carries the amplification from $\|A\|<1$; for $q=1$ it is independent of $\|x_1-x_2\|$, making it the only case to verify. The detailed feasibility analysis of the assumption is illustrated for linear or non-linear DT-SCS by the verification procedure in Appendix~\ref{sec: example}. 
Assumption~\ref{Assum8} needs to hold, together with a closed-form
range of admissible $\lambda$, which is less conservative than assuming a bound on the diffusion function $g(x)$ as in Appendix \ref{ISSnl}.
Under this setup, the following results provide a tight upper bound on the distance between noisy and noise-free trajectories.
\begin{figure}[t]
    \centering
    \includegraphics[scale=0.3]{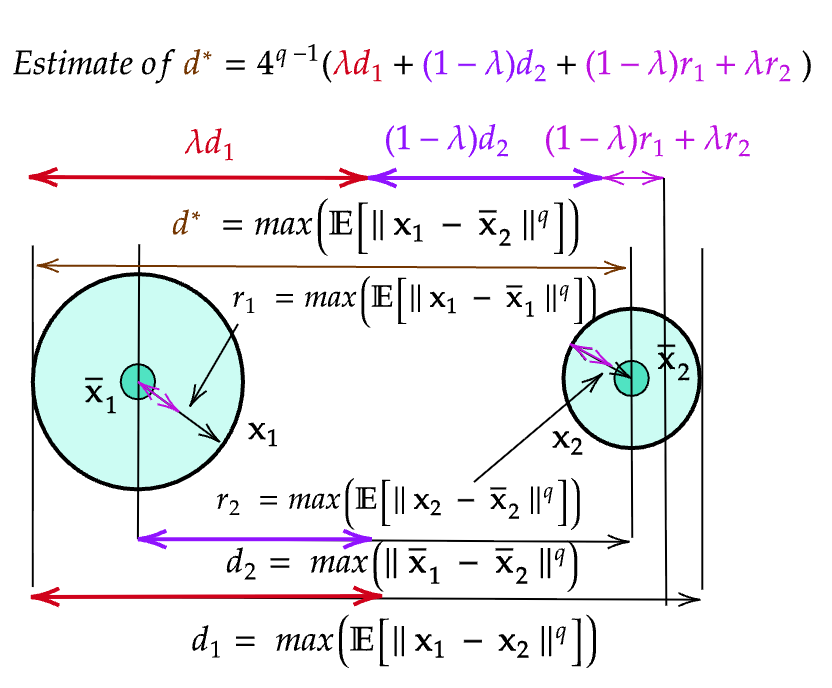}
    \caption{The maximum deviation between the noisy solution process $\x_1$ and noise-free trajectory $\bar{\x}_2$, denoted as $d^{*}$, is estimated via \eqref{Eq6}, utilizing the upper bounds provided by the $\delta$-ISS-$M_{q}$ ($d_1$) and $\delta$-ISS ($d_2$) properties.}
    \label{fig:nfe_bound}
\end{figure}
\begin{lem}\label{lem10}
    Consider a $\delta$-ISS-$M_q$ stable DT-SCS as in \eqref{dt_sys} and its associated $\delta$-ISS stable noise-free system as in \eqref{nfe}. The drift function $f(.)$ is affine in state argument $x$ with coefficient $L_f^x <1, L_f^x\in \mathbb{R}^{+}$. For any $q \in \N$ and initial states $x_1, x_2$ such that $\|x_1 - x_2\| \leqslant \vartheta' < \vartheta$ (where $\vartheta, \vartheta' \in \mathbb{R}^+_0$), with the same input signals $\mathbf{u}_1^{ext} = \mathbf{u}_2^{ext}$ and $\mathbf{u}_1^{int} = \mathbf{u}_2^{int}$, with different independent noise realization $\ww_k := \ww(k)$, $\ww'_k := \ww'(k), k \in \mathbb{N}_0$ and functions $\beta, \psi$ as in Definition \ref{DefD-ISS-M-q}, if there exists a parameter $\lambda \in (0,1)$ and a function ${\bar{\alpha}}(\cdot)$ such that Assumption \ref{Assum8} holds, then for the noisy trajectory $\x_{x_1\varsigma_1}$, $\varsigma_1 := (\uu^{ext}_1, \uu^{int}_1)$, and the noise-free trajectory $\bar \x_{x_2\varsigma_2}$, $\varsigma_2 := (\uu^{ext}_2, \uu^{int}_2)$, we have:
    \begin{align}\label{Eq6}
        \mathbb{E}\left[\|\x_{x_1\varsigma_1}(k)-\bar \x_{x_2\varsigma_2}(k)\|^{q}\right] \leqslant \mathcal{B}_{\vartheta,k}, 
    \end{align}
    where the upper bound $\mathcal{B}_{\vartheta}, \mathcal{B}_{\vartheta,0} = \vartheta',$ is defined by:
   \begin{equation}\label{Def_B_eps}
    \mathcal{B}_\vartheta = \sup_{k\in\mathbb{N}_0}\ \sup_{\substack{x_1,x_2\in X\\ \|x_1-x_2\|\le\vartheta'}}\ \sup_{\substack{(u^{ext},u^{int})\\ \text{admissible}}} \mathcal{B}_{\vartheta,k},
\end{equation}
    where
\begin{align}
    \mathcal{B}_{\vartheta, k} &= (1 - \lambda)\, 4^{q-1}\big(\beta(\|x_{1}-x_{2}\|, k)\big)^q + \lambda\, 4^{q-1}\big(\beta(\|x_{1}-x_{2}\|^q, k) + \psi\big) \nonumber\\
    &\quad + 4^{q-1}\big(\bar{\alpha}(k)+ (L_f^x)^{q}\bar{S}_{\alpha}(k-1)\big), \label{Bk}\\
    \bar{S}_{\alpha}(k-1) &= \frac{1}{L_f^x}\Big(\frac{1}{1-L_f^x}\Big)^{q-1}\sum_{j=0}^{k-2}(L_f^x)^{k-1-j}\,\bar \alpha(j+1), \qquad \forall k \ge 2. \nonumber
\end{align}
\end{lem}
\begin{pf}
    With the bound already satisfied with the initial deterministic points, the $q^{th}$-moment of the distance between the noisy trajectory $\x_1 \equiv \x_{x_1\varsigma_1}$ and the noise-free trajectory $\bar{\x}_2 \equiv \bar \x_{x_2\varsigma_2}$ can be expanded as follows.
    To this end, we construct a coupling of four trajectories on a common probability space, i.e., the noisy and noise-free trajectories $\x_1$ and $\bar \x_2$, starting at $x_1$ and $x_2$, and their auxiliary trajectories (analysis device) $\bar{\x}_1$ following \eqref{nfe}, and ${\x}_2$ following \eqref{dt_sys}, under identical inputs, starting at $x_1$ and $x_2$. 
    Consider the distance at $k^{th}$ step, where $k \in \N$:
    \begin{align*}
        \mathbb{E}&\left[\|\x_{1}(k)-\bar{\x}_{2}(k)\|^{q}\right] \nonumber \\
        &= \mathbb{E}\Big[\|f_{k-1}(\x_1, ., .) - f_{k-1}(\bar \x_2, ., .) + g_{k-1}(\x_1)\ww_{k-1}\|^{q}\Big] \nonumber\\ 
        &= \mathbb{E}\Bigg[\Big\| (1 - \lambda)\Big(f_{k-1}(\x_1,.,.) - f_{k-1}( \bar \x_2,.,.)\Big) \nonumber \\
        &+\lambda\Big(f_{k-1}(\x_1,..) - f_{k-1}( \bar \x_2,.,.) \nonumber \\
        &+ g_{k-1}(\x_{1})\ww_{k-1} - g_{k-1}( \x_{2})\ww'_{k-1}\Big) \nonumber \\
        &+ \Big((1 - \lambda) g_{k-1}(\x_{1})\ww_{k-1} + \lambda g_{k-1}( \x_{2})\ww'_{k-1}\Big)\Big\|^{q}\Bigg] \nonumber \\
        &= \mathbb{E}\Bigg[\Big\|\underbrace{(1 - \lambda)(\bar{\x}_{1}(k) - \bar{\x}_{2}(k))}_{d_2} \nonumber \\
        &\quad + (1 - \lambda)(f_{k-1}(\x_1,.,.) - f_{k-1}(\bar \x_1,.,.)) \nonumber \\
        &\quad + \underbrace{\lambda(\x_{1}(k) - \x_{2}(k))}_{d_1} + \lambda(f_{k-1}(\x_2,.,.) - f_{k-1}(\bar \x_2,.,.)) \nonumber \\
        &\quad + \underbrace{(1 - \lambda)g_{k-1}(\x_{1})\ww_{k-1} + \lambda g_{k-1}(\x_{2})\ww'_{k-1}}_{\text{Noise combination}} \Big\|^{q}\Bigg].
    \end{align*}
    Using the linearity of the expectation operator and the generalized Minkowski inequality, we simplify as
    \begin{align*}
        \mathbb{E}&[\|\x_{1}(k) - \bar{\x}_{2}(k)\|^{q}] \\
        &\leqslant (1 - \lambda)4^{q-1}\|\bar{\x}_{1}(k) - \bar{\x}_{2}(k)\|^{q} \\
        &+ \lambda 4^{q-1} \mathbb{E}[\|\x_{1}(k) - \x_{2}(k)\|^{q}] \\
        &+ 4^{q-1} \big(\Big(\tfrac{L_f^x}{1-L_f^x}\Big)^{q-1}\sum_{j=0}^{k-2}(L_f^x)^{(k-1-j)}\bar \alpha(j+1) \\
        &+ \mathbb{E}[\|(1 - \lambda)g_{k-1}(\x_{1})\ww_{k-1} 
        + \lambda g_{k-1}(\x_{2})\ww'_{k-1}\|^{q}]\big).
    \end{align*}
    Applying the $\delta$-ISS and $\delta$-ISS-$M_{q}$ properties to the first two terms, and using Assumption \ref{Assum8} to bound the third term with ${\bar{\alpha}}(k)$, we obtain
        
\begin{align*}
    \mathbb{E}[\|\x_{1}(k)-\bar{\x}_{2}(k)\|^q] &\leqslant (1 - \lambda)\, 4^{q-1}\big(\beta(\|x_{1}-x_{2}\|,k)\big)^q + \lambda\, 4^{q-1}\big(\beta(\|x_{1}-x_{2}\|^q,k) + \psi\big)\\
    &\quad + 4^{q-1}\big(\bar{\alpha}(k)+ (L_f^x)^q\bar{S}_{\alpha}(k-1)\big).
\end{align*}
    This completes the proof. \hfill $\qed$
\end{pf}

\begin{rem}\label{otherbound}
    If Assumption \ref{Assum8} does not hold, the distance between the noisy solution process $\x_{x_1\varsigma_1}$ and the noise-free trajectory $\bar \x_{x_2\varsigma_2}$ can be upper-bounded using the standard $\delta$-ISS-$M_q$ property as:
    $ \mathcal{B}_{\vartheta, k}^{(2)} = \beta(\|x_{1}-x_{2}\|^q, k) + \psi. $
\end{rem}
\begin{rem}\label{rem11} 
    Given the upper bound $\mathcal{B}_{\vartheta, k}^{(1)} = \mathcal{B}_{\vartheta, k}$ as  in \eqref{Def_B_eps} on the reachable set between any noisy and noise-free trajectories, and the $\delta$-ISS-$M_{q}$ upper bound $\mathcal{B}_{\vartheta, k}^{(2)}$, the following relation holds:
    $ \mathcal{B}_{\vartheta, k}^{(1)} - \mathcal{B}_{\vartheta, k}^{(2)} < 0, $ explicitly under Assumption \ref{Assum8}.
\end{rem}
Rewriting and subtracting the two upper bounds, we obtain:
    \begin{align}\label{tight}
    \mathcal{B}_{\vartheta, k}^{(1)} - \mathcal{B}_{\vartheta, k}^{(2)}
    &= (4^{q-1} - \bar{\lambda})\big(\beta(\|x_{1}-x_{2}\|,k)\big)^q + \bar{\lambda}\big(\beta(\|x_{1}-x_{2}\|^q,k) + \psi\big) \nonumber\\
    &\quad + 4^{q-1}\big({\bar{\alpha}}(k) + (L_f^x)^q{\bar{S}_{\alpha}}(k-1)\big) - \big(\beta(\|x_{1} - x_{2}\|^q,k) + \psi\big) \nonumber\\
    &= (\bar{\lambda} -1)\big(\beta(\|x_{1}-x_{2}\|^q,k) + \psi\big) + (4^{q - 1} - \bar{\lambda})\big(\beta(\|x_{1}-x_{2}\|,k)\big)^q \nonumber\\
    &\quad + 4^{q-1}\big({\bar{\alpha}} (k)+(L_f^x)^q{\bar{S}_{\alpha}}(k-1)\big),
\end{align}
    where $\bar{\lambda} = \lambda 4^{q-1}$. Since \eqref{mainkey} holds for $q \in \N $, we have 
    \begin{align*}
     & {({\bar{\alpha}}(k) + (L_f^x)^q{\bar{S}_{\alpha}}(k-1))} < (4^{-(q-1)} - \lambda) \beta(\|x_1 - x_2\|^q, k)\nonumber \\
     & + (\lambda - 1) (\beta(\|x_1 - x_2\|, k))^{q} + (4^{-(q-1)} - \lambda)\psi.
    \end{align*}
    Multiplying the above inequality with $4^{q-1}$, we obtain the same terms as in RHS of \eqref{tight}, which justify that $\mathcal{B}_{\vartheta, k}^{(1)} - \mathcal{B}_{\vartheta, k}^{(2)} < 0$. Therefore, $\mathcal{B}_{\vartheta, k}^{(1)}$ is a tighter upper bound as compared to the $\delta$-ISS-$M_{q}$ upper bound $\mathcal{B}_{\vartheta, k}^{(2)}$.
\begin{rem}
    In \eqref{tight}, since $0 < {\lambda} < 4^{1-q}$, the term $(\bar{\lambda} - 1)(\beta(\|x_{1}-x_{2}\|^q,k) + \psi)$ dominates in the negative direction due to $\psi$, yielding $\mathcal{B}_{\vartheta, k}^{(1)} - \mathcal{B}_{\vartheta, k}^{(2)} < 0$, especially for small differences in initial conditions where $x_1 \approx x_2$.
\end{rem}

{We define a pseudometric on the output space $Y$ as $\mathbf{d}_{Y}(H(x),H(\bar{x})) = \|H(x) - H(\bar{x})\|.$} Using this metric, we establish the stochastic bisimulation relations below.

Using this metric, we establish the stochastic bisimulation relations below.

\begin{thm}\label{THM-NoiseFree}
    Consider a DT-SCS as \hl{gMDP} $\hl{M(\Sigma)} = \left(X,\, X^{0},\, {U}^{ext},\, {U}^{int},\,\hl{\mathbb{T}},\, Y,\, H\right)$ and its noise-free counterpart $S(\bar{\Sigma}) = (X, X^{0}, U^{ext}, U^{int}, \Delta, Y, H)$ as transition system, represented as the \hl{gMDP $M(S(\bar\Sigma))$ whose stochastic kernel ${\mathbb{T}}_{\Delta}(\cdot \vert \bar x, u^{ext}, u^{int})$ is the Dirac measure on $\Delta(\bar x, u^{ext}, u^{int}), \forall \bar{x} \in X$.} For a scalar $q \in \N$, assume:
    \begin{itemize}
        \item $M(\Sigma)$ is $\delta$-ISS-$M_q$;
        \item \hl{For any gMDP solution process $\x_{x_0\varsigma_1}$, and transition system solution $\bar \x_{\bar x_0\varsigma_1}$, such that $\|x_0 - \bar x_0\| \leqslant \vartheta' < \vartheta$}, the upper bound $\mathcal{B}_{\vartheta}$ (from Lemma \ref{lem10}) satisfies $\mathcal{B}_{\vartheta} \le \vartheta^q$ for all $k \in \mathbb{N}_{0}$.
    \end{itemize}
    \hl{Then, the relation $\mathcal{R}$ between random variable $\x$ and point $\bar x$ defined as} $\mathcal{R} = \{(\x, \bar{x}) \in L^q(\Omega;X) \times X \mid 
        \hl{\mathbb{E}[\mathbf{d}_{Y}(H(\x),H(\bar{x}))^q] ^{\frac1q}} \le \vartheta \},$ is an $(\vartheta, 0, 0)$-approximate \hl{stochastic} bisimulation relation between $M(\Sigma)$ and $M(S(\bar{\Sigma}))$, i.e., $M(\Sigma) \approx^{\vartheta,0}_{0} M(S(\bar\Sigma))$\hl{, as in Special Case \ref{SS: scase/bisim1} of stochastic bisimulation}.
\end{thm}
\begin{pf}
\hl{The starting pair is deterministic and the inputs are identical, which are exactly the hypotheses of Lemma~\ref{lem10}. Hence, $\mathbb E[\|\x_{x_0\varsigma_1}(k)-\bar \x_{\bar x_0\varsigma_1}(k)\|^q]\le\vartheta^q, \forall k\in\mathbb N_0$. Since $H=\mathrm{id}_X$ and $\mathbf d_Y(y,y')=\|y-y'\|$, we conclude $\big(\mathbb E[\mathbf d_Y(H(\x_{x_0\varsigma_1}(k)),H(\bar \x_{\bar x_0\varsigma_1}(k)))^q]\big)^{1/q}\le\vartheta$, so every pair $(\x_{x_0\varsigma_1}(k),\bar \x_{\bar x_0\varsigma_1}(k))$ generated by noisy and noise free solutions belongs to $\mathcal R$. 
Condition (i) holds with the hypothesis on initial conditions. Condition (ii) holds by the run bound, since $\mathcal R$ consists of
pairs generated by runs, and the relational metric on a (random variable, point) pair is the $q^{th}$-moment distance.
For (iii), let $(\x_{x_0\varsigma_1}(k),\bar \x_{\bar x_0\varsigma_1}(k))\in\mathcal R$ and choose the same input on both sides ($\mu=0$,
admissible by $ U_{M(\Sigma)}= U_{S(\bar\Sigma)}$). The successors $(\x_{x_0\varsigma_1}(k+1),\bar \x_{\bar x_0\varsigma_1}(k+1))$ belong to the
same run, hence to $\mathcal R$. Let $\mathbb W_{\mathbb T}$ be their coupling product kernel given
$\mathcal F_k$; its marginals are $\mathbb T(\cdot|x,u^{ext},u^{int})$ and $\mathbb T_\Delta(\cdot|\bar x,u^{ext},u^{int})$, and
$\mathbb W_{\mathbb T}(\mathcal R| x,\bar x,u^{ext},u^{int})=1$, so the coupling failure probability $\hl{\varrho} = 0$ and the symmetric relation holds in the converse direction.}
\hfill $\qed$
\end{pf}
\hl{Note that since the bound of Lemma \ref{lem10} is uniform in $k$ and over admissible inputs,
no restart from a pair in $\mathcal R$ is required}.
Similarly, the stochastic alternating bisimulation relation can be established, as discussed next.

\begin{thm}\label{THM-NoiseFree-Alt}
Consider a DT-SCS as \hl{gMDP} $\hl{M(\Sigma)} = \left(X,\, X^{0},\, {U}^{ext},\, {U}^{int},\,\hl{\mathbb{T}},\, Y,\, H\right)$ and its noise-free counterpart $S(\bar{\Sigma}) = (X, X^{0}, U^{ext}, U^{int}, \Delta, Y, H)$ as transition system, represented as the \hl{gMDP $M(S(\bar\Sigma))$} whose \hl{stochastic kernel $\mathbb{T}_{\Delta}(\cdot \vert \bar x, u^{ext}, u^{int})$ is the Dirac measure on $\Delta(\bar x, u^{ext}, u^{int}), \forall \bar{x} \in X$.} For a scalar $q \in \N$, assume:
    \begin{itemize}
        \item $M(\Sigma)$ is $\delta$-ISS-$M_q$;
        \item \hl{For any gMDP solution process run $\x_{x_0\varsigma_1}$, and transition system solution $\bar \x_{\bar x_0\varsigma_1}$, such that $\|x_0 - \bar x_0\| \leqslant \vartheta' < \vartheta$,} the upper bound (from Lemma \ref{lem10}) satisfies $\mathcal{B}_{\vartheta} \le \vartheta^q$ for all $k \in \mathbb{N}_{0}$.
    \end{itemize}
    \hl{Then, the relation $\mathcal{R}$ as defined in Theorem \ref{THM-NoiseFree}} is an $(\vartheta, 0, 0)$-approximate alternating \hl{stochastic} bisimulation relation between $M(\Sigma)$ and $M(S(\bar{\Sigma}))$, i.e., $M(\Sigma) \widehat \approx^{\vartheta,0}_{0} M(S(\bar\Sigma))$\hl{, as in Special Case \ref{SS: scase/abisim1} of stochastic alternating bisimulation}.	
\end{thm}
\begin{pf}
Conditions (i) and (ii) are identical to those in Definition \ref{DefinitionBisimulation} and were already verified in the proof of Theorem \ref{THM-NoiseFree}.
    For condition (iii) of Definition \ref{DefinitionAlternatingBisimulation}, note that the only difference from the stochastic bisimulation condition (iii) is the order of quantification over inputs. \hl{Since $U^{ext}_{M(\Sigma)}(x) = U^{ext}_{M(S(\bar\Sigma))}(\bar x)$ and $U^{int}_{M(\Sigma)}(x) = U^{int}_{M(S(\bar\Sigma))}(\bar x)$ under the condition on bounded noise support that the kernel of gMDP and transition system holds $\mathbb T(X|x,u^{ext},u^{int}) = 1$ for all states and inputs,} the same input can always be chosen on the other side, giving $\mu = 0$ regardless of which system's input is quantified first. Thus, for any $(u_2^{ext},u_2^{int}) \in U_{M(S(\bar\Sigma))}(\bar x)$, choosing $(u_1^{ext},u_1^{int}) = (u_2^{ext},u_2^{int}) \in U_{M(\Sigma)}(x)$ and repeating verbatim \hl{the coupling construction $\mathbb{W}_{\mathbb{T}}$ from the proof of Theorem \ref{THM-NoiseFree} (in both directions) yields
    $\mathbb{W}_{\mathbb{T}}(\mathcal{R} \vert u_1^{ext},u_1^{int},x,\bar x) = 1 \ge 1-\hl{\varrho}, \quad \hl{\varrho} = 0,$
    establishing $\mathbb{T}(\cdot\vert x,u_1^{ext},u_1^{int}) \, \bar{\mathcal{R}}_0 \, \mathbb{T}_{\Delta}(\cdot\vert \bar x,u_2^{ext},u_2^{int})$ in either direction.}
    Hence $M(\Sigma) \widehat\preccurlyeq^{\vartheta,0}_{0} M(S(\bar\Sigma))$ and $M(S(\bar\Sigma)) \widehat\preccurlyeq^{\vartheta,0}_{0} M(\Sigma)$, so $M(\Sigma) \widehat\approx^{\vartheta,0}_{0} M(S(\bar\Sigma))$.
    \hfill $\qed$
\end{pf}
\hl{
Having established, with Theorems \ref{THM-NoiseFree} and \ref{THM-NoiseFree-Alt}, that the stochastic gMDP $M(\Sigma)$ can be related to its noise-free counterpart $M(S(\bar\Sigma))$ through a $(\vartheta,0,0)$-approximate stochastic (alternating) bisimulation relation, we henceforth work directly with the transition system $S(\bar\Sigma)$ in place of $M(\Sigma)$, at the cost of the quantified error $\vartheta$. 
This translation allows us to shift, in the remainder of the paper, from the probabilistic setting to the purely relational one. We study interconnected transition systems and their $\Xi-$compositions, and establish compositionality results using the bisimulation and alternating bisimulation relations as in Special Cases \ref{SS: scase/bisim2} and \ref{SS: scase/abisim2}. In this setting, condition (iii) no longer requires coupling of stochastic kernels via $\bar{\mathcal{R}}_{\hl{\varrho}}$, but reduces to direct matching of successors under $\Delta_1,\Delta_2$, since transition systems do not carry probability measures.  
Moving forward, unless stated, we drop the coupling measure term $\hl{\varrho}$, and use $(\varepsilon, \mu)$-approximate (simulation) bisimulation relation as in Special Case \ref{SS: scase/bisim2}, denoted by ($S_1 \preccurlyeq^{\varepsilon',\mu'} S_2$) $S_1 \approx^{\varepsilon',\mu'} S_2$ as well as use $(\varepsilon, \mu)$-approximate (alternating simulation) alternating bisimulation as in Special Case \ref{SS: scase/abisim2}, denoted by ($S_1 \widehat \preccurlyeq^{\varepsilon',\mu'} S_2$) $S_1 \widehat \approx^{\varepsilon',\mu'} S_2$.
}

\section{Compositional bisimilar abstractions for interconnected transition systems}

In this section, we delve into the interconnections among transition systems, using the concept of $\Xi$-compatible composition as defined in Definition \ref{DefComposition}. The subsequent subsection explores the relationship between $\Xi$-compatible composition and a more conservative form defined by $\bar{\Xi}\geqslant \Xi$. We provide conditions that ensure the preservation of approximate (alternating) bisimulation relations between the $\Xi$-compatible composition of noise-free systems $S(\bar{\Sigma})$ and their $\bar{\Xi}$-compatible counterparts. Building on these findings, we establish conditions for preserving approximate (alternating) bisimulation relations, while going from noise-free subsystems $S_i(\bar{\Sigma})$ to the global interconnected system $\left\langle {S}_{i}(\bar{\Sigma})\right\rangle_{i \in I}^{\Xi, \mathcal{I}}$. Leveraging transitivity relations and Theorems \ref{THM-NoiseFree} and \ref{THM-NoiseFree-Alt}, we extend these results to encompass the global interconnected concrete stochastic system $\left\langle {M}_{i}({\Sigma})\right\rangle_{i \in I}^{\Xi, \mathcal{I}}$. Figure \ref{fig:compositional} illustrates these four key steps.

\begin{figure*}
    \centering
    \includegraphics[width=1\linewidth]{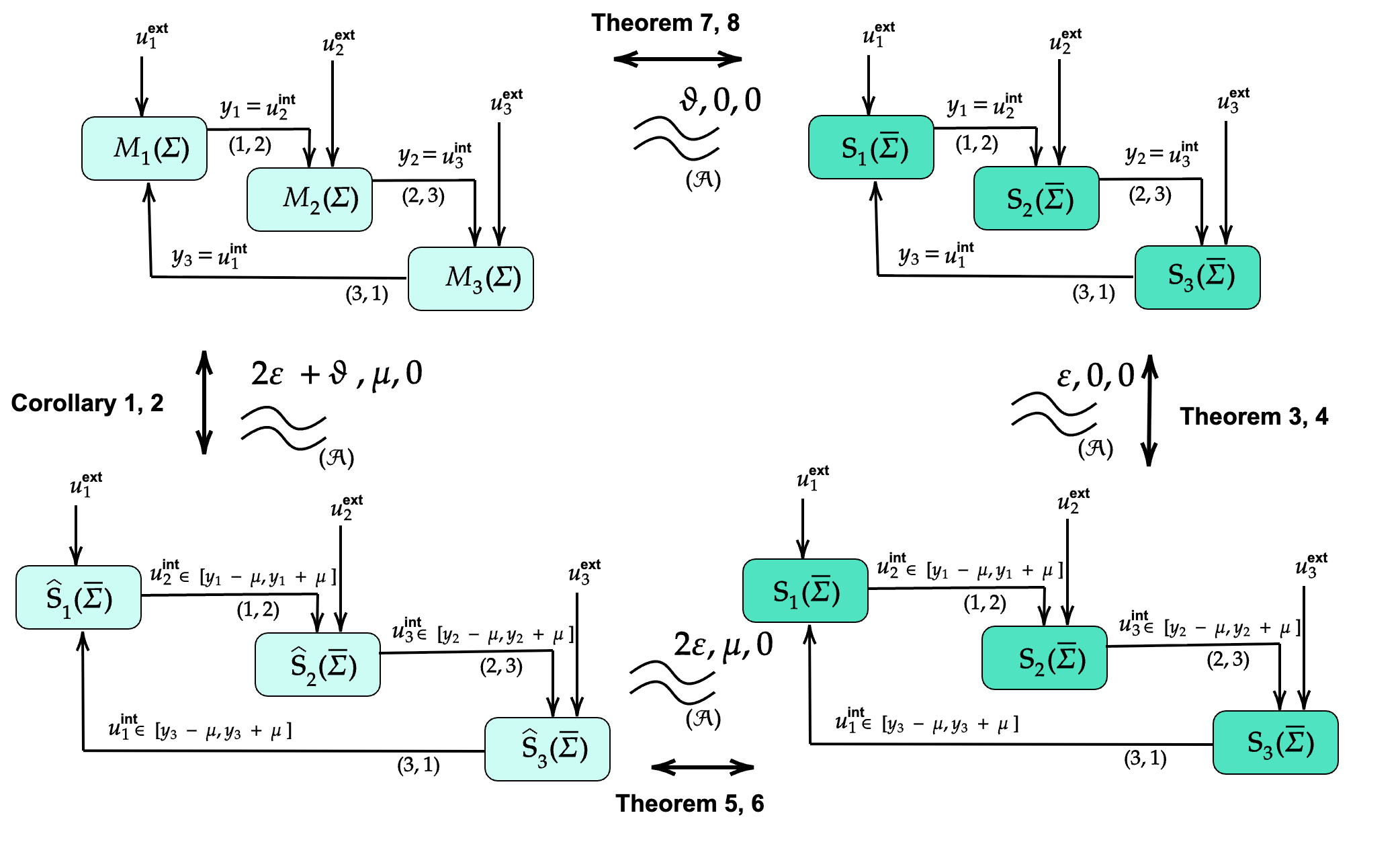}
    \caption{Analyzing the Construction of Bisimilar Abstractions for Interconnected Stochastic Systems. Top Left: Interconnected stochastic systems demonstrating exact composition. Top Right: Interconnected noise-free systems displaying exact composition. Bottom Right: Noise-free interconnected systems demonstrating approximate composition. Bottom Left: Abstraction achieved through approximate composition.}
    \label{fig:compositional}
\end{figure*}

\subsection{Assessing Conservatism in Interconnected Systems with $\delta$-ISS Property and Varying Approximate Composition Parameters}

The behavior of interconnected systems tolerating some composition parameter error becomes more conservative (and less deterministic) when the approximate composition parameters become large. The following result shows that under the $\delta$-ISS property, we can measure the conservatism of the approximate composition when increasing the approximate composition parameter.

\begin{thm}\label{THM6}
    Consider an interconnected noise-free transition system \hl{$\left\langle {S}_{i}(\bar{\Sigma})\right\rangle_{i \in I}^{\Xi, \mathcal{I}}$ and $\Xi=(\mu_{1}, \ldots, \mu_{N}) \in(\mathbb{R}_{0}^{+})^{N}$.} If each subsystem $S_i(\bar{\Sigma})$ is $\delta$-ISS with $\beta_i, \gamma_i^{int}$ as defined in \eqref{D-ISS}, and \hl{ $\left\langle {S}_{i}(\bar{\Sigma})\right\rangle_{i \in I}^{\Xi, \mathcal{I}}$ is compatible for $\Xi$-approximate composition with respect to $\mathcal{I}$, then it is also compatible for $\bar{\Xi}$-approximate composition with respect to $\mathcal{I}$, for any $\bar{\Xi}=\left(\bar{\mu}_{1}, \ldots, \bar{\mu}_{N}\right) \in(\mathbb{R}_{0}^{+})^{N}$ such that $\bar{\Xi} \geqslant \Xi$ (i.e., $\bar{\mu}_{i} \geqslant \mu_{i}, \forall i \in I$), given $U^{int} \subset \R^n$ is convex.} Moreover, for any $\varepsilon \in \mathbb R_0^{+}$ such that 
    \begin{equation}
        \label{eqn:thm6}
        \beta_{i}(\varepsilon,1) + \gamma^{int}_{i}(\varepsilon+\bar{\mu}_i-\mu_i) \leqslant \varepsilon,  \forall i \in I,
    \end{equation}
    using the pseudometric $\mathbf{d}$ as defined in \eqref{Eq10}, composed state $X = \prod_{i \in \mathcal{I}}X_i$, the relation $\mathcal{R}= \{ \left(x^{a}, x^{b}\right)  \in X \times X \: | \: \mathbf{d}(H(x^{a}),H(x^{b}))\leqslant \varepsilon \}$ is an $(\varepsilon,0)$-approximate bisimulation relation between $S_{\bar{\Xi}}=\left\langle S_{i}(\bar{\Sigma})\right\rangle_{i \in I}^{\bar{\Xi}, \mathcal{I}}$ and $S_{\Xi}=\langle S_{i}(\bar{\Sigma})\rangle_{i \in I}^{\Xi, \mathcal{I}}$.
\end{thm}

\begin{pf}
    It was shown in \cite{saoud2021compositional} that $S_{\Xi} \preccurlyeq^{0, 0} S_{\bar{\Xi}}$ is verified for $ \bar{\Xi}\geqslant \Xi$, which implies from Proposition \ref{prop2} that $S_{\Xi} \preccurlyeq^{\varepsilon, 0} S_{\bar{\Xi}}$. The rest of this proof will focus on showing that the symmetrical version $S_{\bar{\Xi}} \preccurlyeq^{\varepsilon, 0} S_{\Xi}$.
    
    The first and second conditions in Definition \ref{DefinitionBisimulation} are directly satisfied. (We have the same transition systems, with the same set of initial conditions and the same sets of states; thus, the distance between two elements $x^{a}$ and $x^{b}$ satisfying $(x^{a},x^{b})\in \mathcal{R}$ is bounded by $\varepsilon$). 
    
    Consider $ (x^{a},x^{b})\in \mathcal{R} $, with $x^{a}=(x_{1}^{a}, \dots, x_{N}^{a})$ and $ x^{b}=(x_{1}^{b}, \dots, x_{N}^{b}) $, any $u^{a,ext} \in U_{S_{\bar{\Xi}}}(x^{a})$ and any $x^{a\prime} \in \Delta_{\bar{\Xi}}(x^{a},u^{a,ext})$. Choose $u^{b,ext}=u^{a,ext}$ and let us show the existence of $x^{b\prime} \in \Delta_{\Xi}(x^{b}, u^{b,ext})$ satisfying $\left(x^{a\prime}, x^{b\prime}\right) \in \mathcal{R}$. Since $x^{a\prime} \in \Delta_{\bar{\Xi}}(x^{a},u^{a,ext})$, there exists $u_{i}^{a,int} \in U_{i}^{a,int}$ with $\mathbf{d}_{U_{i}^{int}}(u_{i}^{a,int}, \prod_{j \in \mathcal{N}(i)}\{y^a_{j}\}) \leqslant \bar{\mu}_{i}$, $(u_{i}^{a,ext}, u_{i}^{a,int}) \in U_{S_{i}(\bar{\Sigma})}(x^a_{i})$ and   $x_{i}^{a\prime} \in \Delta_{i}(x^a_{i}, u_{i}^{a,ext}, u_{i}^{a,int})$. 
    
    Now for $i \in I$, choose $u_{i}^{b,int} \in U_i^{int}$ satisfying $\mathbf{d}_{U_{i}^{int}}(u_{i}^{b,int}, \prod_{j \in \mathcal{N}(i)}\{y^b_{j}\})\leqslant \mu_i$ and $\mathbf{d}_{U_{i}^{int}}(u_{i}^{b,int},u_{i}^{a,int})   \leqslant \varepsilon +\bar{\mu}_{i}-\mu_i$. Such $u_{i}^{b,int}$ always exists using the fact that \hl{the collection $\{S_i(\bar{\Sigma})\}_{i \in I}$ is compatible for $\Xi$-approximate composition and $U^{int} \subset \R^n$ is convex}, $\mathbf{d}(\prod_{j \in \mathcal{N}(i)}\left\{y^a_{j}\right\},\prod_{j \in \mathcal{N}(i)}\left\{y^b_{j}\right\}) \leqslant \varepsilon$ and $\mathbf{d}_{U_{i}^{int}}(u_{i}^{a,int}, \prod_{j \in \mathcal{N}(i)}\{y^a_{j}\}) \leqslant \bar{\mu}_{i}$. Now pick $x^{b\prime} \in \Delta_{{\Xi}}(x^{b},u^{b,ext})$, defined as $x_{i}^{b\prime} \in \Delta_{i}(x^b_{i}, u_{i}^{b,ext}, u_{i}^{b,int})$ for $i \in I$. First, we have that for the chosen values of $u_{i}^{b,int}$, $i \in I$, the transition $x^{b\prime} \in \Delta_{{\Xi}}(x^{b},u^{b,ext})$ is well defined for the approximate composed system $S_{\Xi}$. Let us now show that $(x^{a\prime}, x^{b\prime}) \in \mathcal{R}$. From the definition of the relation $\mathcal{R}$ and under the assumption that each component satisfies the $\delta$-ISS property, we have for all $ i\in I $:
    \begin{align*}
            &\mathbf{d}(H(x^{a\prime}),H(x^{b\prime})) \\
            &= \max _{i \in I}\mathbf{d}_{Y_i}(\Delta_{i}(x_{i}^{a},u_{i}^{a,ext}, u_{i}^{a,int}),\Delta_{i}(x^b_{i},u_{i}^{b,ext},u_{i}^{b,int})) \\
            &\leqslant \max _{i \in I}  \beta_{i}(\varepsilon,1) + \gamma^{int}_{i}(\varepsilon+\bar{\mu_{i}}-\mu_i) \leqslant \varepsilon,
    \end{align*}
    where the last inequality follows from (\ref{eqn:thm6}). Hence, condition (iii) in Special Case \ref{SS: scase/bisim2} of Definition \ref{DefinitionBisimulation} holds and one obtains $S_{\bar{\Xi}} \preccurlyeq^{\varepsilon, 0} S_{\Xi}$. $\hfill\qed$
\end{pf}

\begin{rem}
    Note that the results in \cite{saoud2021compositional} show a simulation relation $S_{\Xi} \preccurlyeq^{\varepsilon,0} S_{\bar{\Xi}}$. The previous result, Theorem \ref{THM6}, shows that under the $\delta$-ISS property of each subsystem the symmetrical relation $S_{\bar{\Xi}} \preccurlyeq^{\varepsilon,0} S_{\Xi}$ holds and thus $S_{\bar{\Xi}} \approx ^{\varepsilon,0}  S_{\Xi}$. Indeed, while any trajectory of the system $S_{\Xi}$ is a trajectory of the system $S_{\bar{\Xi}}$, the proposed result shows that under the $\delta$-ISS property, one can measure the conservatism between $S_{\Xi}$ and $S_{\bar{\Xi}}$, thereby measuring the conservatism of the approximate composition.
\end{rem}

\begin{thm}\label{THM9}
    Consider an interconnected noise-free transition system \hl{$\left\langle {S}_{i}(\bar{\Sigma})\right\rangle_{i \in I}^{\Xi, \mathcal{I}}$ and $\Xi=(\mu_{1}, \ldots, \mu_{N}) \in(\mathbb{R}_{0}^{+})^{N}$.} If each subsystem $S_i(\bar{\Sigma})$ is $\delta$-ISS with $\beta_i, \gamma_i^{int}$ as defined in \eqref{D-ISS}, and \hl{$\left\langle {S}_{i}(\bar{\Sigma})\right\rangle_{i \in I}^{\Xi, \mathcal{I}}$ is compatible for $\Xi$-approximate composition with respect to $\mathcal{I}$, then it is also compatible for $\bar{\Xi}$-approximate composition with respect to $\mathcal{I}$, for any $\bar{\Xi}=\left(\bar{\mu}_{1}, \ldots, \bar{\mu}_{N}\right) \in(\mathbb{R}_{0}^{+})^{N}$ such that $\bar{\Xi} \geqslant \Xi$ (i.e., $\bar{\mu}_{i} \geqslant \mu_{i}, \forall i \in I$), given $U^{int} \subset \R^n$ is convex}. Moreover, for any $\varepsilon \in \mathbb R_0^{+}$ such that 
    \begin{equation}
        \label{eqn:thm8}
        \beta_{i}(\varepsilon,1) + \gamma^{int}_{i}(\varepsilon+\bar{\mu}_i-\mu_i) \leqslant \varepsilon,  \forall i \in I,
    \end{equation}
    using the pseudometric $\mathbf{d}$ as defined in \eqref{Eq10}, composed state $X = \prod_{i \in \mathcal{I}}X_i$, the relation $\mathcal{R}= \{ \left(x^{a}, x^{b}\right)  \in X \times X \: | \: \mathbf{d}(H(x^{a}),H(x^{b}))\leqslant \varepsilon \}$ is an $(\varepsilon,0)$-approximate alternating bisimulation relation between $S_{\bar{\Xi}}=\left\langle S_{i}(\bar{\Sigma})\right\rangle_{i \in I}^{\bar{\Xi}, \mathcal{I}}$ and $S_{\Xi}=\langle S_{i}(\bar{\Sigma})\rangle_{i \in I}^{\Xi, \mathcal{I}}$.
\end{thm}

\begin{pf}
    It was shown in \cite{saoud2021compositional} that $S_{\Xi} \widehat \preccurlyeq^{0, 0} S_{\bar{\Xi}}$ is verified for $ \bar{\Xi}\geqslant \Xi$, which implies from Proposition \ref{prop2} that $S_{\Xi} \widehat \preccurlyeq^{\varepsilon, 0} S_{\bar{\Xi}}$. The rest of this proof will focus on showing that the condition \eqref{eqn:thm8} guarantees the symmetrical version $S_{\bar{\Xi}} \widehat \preccurlyeq^{\varepsilon, 0} S_{\Xi}$. The first and second conditions in Definition \ref{DefinitionAlternatingBisimulation} are directly satisfied. Consider $ (x^{a},x^{b})\in \mathcal{R} $, with $x^{a}=(x_{1}^{a}, \dots, x_{N}^{a})$ and $x^{b}=(x_{1}^{b}, \dots, x_{N}^{b})$. Let $u^{b,ext} \in U_{S_{{\Xi}}}(x^{b})$ be any external input for the system $S_{\Xi}$ and choose $u^{a,ext}=u^{b,ext}$ for the system $S_{\bar{\Xi}}$. We must prove that for any $ x^{a\prime} \in \Delta_{\bar{\Xi}}(x^{a},u^{a,ext})$, there exists $x^{b\prime} \in \Delta_{\Xi}(x^{b}, u^{b,ext})$ satisfying $\left(x^{a\prime}, x^{b\prime}\right) \in \mathcal{R}$. 
    
    Consider an arbitrary transition $x^{a\prime} \in \Delta_{\bar{\Xi}}(x^{a},u^{a,ext})$. This implies that for each $i \in I$, there exists $u_{i}^{a,int} \in U_i^{int}$ satisfying $\mathbf{d}_{U_{i}^{int}}(u_{i}^{a,int}, \prod_{j \in \mathcal{N}(i)}\left\{y^a_{j}\right\})\leqslant \bar{\mu}_i$, such that $x_{i}^{a\prime} \in \Delta_{i}(x^a_{i}, u_{i}^{a,ext}, u_{i}^{a,int})$.
    
    Now for $i \in I$, we choose $u_{i}^{b,int} \in U_{i}^{b,int}$ such that $\mathbf{d}_{U_{i}^{int}}(u_{i}^{b,int}, \prod_{j \in \mathcal{N}(i)}\left\{y^b_{j}\right\}) \leqslant {\mu}_{i}$ and $\mathbf{d}_{U_{i}^{int}}(u_{i}^{b,int}, u_{i}^{a,int})  \leqslant \varepsilon + \bar{\mu}_{i} -\mu_i$. Such $u_{i}^{b,int}$ always exists using the fact that \hl{the collection $\{S_i(\bar{\Sigma})\}_{i \in I}$ is compatible for $\Xi$-approximate composition and $U^{int} \subset \R^n$ is convex}, $\mathbf{d}(\prod_{j \in \mathcal{N}(i)}\left\{y^a_{j}\right\},\prod_{j \in \mathcal{N}(i)}\left\{y^b_{j}\right\}) \leqslant \varepsilon$ and $\mathbf{d}_{U_{i}^{int}}\left(u_{i}^{a,int}, \prod_{j \in \mathcal{N}(i)}\left\{y^a_{j}\right\}\right) \leqslant \bar{\mu}_{i}$.
    
    Pick $x^{b\prime} \in \Delta_{\Xi}(x^{b},u^{b,ext})$, defined for $i \in I$ as $x_{i}^{b\prime} \in \Delta_{i}\left(x^b_{i}, u_{i}^{b,ext}, u_{i}^{b,int}\right)$. First, we have that for the chosen values of $u_{i}^{b,int}$, the transition $x^{b\prime}$ is well defined for the approximate composed system $S_{\Xi}$. Let us now show that $\left(x^{a\prime}, x^{b\prime}\right) \in \mathcal{R}$. From the definition of the relation $\mathcal{R}$ and under the assumption that each component satisfies the $\delta$-ISS property, we have for all $ i\in I $:
    
\begin{align*}
    \mathbf{d}(H(x^{a\prime}),H(x^{b\prime})) &= \max _{i \in I}\mathbf{d}_{Y_i}\big(\Delta_{i}(x_{i}^{a},u_{i}^{a,ext}, u_{i}^{a,int}),\Delta_{i}(x^b_{i},u_{i}^{b,ext},u_{i}^{b,int})\big)\\
    &\leqslant \max _{i \in I}\ \beta_{i}(\varepsilon,1) + \gamma^{int}_{i}(\varepsilon+\bar{\mu_{i}}-\mu_i) \leqslant \varepsilon,
\end{align*}
    where the last inequality follows from (\ref{eqn:thm8}). Thus, condition (iii) in Special Case \ref{SS: scase/abisim2} of Definition \ref{DefinitionAlternatingBisimulation} yields $S_{\bar{\Xi}} \widehat \preccurlyeq^{\varepsilon, 0} S_{\Xi}$. $\hfill\qed$
\end{pf}

\begin{rem}
    Note that in \cite{saoud2021compositional}, we have the alternating simulation relation from $S_{\Xi} \widehat \preccurlyeq^{\varepsilon,0} S_{\bar{\Xi}}$ without any stability requirement. The proposed result, Theorem \ref{THM9}, shows the symmetrical version $ S_{\bar{\Xi}} \widehat \preccurlyeq^{\varepsilon,0} S_{\Xi}  $ under the $\delta$-ISS property of each subsystem.
\end{rem}	

\subsection{Approximate (Alternating) Bisimilar Composition of the Global noise-free subsystems}

The compositionality results for the approximate bisimulation relation are stated as follows.

\begin{thm}\label{THM10} 
    Let $\left\langle {S}_{i}(\bar{\Sigma})\right\rangle_{i \in I}^{{\Xi}, \mathcal{I}}$ and $\left\langle \hat{S}_{i}(\bar{\Sigma})\right\rangle_{i \in I}^{\hat{\Xi}, \mathcal{I}}$ be two interconnected transition systems with $S_{i}=(X_{i}, X_{i}^{0}, U_{i}^{ext}, U_{i}^{int}, \Delta_{i}, Y_{i}, H_{i})$ and $\hat{S}_{i}=\left(\hat{X}_{i}, \hat{X}_{i}^{0}, \hat{U}_{i}^{ext}, \right. $ $\left.\hat{U}_{i}^{int}, \hat{\Delta}_{i}, \hat{Y}_{i}, \hat{H}_{i}\right)$. Consider positive constants $\varepsilon_{i}, \mu_{i},$ for $ i \in I$, with $\varepsilon=\max _{i \in I} \varepsilon_{i}$, $\mu=\max_{i \in I} \mu_{i}$ and consider $\Xi=\left(\kappa_{1}, \ldots, \kappa_{N}\right) $ and $\hat{\Xi}=(\mu_1+\kappa_1+\varepsilon,\ldots, \mu_N+\kappa_N+\varepsilon)$. Let the following hold:
    \begin{itemize}
        \item[$(\mathrm{i})$] For all $i \in I$, $ S_{i}(\bar{\Sigma}) $ is $ \delta $-ISS and satisfies the following inequality $ \max_{i \in I} \big(\beta_{i}(\varepsilon,1) + \gamma^{int}_{i}(3\varepsilon+2\mu_i) \big)\leqslant \varepsilon$;
        \item[$(\mathrm{ii})$] For all $i \in I$, $ S_{i}(\bar{\Sigma}) $ is $(\varepsilon_{i},\mu_{i})$-approximately bisimilar to $ \hat{S}_{i}(\bar{\Sigma}) $, and we denote $ S_{i}(\bar{\Sigma}) \approx^{\varepsilon_{i},\mu_{i}} \hat{S}_{i}(\bar{\Sigma}) $;
        \item[$(\mathrm{iii})$] $\left\langle {S}_{i}(\bar{\Sigma})\right\rangle_{i \in I}^{{\Xi}, \mathcal{I}}$ is compatible for $\Xi$-approximate composition with respect to $\mathcal{I}$;
        \item[$(\mathrm{iv})$] $\left\langle \hat{S}_{i}(\bar{\Sigma})\right\rangle_{i \in I}^{\hat{\Xi}, \mathcal{I}}$ is compatible for $\hat{\Xi}$-approximate composition with respect to $\mathcal{I}$;
    \end{itemize}
    then, $S_{\Xi} = \left\langle {S}_{i}(\bar{\Sigma})\right\rangle_{i \in I}^{{\Xi}, \mathcal{I}}$ is $(2\varepsilon,\mu)$-approximately bisimilar to $\hat{S}_{\hat{\Xi}}=\left\langle \hat{S}_{i}(\bar{\Sigma})\right\rangle_{i \in I}^{\hat{\Xi}, \mathcal{I}}$.
\end{thm}

\begin{pf}
    First, we have from \cite{saoud2021compositional} that $ {S}_{\Xi} \preccurlyeq^{\varepsilon,\mu}  \hat{S}_{\hat{\Xi}}$. Hence, one gets from Proposition \ref{prop2} that ${S}_{\Xi} \preccurlyeq^{2\varepsilon,\mu}  \hat{S}_{\hat{\Xi}}$. Let us now show that $\hat{S}_{\hat{\Xi}}\preccurlyeq^{2\varepsilon,\mu} {S}_{\Xi}$.
    
    For $i \in I$, let $\mathcal{R}_i$ be the $(\varepsilon_i,\mu_i)$-approximate simulation relation from $\hat{S}_i(\bar{\Sigma})$ to $S_i(\bar{\Sigma})$, and let us first show that for $X = \prod_{i \in I}X_i, \hat{X} = \prod_{i \in I}\hat{X}_i$, the relation $\mathcal{R}$ defined by $\mathcal{R}:= \{ \left(x, \hat{x}\right)  \in X \times \hat{X} \: | \: (x_i,\hat{x}_i) \in \mathcal{R}_i, \forall i \in I \}$, with $x=\left(x_{1}, \ldots, x_{N}\right)$ and $\hat{x}=\left(\hat{x}_{1}, \ldots, \hat{x}_{N}\right)$, is an $(\varepsilon,\mu)$-approximate simulation relation from $\hat{S}_{\hat{\Xi}}$ to $S_{\bar{\Xi}}$, with $\bar{\Xi}=(2\mu_1+\kappa_1+2\varepsilon,\ldots, 2\mu_N+\kappa_N+2\varepsilon)$. The first condition is directly satisfied (we suppose that for all $ i\in I $, $ \hat{S}_{i}(\bar{\Sigma}) $ is $(\varepsilon_i,\mu_i)$-approximately similar to $ S_{i}(\bar{\Sigma}) $, thus, for each initial condition in $ \hat{S}_{\hat{\Xi}} $ we can find an initial condition in $ S_{\bar{\Xi}} $). Let $(x, \hat{x}) \in \mathcal{R}$ with $x=\left(x_{1}, \ldots, x_{N}\right)$ and $\hat{x}=\left(\hat{x}_{1}, \ldots, \hat{x}_{N}\right)$. Using the definition of the output map for approximate composition, relation \eqref{Eq10} and condition (ii) of  {Definition \ref{DefinitionBisimulation}}, we have, 
    
\begin{align*}
    \mathbf{d}(H(x), \hat{H}(\hat{x})) &= \mathbf{d}\big((H_{1}(x_{1}), \ldots, H_{N}(x_{N})),(\hat{H}_{1}(\hat{x}_{1}), \ldots, \hat{H}_{N}(\hat{x}_{N}))\big)\\
    &= \max _{i \in I} \mathbf{d}_{Y_{i}}\big(H_{i}(x_{i}), \hat{H}_{i}(\hat{x}_{i})\big) \leqslant \max _{i \in I} \varepsilon_{i}=\varepsilon,
\end{align*}
    where the last inequality follows from the fact that $(x_i,\hat{x}_i) \in \mathcal{R}_i$ for all $i \in I$.
    
    Consider $(x,\hat{x})\in\mathcal{R}$ with $x=(x_1,\ldots,x_N)$ and $\hat{x}=(\hat{x}_1,\ldots,\hat{x}_N)$, any $\hat{u}^{ext} \in \hat{U}_{\hat{S}_{\hat{\Xi}}}(\hat{x})$ with $\hat{u}^{ext}=(\hat{u}_1^{ext},\ldots,\hat{u}_N^{ext})$ and any $\hat{x}' \in \hat{\Delta}_{\hat{\Xi}}(\hat{x},\hat{u}^{ext})$. Let us prove the existence of $u^{ext} \in U_{S_{\bar{\Xi}}}(x)$ with $\mathbf{d}_{U^{ext}}(u^{ext},\hat{u}^{ext}) \leqslant \mu$ and ${x}' \in \Delta_{\bar{\Xi}}({x},{u}^{ext})$ satisfying $(x',\hat{x}')\in \mathcal{R}$.
    
    From the definition of the relation $\mathcal R$, we have for all $i\in I$, $(x_i,\hat{x}_i)\in \mathcal{R}_i$, then from the third condition of Special Case \ref{SS: scase/bisim2} of  {Definition \ref{DefinitionBisimulation}}, we have for all $(\hat{u}_i^{ext},\hat{u}_i^{int}) \in \hat{U}_{\hat{S}_i(\bar{\Sigma})}(\hat{x}_i)$, and for any $\hat{x}_i'\in \hat{\Delta}_i(\hat{x}_i,\hat{u}_i^{ext},\hat{u}_i^{int})$ the existence of $(u_i^{ext},u_i^{int})\in U_{S_i(\bar{\Sigma})}(x_i)$ with $\mathbf{d}_{U_i^{ext}}(u_i^{ext},\hat{u}_i^{ext})\leq\mu_i$ and  $\mathbf{d}_{U_i^{int}}(u_i^{int},\hat{u}_i^{int})\leq\mu_i$ and the existence of ${x}_i'\in \Delta_i({x}_i,{u}_i^{ext},{u}_i^{int})$ such that $(x_i',\hat{x}_i')\in \mathcal{R}_i$.
    
    Let us show that the input ${u}^{int}=({u}^{int}_1,\ldots,{u}^{int}_N)$ satisfies the requirement of the $\bar{\Xi}$-approximate composition of the components $\{{S}_i(\bar{\Sigma})\}_{i\in I}$. The condition $\mathbf{d}_{U_i^{int}}(u_i^{int},\hat{u}_i^{int})\leqslant \mu_i$ implies that,
    \begin{align*}
        &\mathbf{d}_{U_i^{int}}({u}_i^{int},\prod_{j\in \mathcal{N}(i)}\{{y}_j\}) \\
        &\leqslant \mathbf{d}_{U_i^{int}}(\hat{u}_i^{int},u_i^{int})+\mathbf{d}_{U_i^{int}}(\hat{u}_i^{int},\prod_{j\in \mathcal{N}(i)}\{{y}_j\})\\
        & \leqslant \mathbf{d}_{U_i^{int}}(\hat{u}_i^{int},u_i^{int}) +\mathbf{d}_{U_i^{int}}(\hat{u}_i^{int},\prod_{j\in \mathcal{N}(i)}\{\hat{y}_j\}) \\
        &+\mathbf{d}_{U_i^{int}}(\prod_{j\in \mathcal{N}(i)}\{y_j\},\prod_{j\in \mathcal{N}(i)}\{\hat{y}_j\})\\
        &\leqslant \mu_i + \mu_i+\kappa_i+\varepsilon+ \max\limits_{j\in \mathcal{N}(i)} \varepsilon_j \\
        &\leqslant 2\mu_i +\kappa_i+\varepsilon+ \max\limits_{j \in I} \varepsilon_j =2\mu_i+\kappa_i+2\varepsilon.
    \end{align*}
    Hence, from the condition (iii) of the $\bar{\Xi}$-approximate composition with respect to $\mathcal{I}$ of $\{{S}_i(\bar{\Sigma})\}_{i\in I}$ is well defined in the sense of  {Definition \ref{DefComposition}}. Thus, condition (iii) in Special Case \ref{SS: scase/bisim2} of  {Definition \ref{DefinitionBisimulation}} holds with $u^{ext}=(u_1^{ext},\ldots,u_N^{ext})$ satisfying $\mathbf d_{U^{ext}}(u^{ext},\hat{u}^{ext})=\max\limits_ {i \in I}\{\mathbf d_{U_i^{ext}}(u^{ext}_i,\hat{u}^{ext}_i)\}  = \max\limits_{i \in I}\{\mu_i\} =\mu$, and one obtains $\hat{S}_{\hat{\Xi}}\preccurlyeq^{\varepsilon,\mu} {S}_{\bar{\Xi}}$. 
    
\begin{table}[t]
\centering
\color{hlcolor}\caption{Derivation of the internal input mismatch $3\varepsilon+2\mu_i$.}
\label{tab:error-terms}
\begin{tabular}{cll}
\toprule
\textbf{Term} & \textbf{Source} & \textbf{Bound} \\
\midrule
1 & $\mathbf{d}_{U_i^{int}}(\hat{u}_i^{int},u_i^{int})$ & $\mu_i$ \\
2 & $\mathbf{d}_{U_i^{int}}(\hat{u}_i^{int},\prod_{j\in \mathcal{N}(i)}\{\hat{y}_j\})$ & $\mu_i +\kappa_i + \varepsilon$ \\
3 & $\mathbf{d}_{U_i^{int}}(\prod_{j\in \mathcal{N}(i)}\{y_j\},\prod_{j\in \mathcal{N}(i)}\{\hat{y}_j\})$ & $\varepsilon$ \\
\midrule
\multicolumn{3}{l}{Input composition mismatch: $\;\bar\mu_i-\mu_i=2\mu_i+ \kappa_i + 2\varepsilon - \kappa_i$} \\
\midrule
\multicolumn{3}{l}{\textbf{Input mismatch by Thm.~\ref{THM6}: $\;\varepsilon + 2\varepsilon+2\mu_i$}}\\
\bottomrule
\end{tabular}
\end{table}
    Now from (i), using the fact that $\beta_{i}(\varepsilon,1) + \gamma^{int}_{i}(3\varepsilon+2\mu_i) \leqslant \varepsilon,~~  \forall i \in I$, one gets from Theorem \ref{THM6}, \hl{also detailed in Table \ref{tab:error-terms}}, that ${S}_{\bar{\Xi}} \preccurlyeq^{\varepsilon,0}  {S}_{{\Xi}}$. Hence, using the transitivity relation of the simulation relation in Proposition \ref{prop1}, one gets $\hat{S}_{\hat{\Xi}}\preccurlyeq^{2\varepsilon,\mu} {S}_{\Xi}$. $\hfill\qed$
\end{pf}

\begin{rem}
    Definitions \ref{DefinitionBisimulation} and \ref{DefinitionAlternatingBisimulation} define the notions of stochastic bisimulation and alternating bisimulation for subsystems. These definitions extend naturally from transition subsystems to interconnected transition systems by considering systems without internal inputs and with coupling measure $\hl{\varrho} = 0$.
\end{rem}

We now present the analogous result for approximate alternating bisimulation relations.

\begin{thm}\label{THMCompVerification}
    Let $\left\langle {S}_{i}(\bar{\Sigma})\right\rangle_{i \in I}^{{\Xi}, \mathcal{I}}$ and $\left\langle \hat{S}_{i}(\bar{\Sigma})\right\rangle_{i \in I}^{\hat{\Xi}, \mathcal{I}}$ be two interconnected transition systems with $S_{i}=(X_{i}, X_{i}^{0}, U_{i}^{ext}, U_{i}^{int}, \Delta_{i}, Y_{i}, H_{i})$ and $\hat{S}_{i}=\left(\hat{X}_{i}, \hat{X}_{i}^{0}, \hat{U}_{i}^{ext}, \right. $ $\left. \hat{U}_{i}^{int}, \hat{\Delta}_{i}, \hat{Y}_{i}, \hat{H}_{i}\right)$. Consider positive constants $\varepsilon_{i}, \mu_{i},$ for $ i \in I$, with $\varepsilon=\max _{i \in I} \varepsilon_{i}$, $\mu=\max_{i \in I} \mu_{i}$ and consider $\Xi=\left(\kappa_{1}, \ldots, \kappa_{N}\right) $ and $\hat{\Xi}=(\mu_1+\kappa_1+\varepsilon,\ldots, \mu_N+\kappa_N+\varepsilon)$. Let the following hold:
    \begin{itemize}
        \item[$(\mathrm{i})$] For all $i \in I$, $ S_{i}(\bar{\Sigma}) $ is $ \delta $-ISS and satisfies the following inequality $ \max_{i \in I} \big(\beta_{i}(\varepsilon,1) + \gamma^{int}_{i}(3\varepsilon+2\mu_i) \big)\leqslant \varepsilon$;
        \item[$(\mathrm{ii})$] For all $i \in I$, $ S_{i}(\bar{\Sigma}) $ is $(\varepsilon_{i},\mu_{i})$-approximately alternatingly bisimilar to $ \hat{S}_{i}(\bar{\Sigma}) $, and we denote $ S_{i}(\bar{\Sigma}) \widehat \approx^{\varepsilon_{i},\mu_{i}} \hat{S}_{i}(\bar{\Sigma}) $;
        \item[$(\mathrm{iii})$] $\left\langle {S}_{i}(\bar{\Sigma})\right\rangle_{i \in I}^{{\Xi}, \mathcal{I}}$ is compatible for $\Xi$-approximate composition with respect to $\mathcal{I}$;
        \item[$(\mathrm{iv})$] $\left\langle \hat{S}_{i}(\bar{\Sigma})\right\rangle_{i \in I}^{\hat{\Xi}, \mathcal{I}}$ is compatible for $\hat{\Xi}$-approximate composition with respect to $\mathcal{I}$;
    \end{itemize}
    then, ${S}_{\Xi}=\left\langle {S}_{i}(\bar{\Sigma})\right\rangle_{i \in I}^{{\Xi}, \mathcal{I}}$ is $(2\varepsilon,\mu)$-approximately alternatingly bisimilar to $\hat{S}_{\hat{\Xi}}=\left\langle \hat{S}_{i}(\bar{\Sigma})\right\rangle_{i \in I}^{\hat{\Xi}, \mathcal{I}}$.
\end{thm}

\begin{pf}
    First, we have from \cite{saoud2021compositional} that $ {S}_{{\Xi}} \widehat \preccurlyeq^{\varepsilon,\mu}  \hat{S}_{\hat{{\Xi}}}$. Hence, one gets from Proposition \ref{prop2} that ${S}_{\Xi} \widehat \preccurlyeq^{2\varepsilon,\mu}  \hat{S}_{\hat{{\Xi}}}$. Let us now show that $\hat{S}_{\hat{\Xi}}\widehat \preccurlyeq^{2\varepsilon,\mu} {S}_{{\Xi}}$.
    
    For $i \in I$, let $\mathcal{R}_i$ be the $(\varepsilon_i,\mu_i)$-approximate alternating simulation relation from $\hat{S}_i(\bar{\Sigma})$ to $S_i(\bar{\Sigma})$, and let us first show that for $X = \prod_{i \in I}X_i, \hat{X} = \prod_{i \in I}\hat{X}_i$, the relation $\mathcal{R}$ defined by $\mathcal{R}:= \{ \left(x, \hat{x}\right)  \in X \times \hat{X} \: | \: (x_i,\hat{x}_i) \in \mathcal{R}_i, \forall i \in I \}$, with $x=\left(x_{1}, \ldots, x_{N}\right)$ and $\hat{x}=\left(\hat{x}_{1}, \ldots, \hat{x}_{N}\right)$, is an $(\varepsilon,\mu)$-approximate alternating simulation relation from $\hat{S}_{\hat{\Xi}}$ to $S_{\bar{\Xi}}$, with $\bar{\Xi}=(2\mu_1+\kappa_1+2\varepsilon,\ldots, 2\mu_N+\kappa_N+2\varepsilon)$.    
    The first condition is directly satisfied (we suppose that for all $ i\in I $, $ \hat{S}_{i}(\bar{\Sigma}) $ is $(\varepsilon_i,\mu_i)$-approximately alternatingly similar to $ S_{i}(\bar{\Sigma}) $, thus, for each initial condition in $ S_{\bar{\Xi}} $ we can find an initial condition in $ \hat{S}_{\hat{\Xi}}$).   
    Let $(x, \hat{x}) \in \mathcal{R}$ with $x=\left(x_{1}, \ldots, x_{N}\right)$ and $\hat{x}=\left(\hat{x}_{1}, \ldots, \hat{x}_{N}\right)$. Using the definition of the output map for approximate composition, relation \eqref{Eq10} and condition (ii) of  {Definition \ref{DefinitionAlternatingBisimulation}}, we have, 
\begin{align*}
    \mathbf{d}(H(x), \hat{H}(\hat{x})) &= \mathbf{d}\big((H_{1}(x_{1}), \ldots, H_{N}(x_{N})),(\hat{H}_{1}(\hat{x}_{1}), \ldots, \hat{H}_{N}(\hat{x}_{N}))\big)\\
    &= \max _{i \in I} \mathbf{d}_{Y_{i}}\big(H_{i}(x_{i}), \hat{H}_{i}(\hat{x}_{i})\big) \leqslant \max _{i \in I} \varepsilon_{i}=\varepsilon,
\end{align*}
    where the last inequality follows from the fact that $(x_i,\hat{x}_i) \in \mathcal{R}_i$ for all $i \in I$. Consider $(x,\hat{x})\in\mathcal{R}$ with $x=(x_1,\ldots,x_N)$ and $\hat{x}=(\hat{x}_1,\ldots,\hat{x}_N)$, and $\hat{u}^{ext} \in \hat{U}_{\hat{S}_{\hat{\Xi}}}(\hat{x})$ with $\hat{u}^{ext}=(\hat{u}_1^{ext},\ldots,\hat{u}_N^{ext})$. Let us prove the existence of $\hat{u}^{ext} \in \hat{U}_{\hat{S}_{\hat{\Xi}}}(\hat{x})$ for any ${u}^{ext} \in {U}_{{S}_{\bar{\Xi}}}({x})$ with $\mathbf{d}_{U^{ext}}(u^{ext},\hat{u}^{ext}) \leq \mu$ such that for all $\hat{x}' \in \hat{\Delta}_{\hat{\Xi}}(\hat{x},\hat{u}^{ext})$, we have the existence of ${x}' \in \Delta_{\bar{\Xi}}({x},{u}^{ext})$ satisfying $(x',\hat{x}')\in \mathcal{R}$. From the definition of the relation $\mathcal{R}$, we have for all $i\in I$, $(x_i,\hat{x}_i)\in \mathcal{R}_i$, then from the third condition of Special Case \ref{SS: scase/abisim2} of  {Definition \ref{DefinitionAlternatingBisimulation}}, we have for all $(u_i^{ext},u_i^{int})\in U_{S_i(\bar{\Sigma})}(x_i)$, the existence of $(\hat{u}_i^{ext},\hat{u}_i^{int}) \in \hat{U}_{\hat{S}_i(\bar{\Sigma})}(\hat{x}_i)$ with $\mathbf{d}_{U_i^{ext}}(u_i^{ext},\hat{u}_i^{ext})\leq\mu_i$ and $\mathbf{d}_{U_i^{int}}(u_i^{int},\hat{u}_i^{int})\leq\mu_i$, such that for all $\hat{x}_i'\in \hat{\Delta}_i(\hat{x}_i,\hat{u}_i^{ext},\hat{u}_i^{int})$, we have the existence of $x_i'\in \Delta_i(x_i,u_i^{ext},u_i^{int})$ satisfying $(x_i',\hat{x}_i')\in \mathcal{R}_i$.

    Let us show that the input ${u}^{int}=({u}^{int}_1,\ldots,{u}^{int}_N)$ satisfies the requirement of the $\bar{\Xi}$-approximate composition of the components $\{{S}_i(\bar{\Sigma})\}_{i\in I}$. The condition $\mathbf{d}_{U_i^{int}}(u_i^{int},\hat{u}_i^{int})\leq \mu_i$ implies that,  
    
\begin{align*}
    \mathbf{d}_{U_i^{int}}\Big({u}_i^{int},\prod_{j\in \mathcal{N}(i)}\{{y}_j\}\Big)
    &\leqslant \mathbf{d}_{U_i^{int}}(\hat{u}_i^{int},u_i^{int})+\mathbf{d}_{U_i^{int}}\Big(\hat{u}_i^{int},\prod_{j\in \mathcal{N}(i)}\{{y}_j\}\Big) \\
    &\leqslant \mathbf{d}_{U_i^{int}}(\hat{u}_i^{int},u_i^{int}) +\mathbf{d}_{U_i^{int}}\Big(\hat{u}_i^{int},\prod_{j\in \mathcal{N}(i)}\{\hat{y}_j\}\Big) \\
    &\quad +\mathbf{d}_{U_i^{int}}\Big(\prod_{j\in \mathcal{N}(i)}\{y_j\},\prod_{j\in \mathcal{N}(i)}\{\hat{y}_j\}\Big)\\
    &\leqslant \mu_i + \mu_i+\kappa_i+\varepsilon+ \max\limits_{j\in \mathcal{N}(i)} \varepsilon_j \\
    &\leqslant 2\mu_i +\kappa_i+\varepsilon+ \max\limits_{j \in I} \varepsilon_j =2\mu_i+\kappa_i+2\varepsilon.
\end{align*}
    Hence, from the condition (iii) of the $\bar{\Xi}$-approximate composition with respect to $\mathcal{I}$ of $\{{S}_i(\bar{\Sigma})\}_{i\in I}$ is well defined in the sense of  {Definition \ref{DefComposition}}. Thus, condition (iii) in Special Case \ref{SS: scase/abisim2} of  {Definition \ref{DefinitionAlternatingBisimulation}} holds with $u^{ext}=(u_1^{ext},\ldots,u_N^{ext})$ satisfying $\mathbf d_{U^{ext}}(u^{ext},\hat{u}^{ext})=\max\limits_ {i \in I}\{\mathbf d_{U_i^{ext}}(u^{ext}_i,\hat{u}^{ext}_i)\}  = \max\limits_{i \in I}\{\mu_i\} =\mu$, and one obtains $\hat{S}_{\hat{\Xi}}\widehat \preccurlyeq^{\varepsilon,\mu} {S}_{\bar{\Xi}}$.    
    Now from (i), using the fact that $\beta_{i}(\varepsilon,1) + \gamma^{int}_{i}(3\varepsilon+2\mu_i) \leq \varepsilon,~~  \forall i \in I$, one gets from Theorem \ref{THM9}, \hl{similar to steps in Table \ref{tab:error-terms}}, that ${S}_{\bar{\Xi}} \widehat \preccurlyeq^{\varepsilon,0}  {S}_{{\Xi}}$. The transitivity relation of the alternating simulation relation in Proposition \ref{prop1}, yields $\hat{S}_{\hat{\Xi}}\widehat \preccurlyeq^{2\varepsilon,\mu} {S}_{{\Xi}}$. $ \hfill\qed $
\end{pf}

\subsection{Approximate Stochastic (Alternating) Bisimilar Relation between the Global Interconnected Stochastic system and noise-free system}

The compositionality results for the approximate bisimulation relation are stated as follows.

\begin{thm}\label{THM16} 
Let $\hl{\left\langle {M}_{i}({\Sigma})\right\rangle_{i \in I}^{\mathbf{0}_{N}, \mathcal{I}}} = \left(X, X^{0}, U^{ext}, \hl{\mathbb{T}_{C}}, Y, \right. $ $\left. H\right)$ be interconnected gMDPs and $\left\langle {S}_{i}(\bar{\Sigma})\right\rangle_{i \in I}^{{\Xi}, \mathcal{I}} = \left(X, X^{0}, U^{ext}, \Delta_{\Xi}, Y, \right. $ $\left. H\right)$ be interconnected transition system with $\hl{M_{i}(\Sigma)}=(X_{i}, X_{i}^{0}, U_{i}^{ext}, U_{i}^{int}, \hl{\mathbb{T}_{i}}, Y_{i}, H_{i})$ and $S_{i}(\bar{\Sigma})=\left(X_{i}, X_{i}^{0},  \right. $ $\left. U_{i}^{ext}, U_{i}^{{int}}, \Delta_{i}, Y_{i}, H_{i}\right)$. Consider positive constants $\vartheta_{i}$, with \hl{$\vartheta^q=\max _{i \in I} \vartheta_{i}^q$}, and consider $\hl{\Xi}=\left(\kappa_{1}, \ldots, \kappa_{N}\right) $. Let the following hold:
\begin{itemize}
	\item[$(\mathrm{i})$]  For all $i \in I$, $ M_{i}(\Sigma) $ is $\delta$-ISS-$M_{q}$ and satisfies the inequality $ \max_{i \in I} \mathcal{B}_{\vartheta}^{i} + \hl{2^{q-1}}(\gamma_i^{int}(\kappa_i) + \gamma_i^{int}(\hl{\vartheta_i}))\leqslant \hl{\vartheta^q}$;
    \item[$(\mathrm{ii})$]  For all $i \in I$, $ M_{i}(\Sigma) $ is $(\vartheta_{i},0,0)$-approximately stochastic bisimilar to $ M(S_{i}(\bar{\Sigma})) $, denoted as $ M_{i}(\Sigma) \approx^{\vartheta_{i},0}_0 M(S_{i}(\bar{\Sigma})) $;
	\item[$(\mathrm{iii})$] $\left\langle M({S}_{i}(\bar{\Sigma}))\right\rangle_{i \in I}^{{\Xi}, \mathcal{I}}$ is compatible for $\Xi$-approximate composition with respect to $\mathcal{I}$;
\end{itemize}
then the concrete system ${M}^{c}=\left\langle {M}_{i}({\Sigma})\right\rangle_{i \in I}^{\mathbf{0}_{N}, \mathcal{I}}$ is $(\vartheta,0,0)$-approximately \hl{stochastic} bisimilar to $M(S_{{\Xi}})=\left\langle M({S}_{i}(\bar{\Sigma}))\right\rangle_{i \in I}^{{\Xi}, \mathcal{I}}$ \hl{as in Special Case \ref{SS: scase/bisim1}.}
\end{thm}
\begin{pf}
	For all $ i\in I$, $ M_{i}(\Sigma) $ is $(\vartheta_i,0,0)$-approximately stochastic bisimilar to $ M(S_{i}(\bar{\Sigma})) $. For $i \in I$, let $\mathcal{R}_i$ be the $(\vartheta_i,0,0)$-approximate \hl{stochastic} bisimulation relation between $M({S}_i(\bar{\Sigma}))$ and $M_{i}(\Sigma)$ \hl{as in Special Case \ref{SS: scase/bisim1} where $\x$ as bold letter represents random variable}. Let us first show that the composed relation $\mathcal{R}$, for composed state set $X = \prod_{i \in I}X_i$, defined by $\mathcal{R}= \{ \left(\x, \bar{x}\right)  \in X \times {X} \: | \: (\x_i,\bar{x}_i) \in \mathcal{R}_i, i\in I\}$ \hl{as in Special Case \ref{SS: scase/bisim1}}, with $\x=\left(\x_{1}, \ldots, \x_{N}\right)$ and $\bar{x}=\left(\bar{x}_{1}, \ldots, \bar{x}_{N}\right)$, is a $(\vartheta,0,0)$-approximate stochastic bisimulation relation between $M(S_{\Xi})$ and $M^{c}$. The first condition is directly satisfied by $\mathcal{R}_i$ (for each initial condition in $ M^{c} $ we can find an initial condition in $ M(S_{\Xi}) $ and vice versa).    
    Using the definition of the output map for approximate composition, relation \eqref{Eq10}, and condition (ii) \hl{of Special Case \ref{SS: scase/bisim1}, we have:
	
\begin{align*}
    \mathbb E\big[\mathbf{d}(H(\x), H(\bar{x}))^q\big]^{\frac1q} &= \Big(\mathbb E\Big[\max_{i \in I} \mathbb E\big[\mathbf{d}_{Y_i}\big(H_i(\x_i), H_i(\bar{x}_i)\big)^q\big]\Big]\Big)^{\frac 1q}\\
    &\le \Big(\max_{i \in I} \vartheta_i^q\Big)^{\frac 1q} = \vartheta,
\end{align*}}
	where the last inequality follows from the fact that $(\x_i,\bar{x}_i) \in \mathcal{R}_i$ for all $i \in I$.	
	Consider $(\x,\bar x) \in \mathcal{R}$ and any joint input $u^{ext} = (u_1^{ext},\ldots,u_N^{ext}) \in U_{M^c}(x)$; choose the same input on the other side, $u^{ext} \in U_{M(S_\Xi)}(\bar x)$, so that $\mu = 0$. \hl{By the interconnection structure (condition (iii) in compatibility for $\Xi$-approximate composition), the composed stochastic kernels factor as $\mathbb{T}_C(\cdot \vert x, u^{ext}) = \prod_{i \in I} \mathbb{T}_i(\cdot \vert x_i, u_i^{ext}, u_i^{int})$, and Dirac measure composed stochastic kernel, as defined in Remark \ref{SS: rem/gMDP-TS}, $(\mathbb{T}_{\Delta})_\Xi(\cdot \vert \bar x, u^{ext}) = \prod_{i \in I} \mathbb{T}_{\Delta_i}(\cdot \vert \bar x_i, u_i^{ext}, u_i^{int})$.
    Define the composed coupling as the product measure of the individual couplings,
    
\begin{equation*}
    \mathbb{W}_{\mathbb{T}}(\cdot \vert u^{ext}, x, \bar x) := \prod_{i \in I} \mathbb{W}_i(\cdot \vert u_i^{ext}, u_i^{int}, x_i, \bar x_i).
\end{equation*}
    This is a valid probability measure on $X \times X$ with marginals $\mathbb{T}_C(\cdot\vert x,u^{ext})$ and $(\mathbb{T}_{\Delta})_\Xi(\cdot\vert \bar x,u^{ext})$, satisfying (i)--(ii) of Definition \ref{DefinitionLifting}. Moreover, since $\mathcal{R} = \prod_{i\in I}\mathcal{R}_i$,
    
\begin{equation*}
    \mathbb{W}_{\mathbb{T}}(\mathcal{R} \vert u^{ext},x,\bar x) = \prod_{i \in I} \mathbb{W}_i(\mathcal{R}_i \vert u_i^{ext},u_i^{int},x_i,\bar x_i) \ge \prod_{i\in I} 1 = 1 \ge 1 - 0.
\end{equation*}
    Fix $i \in I$. Under the coupling $\mathbb{W}_i(\cdot\vert u_i^{ext},u_i^{int},x_i,\bar x_i)$, let $(\x_i',\bar x_i')$ denote the coupled successor pair, with $x_i' \sim \x_i'
    $ generated by $\mathbb{T}_i(\cdot\vert x_i,u_i^{ext},u_i^{int})$ and $\bar x_i'$ by ${\mathbb{T}}_{\Delta_i}(\cdot\vert \bar x_i,u_i^{ext},u_i^{int})$. Here $u_i^{int}$ denotes the internal input received by subsystem $i$, generated by the interconnection map (via $\Xi$ on the $M(S_\Xi)$ side and $\mathbf{0}_N$ on the $M^c$ side) as a function of neighboring outputs.} Using the composed output map \eqref{Eq10}, the $\delta$-ISS-$M_q$ property of $M_i(\Sigma)$, and induction used in the proof of Theorem \ref{THM-NoiseFree} \hl{with triangle inequality for input mismatch}, we obtain\hl{
\begin{align}\label{Eq}
            &\mathbb{E}[\mathbf d(H(\x^{\prime}), H(\bar{x}^{\prime}))^q]^{\frac 1q} \nonumber \\
            &\leq \mathbb E [\max _{i \in I}\mathbb{E}[\|\x_i' - {\Delta}_{i}(\bar{x}_{i},u_{i}^{ext},u_{i}^{int})\|^q]]^{\frac 1q} \nonumber\\
            &\leq (\max_{i \in I} \mathcal{B}_{\vartheta}^{i}  + \mathbb E[\gamma_i^{int}(\kappa_i +  \mathbf{d}_{Y}\left(H(\x_{i}), y_i\right))])^{\frac 1q} \nonumber \\
            &\leq (\max_{i \in I} \mathcal{B}_{\vartheta}^{i}  +2^{q-1}(\gamma_i^{int}(\kappa_i) +\gamma_i^{int}(\vartheta_i)))^{\frac1q} \leqslant \hl{\vartheta},
    \end{align}}
 Hence $(\x',\bar x') \in \mathcal{R}$ \hl{with probability 1 under the coupling $\mathbb{W}_{\mathbb{T}}$, i.e., $\mathbb{W}_{\mathbb{T}}(\mathcal{R}\vert u^{ext},x,\bar x) \ge 1 = 1-0$, establishing $\mathbb{T}_C(\cdot\vert x,u^{ext})\,\bar{\mathcal{R}}_0\,{(\mathbb{T}_\Delta)}_\Xi(\cdot\vert\bar x,u^{ext})$, i.e., $M^c \preccurlyeq^{\vartheta,0}_0 M(S_\Xi)$.}
    By symmetry of the construction (swapping the roles of $x,\bar x$ and using the stochastic simulation in the opposite direction from each $\mathcal{R}_i$), the same coupling argument gives \hl{$M(S_\Xi) \preccurlyeq^{\vartheta,0}_0 M^c$}.
    Hence condition (iii) of Definition \ref{DefinitionBisimulation} holds in both directions, and we conclude \hl{$M^c \approx^{\vartheta,0}_0 M(S_\Xi)$}.
    \hfill $\qed$
\end{pf}
We now present the analogous result for approximate stochastic alternating bisimulation relations.
\begin{thm}\label{THM17}
	Let $\left\langle \hl{{M}_{i}({\Sigma})}\right\rangle_{i \in I}^{\mathbf{0}_{N}, \mathcal{I}} = \left(X, X^{0}, U^{ext}, \hl{\mathbb{T}_{C}}, Y, \right. $ $\left. H\right)$ be interconnected gMDPs and $\left\langle {S}_{i}(\bar{\Sigma})\right\rangle_{i \in I}^{{\Xi}, \mathcal{I}} = \left(X, X^{0}, U^{ext}, \Delta_{\Xi}, Y, \right. $ $\left. H\right)$ be interconnected transition system with $\hl{M_{i}(\Sigma)}=(X_{i}, X_{i}^{0}, U_{i}^{ext}, U_{i}^{int}, \hl{\mathbb{T}_{i}}, Y_{i}, H_{i})$ and $S_{i}(\bar{\Sigma})=(X_{i}, X_{i}^{0}, U_{i}^{ext}, U_{i}^{{int}}, \Delta_{i}, Y_{i}, H_{i})$. Consider positive constants $\vartheta_{i}$, with $\hl{\vartheta^q}=\max _{i \in I} \hl{\vartheta_{i}^q}$ and consider $\hl{\Xi}=\left(\kappa_{1}, \ldots, \kappa_{N}\right) $. Let the following hold:
\begin{itemize}
	\item[$(\mathrm{i})$]  For all $i \in I$, $ M_{i}(\Sigma) $ is $\delta$-ISS-$M_{q}$ and satisfies the inequality $ \max_{i \in I} \mathcal{B}_{\vartheta}^{i} + \hl{2^{q-1}}(\gamma_i^{int}(\kappa_i) + \gamma_i^{int}(\hl{\vartheta_i}))\leqslant \hl{\vartheta^q}$;
    \item[$(\mathrm{ii})$]  For all $i \in I$, $ M_{i}(\Sigma) $ is $(\vartheta_{i},0,0)$-approximately stochastic alternating bisimilar to $ M(S_{i}(\bar{\Sigma})) $, denoted as $ M_{i}(\Sigma) \widehat \approx^{\vartheta_{i},0}_0 M(S_{i}(\bar{\Sigma})) $;
	\item[$(\mathrm{iii})$] $\left\langle M({S}_{i}(\bar{\Sigma}))\right\rangle_{i \in I}^{{\Xi}, \mathcal{I}}$ is compatible for $\Xi$-approximate composition with respect to $\mathcal{I}$;
\end{itemize}
then the concrete system ${M}^{c}=\left\langle {M}_{i}({\Sigma})\right\rangle_{i \in I}^{\mathbf{0}_{N}, \mathcal{I}}$ is $(\vartheta,0,0)$-approximately \hl{stochastic} alternating bisimilar to $M(S_{{\Xi}})=\left\langle M({S}_{i}(\bar{\Sigma}))\right\rangle_{i \in I}^{{\Xi}, \mathcal{I}}$ \hl{as in Special Case \ref{SS: scase/abisim1}.}
\end{thm}
\begin{pf}
The proof follows the same construction as in the proof of Theorem \ref{THM16} (stochastic bisimulation case), with the composed relation $\mathcal{R} = \{(\x,\bar x) \in X\times X \mid (\x_i,\bar x_i)\in\mathcal{R}_i,\ i\in I\}$, where each $\mathcal{R}_i$ is now the $(\vartheta_i,0,0)$-approximate stochastic  {alternating} bisimulation relation witnessing $M_i(\Sigma) \widehat\approx^{\vartheta_i,0}_0 M(S_i(\bar\Sigma))$, together with the composed coupling $\mathbb{W}_{\mathbb{T}} = \prod_{i \in I}\mathbb{W}_i$. Conditions (i) and (ii) of Definition \ref{DefinitionAlternatingBisimulation} are verified exactly as in Theorem \ref{THM16}.
    The only difference lies in condition (iii): the input quantifiers are swapped, so we must start from $u^{ext} \in U_{M(S_\Xi)}(\bar x)$ and find a matching $u^{ext} \in U_{M^c}(x)$ (rather than the reverse). Since $M_i(\Sigma)$ and $M(S_i(\bar\Sigma))$ share the same external input space $U_i^{ext}$ for each $i\in I$, choosing identical inputs on both sides still gives $\mu=0$ regardless of the quantifier order. \hl{With this input choice, the same coupling $\mathbb{W}_i$ and the same chain of inequalities \eqref{Eq}, using the $\delta$-ISS-$M_q$ property, Lemma \ref{lem10} for the coupled successors $(\x_i',\bar x_i')$,
    
\begin{align*}
    \mathbb{E}[\|H(\x')-H(\bar x')\|^q]^{\frac 1q} &\le \Big(\max_{i\in I}\mathcal{B}_\vartheta^i + 2^{q-1}\big(\gamma_i^{int}(\kappa_i)+\gamma_i^{int}(\vartheta_i)\big)\Big)^{\frac 1q} \le \vartheta,
\end{align*}
    so $(\x',\bar x')\in\mathcal{R}$ with probability 1 under $\mathbb{W}_{\mathbb{T}} = \prod_{i \in I}\mathbb{W}_i$, i.e., $\mathbb{W}_{\mathbb{T}}(\mathcal{R}\vert u^{ext},x,\bar x) \ge 1-0, x \sim \x$. This establishes $\mathbb{T}_C(\cdot\vert x,u^{ext})\,\bar{\mathcal{R}}_0\,(\mathbb{T}_\Delta)_\Xi(\cdot\vert\bar x,u^{ext})$ in the alternating order required by Definition \ref{DefinitionAlternatingBisimulation}, condition (iii), giving $M^c \widehat\preccurlyeq^{\vartheta,0}_0 M(S_\Xi)$.}
    Repeating the argument with the roles of $M^c$ and $S_\Xi$ exchanged gives $M(S_\Xi) \widehat\preccurlyeq^{\vartheta,0}_0 M^c$. Hence $M^c \widehat\approx^{\vartheta,0}_0 M(S_\Xi)$.
    \hfill $\qed$
\end{pf}

\subsection{Approximate Stochastic (Alternating) Bisimilar {Relations for} the Global Interconnected Stochastic System and its noise-free abstraction.}
The compositionality results for the approximate stochastic bisimulation relation {are} presented in the following corollary.
\begin{cor}\label{cor1} 
	Let $\{M_{i}(\Sigma)\}_{i \in I}$, $\{S_{i}(\bar{\Sigma})\}_{i \in I}$ and  $\{\hat{S}_{i}(\bar{\Sigma})\}_{i \in I}$ be three collections of systems with subsystems representing {gMDPs} for DT-SCS as defined in \eqref{dtss-trans}, {their} noise-free transition system {counterparts} as defined in \eqref{dtss-nf-trans}, and the {abstractions} of the noise-free {counterparts} using Definition \ref{DefTransitionSys}, respectively. Here $M_{i}(\Sigma)=(X_{i}, X_{i}^{0}, U_{i}^{ext}, U_{i}^{{int}}$, $ \mathbb{T}_{i}, Y_{i}, H_{i})$, $S_{i}(\bar{\Sigma})=(X_{i}, X_{i}^{0}$, $ U_{i}^{ext}, U_{i}^{{int}}, \Delta_{i}, Y_{i}, H_{i})$ and $\hat{S}_{i}(\bar{\Sigma})=(\hat{X}_{i}, \hat{X}_{i}^{0}, \hat{U}_{i}^{ext}, \hat{U}_{i}^{int}, \hat{\Delta}_{i}$, $ \hat{Y}_{i}, \hat{H}_{i})$. Consider positive constants $\vartheta_{i},\varepsilon_{i}, \mu_{i},$ for $ i \in I$, with $\vartheta=\max _{i \in I} \vartheta_{i}$, $\varepsilon=\max _{i \in I} \varepsilon_{i}$, $\mu=\max_{i \in I} \mu_{i}$ and consider $\Xi=\left(\kappa_{1}, \ldots, \kappa_{N}\right) $ and $\hat{\Xi}=(\mu_1+\kappa_1+\varepsilon,\ldots, \mu_N+\kappa_N+\varepsilon)$. If $\left\langle {M}_{i}({\Sigma})\right\rangle_{i \in I}^{\mathbf{0}_{N}, \mathcal{I}}$, $\left\langle M({S}_{i}(\bar{\Sigma}))\right\rangle_{i \in I}^{{\Xi}, \mathcal{I}}$ satisfy the conditions in Theorem \ref{THM16} and $\left\langle {S}_{i}(\bar{\Sigma})\right\rangle_{i \in I}^{{\Xi}, \mathcal{I}}$, $\left\langle \hat{S}_{i}(\bar{\Sigma})\right\rangle_{i \in I}^{\hat{\Xi}, \mathcal{I}}$ satisfy the conditions in Theorem \ref{THM10}, then $M^c = \left\langle {M}_{i}({\Sigma})\right\rangle_{i \in I}^{\mathbf{0}_{N}, \mathcal{I}}$ is  $(2\varepsilon + \vartheta,\mu,0)$-approximately stochastic bisimilar to $M(\hat{S}_{\hat{\Xi}}) = \left\langle M(\hat{S}_{i}(\bar{\Sigma}))\right\rangle_{i \in I}^{\hat{\Xi}, \mathcal{I}}$.
\end{cor}
\begin{pf}
	\hl{By Theorem \ref{THM16}, we have $M^c \approx^{\vartheta,0}_0 M(S_\Xi).$
    By Theorem \ref{THM10}, we have, at the transition-system level, $S_\Xi \approx^{2\varepsilon,\mu} \hat S_{\hat\Xi}$.
    Since $S_\Xi$ and $\hat S_{\hat\Xi}$ are both transition systems, we represent them as gMDPs $M(S_\Xi)$, $M(\hat S_{\hat\Xi})$ with Dirac transition kernels. The condition (iii) of Special Case \ref{SS: scase/bisim2} (bisimulation between transition systems) is precisely condition (iii) of Definition \ref{DefinitionBisimulation} specialized to Dirac kernels with $\hl{\varrho}=0$, the relation witnessing $S_\Xi \approx^{2\varepsilon,\mu} \hat S_{\hat\Xi}$ also witnesses $M(S_\Xi) \approx^{2\varepsilon,\mu}_0 M(\hat S_{\hat\Xi})$. Applying the transitivity of the stochastic bisimulation relation (Proposition \ref{prop1}) to the chain
    $M^c \approx^{\vartheta,0}_0 M(S_\Xi) \approx^{2\varepsilon,\mu}_0 M(\hat S_{\hat\Xi}),$
    yields
    $M^c \approx^{2\varepsilon+\vartheta,\ \mu}_0 M(\hat S_{\hat\Xi}).$
    \hfill $\qed$}
\end{pf}

We now present the analogous result for approximate stochastic alternating bisimulation relations.
\begin{cor}\label{cor18}
	Let $\{M_{i}(\Sigma)\}_{i \in I}$, $\{S_{i}(\bar{\Sigma})\}_{i \in I}$ and  $\{\hat{S}_{i}(\bar{\Sigma})\}_{i \in I}$ be three collections of systems with subsystems representing {gMDPs} for DT-SCS as defined in \eqref{dtss-trans}, {their} noise-free transition system {counterparts} as defined in \eqref{dtss-nf-trans}, and the {abstractions} of the noise-free {counterparts} using Definition \ref{DefTransitionSys}, respectively. Here $M_{i}(\Sigma)=(X_{i}, X_{i}^{0}, U_{i}^{ext}, U_{i}^{{int}}$, $ \mathbb{T}_{i}, Y_{i}, H_{i})$, $S_{i}(\bar{\Sigma})=(X_{i}, X_{i}^{0}$, $ U_{i}^{ext}, U_{i}^{{int}}, \Delta_{i}, Y_{i}, H_{i})$ and $\hat{S}_{i}(\bar{\Sigma})=(\hat{X}_{i}, \hat{X}_{i}^{0}, \hat{U}_{i}^{ext}, \hat{U}_{i}^{int}, \hat{\Delta}_{i}$, $ \hat{Y}_{i}, \hat{H}_{i})$. Consider positive constants $\vartheta_{i},\varepsilon_{i}, \mu_{i},$ for $ i \in I$, with $\vartheta=\max _{i \in I} \vartheta_{i}$, $\varepsilon=\max _{i \in I} \varepsilon_{i}$, $\mu=\max_{i \in I} \mu_{i}$ and consider $\Xi=\left(\kappa_{1}, \ldots, \kappa_{N}\right) $ and $\hat{\Xi}=(\mu_1+\kappa_1+\varepsilon,\ldots, \mu_N+\kappa_N+\varepsilon)$. If $\left\langle {M}_{i}({\Sigma})\right\rangle_{i \in I}^{\mathbf{0}_{N}, \mathcal{I}}$, $\left\langle M({S}_{i}(\bar{\Sigma}))\right\rangle_{i \in I}^{{\Xi}, \mathcal{I}}$ satisfy the conditions in Theorem \ref{THM17} and $\left\langle {S}_{i}(\bar{\Sigma})\right\rangle_{i \in I}^{{\Xi}, \mathcal{I}}$, $\left\langle \hat{S}_{i}(\bar{\Sigma})\right\rangle_{i \in I}^{\hat{\Xi}, \mathcal{I}}$ satisfy the conditions in Theorem \ref{THMCompVerification}, then $M^c = \left\langle {M}_{i}({\Sigma})\right\rangle_{i \in I}^{\mathbf{0}_{N}, \mathcal{I}}$ is  $(2\varepsilon + \vartheta,\mu, 0)$-approximately stochastic {alternating} bisimilar to $M(\hat{S}_{\hat{\Xi}}) = \left\langle M(\hat{S}_{i}(\bar{\Sigma}))\right\rangle_{i \in I}^{\hat{\Xi}, \mathcal{I}}$.
\end{cor}
\begin{pf}
	The proof is straightforward, same as for Corollary \ref{cor1} by applying the transitivity relation of the alternating bisimulation relation in Proposition \ref{prop1} on the conclusions of Theorems \ref{THM17} and \ref{THMCompVerification}. One gets $M^{c} \widehat \approx^{2\varepsilon + \vartheta,\mu}_0 M(\hat{S}_{\hat{\Xi}})$. $ \hfill\qed $ 
\end{pf}

\section{Case study: Temperature Control for Multi-room Building}\label{sec: case study}
This section demonstrates the engineering relevance of the proposed compositional results using a temperature regulation problem for a ten-room building (as shown in Fig. \ref{fig:troom}). The implementation was performed in MATLAB on an AMD Ryzen 9 5950X 16-Core Processor with 125 GB of memory.

\subsection{Model description and control objective}	
Consider the discrete-time stochastic version of the thermal model for a ten-room building \cite{jagtap2020symbolic}, described as:
{\small
	\begin{align}\label{model}
		x_i(k{+}1) &= x_i(k) +  \big[ (-\alpha - \alpha_{ei})\,x_i(k) + \alpha\,x_2(k) \nonumber\\
		&\quad+ \alpha_{ei}\,t_e \big] + gw_i(k), i \in \{1,7,9\},\nonumber \\
		x_2(k{+}1) &= x_2(k) + \big[ (-4\alpha - \alpha_{e2} - \alpha_h u_1(k))\,x_2(k) \nonumber\\ 
		&\quad+ \alpha\,x_1(k) + \alpha\,x_3(k) + \alpha\,x_7(k) + \alpha\,x_9(k) \nonumber\\
		&\quad+ \alpha_{e2}\,t_e + \alpha_h t_h u_1(k) \big] + gw_{2}(k),\nonumber \\
		x_3(k{+}1) &= x_3(k) +  \big[ (-2\alpha - \alpha_{e3})\,x_3(k)\nonumber\\
		&\quad+ \alpha\,x_2(k) + \alpha\,x_4(k) + \alpha_{e3}\,t_e \big] + gw_{3}(k), \nonumber\\
		x_4(k{+}1) &= x_4(k) +  \big[ (-2\alpha - \alpha_{e4})\,x_4(k) \nonumber\\
		&\quad+ \alpha\,x_3(k) + \alpha\,x_5(k) + \alpha_{e4}\,t_e \big] + gw_{4}(k), \nonumber\\
		x_5(k{+}1) &= x_5(k) +  \big[ (-4\alpha - \alpha_{e5} - \alpha_h u_2(k))\,x_5(k)\nonumber\\
		&\quad+ \alpha\,x_4(k) + \alpha\,x_6(k) + \alpha\,x_8(k) + \alpha\,x_{10}(k)\nonumber\\
		&\quad+ \alpha_{e5}\,t_e + \alpha_h t_h u_2(k) \big] + gw_{5}(k), \nonumber\\
		x_i(k{+}1) &= x_i(k) +  \big[ (-\alpha - \alpha_{ei})\,x_i(k) + \alpha\,x_5(k) \nonumber\\
		&\quad+ \alpha_{ei}\,t_e \big] + gw_i(k), i\in\{6,8,10\},
	\end{align}
}
\begin{figure}[h]
	\centering
	\includegraphics[scale=0.23]{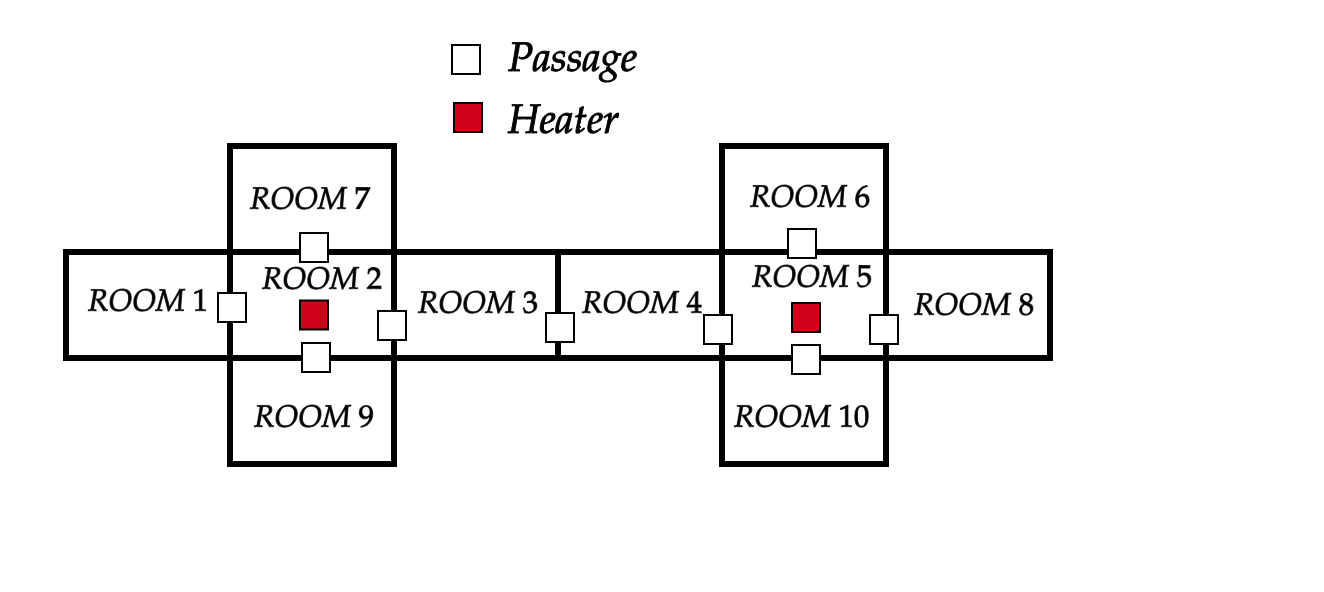}
	\caption{The schematic of a ten-room building.}
	\label{fig:troom}
\end{figure}
where the stochastic process state $x_i(k), i \in I=\{1,2,\dots,10\}$, represents the room temperature in the $i^{t h}$ room in the building, expressed in $^{o}C$, independent different noise \hl{$w_i(k)$ is the truncated gaussian noise of mean $\varpi = 0.2$ and variance $\sigma^2 = 9$, truncated to $\pm 4\sigma$ with negligible discarded mass,} and $g = 0.01$ is the diffusion function. The system parameters are $\alpha = 0.05,\alpha_{e1} = \alpha_{e3} = \alpha_{e4} = \alpha_{e6} = \alpha_{e7} = \alpha_{e8} = \alpha_{e9} = \alpha_{e10} = 0.0033, \alpha_{e2} = \alpha_{e5} = 0.005$ and $\alpha_h = 0.0036$. The external temperature is $t_e = 12 ^{o}C$ and heaters temperature is $t_{h} = 100 ^{o}C$. For each $ i^{th} $ difference equation, the states $ x_{j} $ with $ j \neq i$ represent the internal inputs. The external inputs are given by the joint vector $u(k) = (u_1(k), u_2(k)) \in U$, where $U = \{(1,0),(0,1),(0,0)\}$, where $ 0 $ represents the heater in the off state and $ 1 $ represents the heater in the on state. Here, we are also operating under an additional constraint that both the heaters cannot be switched on at the same time, which results in $(1,1) \not \in U$. We consider the compact state-space $\hl{\bar X}=[0,25]^{10}$, \hl{enlarged as $X=[-0.5,25.5]^{10}$ to ensure validity of kernel $\mathbb T(X|x,u) = 1, \forall x,u,$ under worst-case noise.} The control objective is to synthesize a controller to ensure that the mean room temperatures are inside a safe region $\mathfrak{S}=[18,21.8]^{10}$. \hl{The relation is in the $q^{th}$ moment, so the guarantee covers the mean, not probabilistic safety.}

The model can be seen as an exact composition of $10$ noisy subsystems, as described in \eqref{model}, represented as $M^c = \left\langle {M}_{i}({\Sigma})\right\rangle_{i \in I}^{\mathbf{0}_{N}, \mathcal{I}}$, with
\begin{align*}
    \mathcal{I} =\{ &(1,1),(1,2),(2,2),(2,1),(2,3),(2,7),(2,9),(3,3),(3,2),(3,4),\\
    &(4,4),(4,3),(4,5),(5,5),(5,4),(5,6),(5,8),(5,10),(6,6),(6,5),\\
    &(7,7),(7,2),(8,8),(8,5),(9,9),(9,2),(10,10),(10,5)\}.
\end{align*}	
\hl{First, one can check that each subsystem $M_i(\Sigma)$, $i \in I$, is
$\delta$-ISS-$M_q$ with
$\beta_i(r,k)=4^{q-1}\big((0.947)^k\big)^qr$,
$\gamma^{ext}(r)=4^{q-1}(1.765\,r)^q$,
$\gamma^{int}(r)=4^{q-1}(0.976\,r)^q$ and
$\psi = 4^{q-1}(0.90)^q$.
Each noise-free subsystem $S_i(\bar\Sigma)$, $i \in I$, is $\delta$-ISS with
$\beta_i(r,k)=(0.947)^k r$, $\gamma^{ext}(r)=1.765\,r$ and $\gamma^{int}(r)=0.976\,r$.
For $q=1$ and using the concave function $\bar{\alpha}(r) = 1.011\log_e(1 + r)$ at the noise moment $r=g\,\mathbb E|(1-\lambda)w+\lambda w'|=0.022$ ($\lambda=0.1$, $\|x_1-x_2\|=0.01$) as in Appendix \ref{sec: example}, we obtain a stricter upper bound $\mathcal{B}_{\vartheta}^{(1)} = 0.498 < \mathcal{B}_{\vartheta}^{(2)} = 0.908$, almost 45 \% tighter with respect to $\delta$-ISS-$M_q$ upper bound}, in order to quantify the bisimulation relation parameter $\vartheta_i$ between $M_i(\Sigma)$ and $S_i(\bar{\Sigma})$. 

We compute the local abstraction $\hat{S}_i(\bar{\Sigma})$ for each subsystem $S_i(\bar{\Sigma})$, $i \in I$, using the symbolic approach presented in~\cite{girard2009approximately}. Hence, we have that ${M}^{c}\widehat \preccurlyeq^{(\vartheta + 2\varepsilon,0)}_0M(\hat{S}_{{\Xi}})$,\footnote{Given the safety specification for the original system $\mathfrak{S}$ and since the original system is related to the compositional abstraction by an approximate alternating simulation relation, the abstract specification is a deflated version of the original one.} where $\hat{S}_{\Xi}=\langle \hat{S}_i(\bar{\Sigma}) \rangle_{i\in I}^{\mathbf{0}_{10},\mathcal{I}}$, $\alpha = 0.01, \alpha_e = 0.52, t_e = 18^{\circ} C, \mathcal{B}_{\vartheta}^{(1)} = 0.091$, with a heater in every room, for the abstraction's relation parameter $\varepsilon$ in the case of the first moment (mean).
\hl{
At runtime, the current continuous measurement $x \in X_i$, for each subsystem $i$, is mapped
to a symbol via the deterministic quantizer $q_{\nu}(x)$, defined by
half-open partition cells $\prod_{i} [j_i \nu, (j_i+1)\nu)$ (constructed on the safe set), with
$q_\nu(x)$ as the unique center of partition cell containing $x$. This map is single-valued
by construction, as $\|x - q_\nu(x)\| \;\le\; \nu/2$, and number of discretized cells $n_x = 190$ in the preceding relation study.}
\hl{\hl{Concretely, the local abstraction $\hat S_i(\bar\Sigma)$ is constructed by quantizing the state space $X_i$ via $q_\nu$ and evaluating the successor map on the quantized states, i.e., $\hat x_i' = q_\nu\big(f_i(\hat x_i, u_i^{ext}, u_i^{int})\big)$, which is the standard symbolic-model construction for $\delta$-ISS systems~\cite{girard2009approximately}.}}
\begin{table}[h!]
	\centering
	\caption{Comparison of Computation time for Abstraction Methods across varying discretization parameter $n_x$.}
	\label{tab:abstraction_comparison1}
	\vspace{5pt}
	\begin{tabular}{|p{2.5cm}|c|c|c|c|c|}
		\hline
		\centering Method & \multicolumn{5}{c|}{Computation time (s)} \\
		\hline
		\centering $n_x$ & 2 & 3 & 4 & 5 & 6 \\
		\hline
		\shortstack[l]{\rule{0pt}{12pt}Compositional\\Abstraction $\hat{S}_{\Xi}$} & 1.41 & 1.44 & 1.34 & 1.35 & 1.62\\
		\hline
		\shortstack[l]{\rule{0pt}{12pt}Monolithic\\Abstraction $\hat{S}^{c}$} & 0.176 & 0.49 & 6.05 & 55 & 343\\
		\hline
	\end{tabular}
\end{table}
\hl{ The computation of the symbolic model and the controller synthesis have been performed using the tool SCOTS provided by \cite{rungger2016scots}.}
Figure \ref{fig:resp_traffic} shows the evolution of the mean temperature in each room, for 10000 runs, starting from the initial condition

\[
    x_0 = [19.0, 19.0, 19.0, 19.8, 20.8, 19.8, 19.0, 19.8, 19.0, 19.8]^T,
\]
using a safety controller synthesized for the compositional abstraction constructed of the global interconnected stochastic system $\hat{S}_{{\Xi}}$. The dashed black lines represent the boundary of the safe set for each abstracted section. One can readily see that the mean trajectory evolve within the safe region $\mathfrak{S}$.
\begin{figure}[t!]
	\centering
	\includegraphics[scale=0.2]{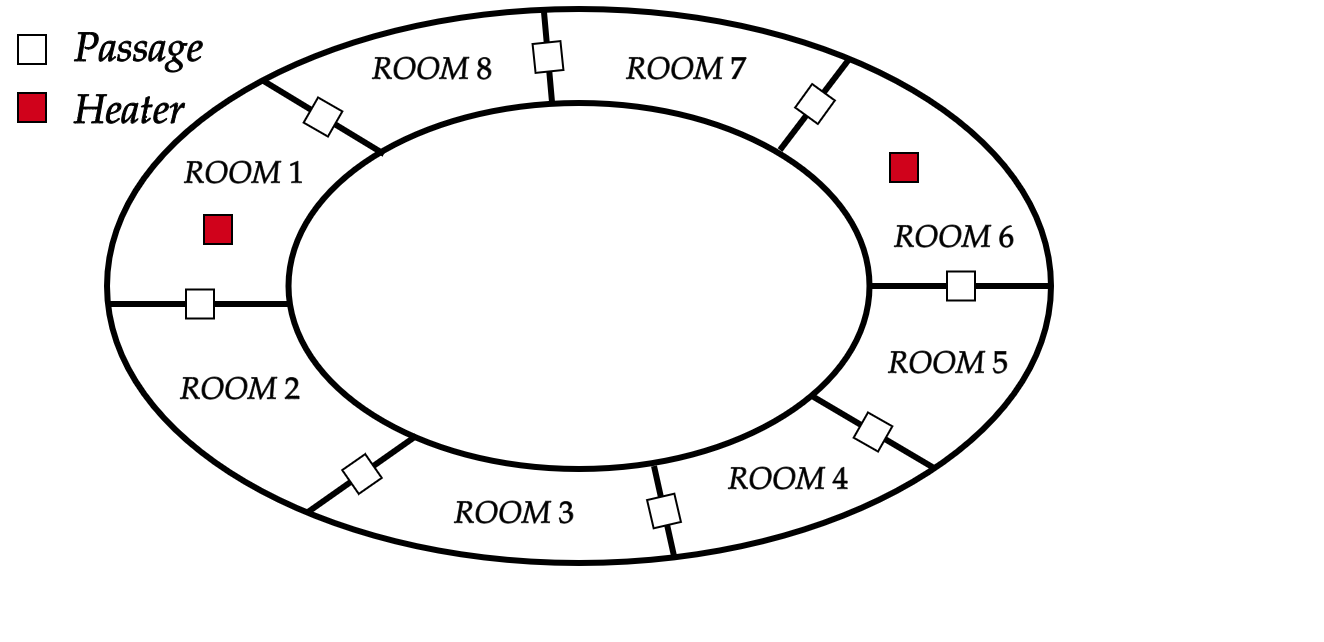}
	\caption{The schematic of an eight-room building in circular cascade interconnection.}
	\label{fig:eroom}
\end{figure}
\hl{The number of discretized cells $n_x$ used in the remainder of this
section are chosen to illustrate the scalability
of the compositional approach relative to the monolithic one, and are
smaller as compared to the one used in relation study $n_x = 190$, as with the number of transitions of order $O(|U| \times n_x^r) = O(3 \times190^{10})$, for monolithic abstraction of $(r=10)$ inter-connected rooms, it breaks the limits by MATLAB memory size preventing any valid comparisons.} 
For the study in Table \ref{tab:abstraction_comparison1}, the computation time of the abstractions of ten rooms individually, with the number of discretized intervals $n_x = 6$, are given by 
\[
    [0.1936, 0.1465, 0.13, 0.1408, 0.1310, 0.1252, 0.1096, 0.1189, 0.1129, 0.1184]
\]
 seconds, and the composition of the global abstraction from local ones using an approximate composition takes less than $0.341$ seconds. Together, this resulted in $1.621$ seconds to compute an abstraction compositionally. Constructing an abstraction monolithically, using the same discretization parameter $n_x = 6$, took $343$ seconds. Hence, the proposed compositional approach is \hl{faster for a higher number of discretized intervals $n_x$} (needed for accuracy) as shown in Table \ref{tab:abstraction_comparison1}. \hl{Here the grid resolution for safe set is $\frac{(21.8 - 18)}{n_x} = 0.63$, the number of abstract states is $n_x^r = (6)^{10}$, where $r$ is the total number of rooms and given $U$ as the input set, the number of transitions is $|U| \times n_x^r = 3 \times 6^{10}$ with memory usage of 0.483 GB.}

\subsection{Case study with Multi-room Building in Circular Cascade Interconnection}
We also compared our proposed compositional abstraction method for different numbers of subsystems $N$ in a circular cascade interconnected global system as shown in Fig. \ref{fig:eroom}, and as defined below $\forall i \in \{1,2,\cdots,N\}$ (where indices are taken modulo $N$):

\begin{align}\label{model_general}
    x_i(k{+}1) &= x_i(k) + (-2\alpha - \alpha_e)\,x_i(k) + \alpha\,x_{i-1}(k) \nonumber\\
    &\quad + \alpha\,x_{i+1}(k) + \alpha_h u_i(k) + \alpha_e\,t_e + g\,w_i(k).
\end{align}
\hl{As the number of subsystems increases, our compositional method scales markedly better than the monolithic approach. It remains comparable at small scale but yields a significant reduction in computation time from eight rooms onward while the monolithic approach fails to converge beyond ten rooms}, as shown in Table \ref{tab:abstraction_comparison2}.

\begin{table}[h!]
	\centering
	\caption{Comparison of Computation time for Abstraction Methods across varying the number of rooms in a circular interconnected multi-room building}
	\label{tab:abstraction_comparison2}
	\vspace{5pt}
	\begin{tabular}{|p{2.5cm}|c|c|c|c|c|}
		\hline
		\centering Method & \multicolumn{5}{c|}{Computation time (s)} \\
		\hline
		Number of rooms & 6 & 8 & 10  & 12  & 14 \\
            \hline
            Grid ($n_x{=}5$) & $5^{6}$ & $5^{8}$ & $5^{10}$ & $5^{12}$ & $5^{14}$ \\
		\hline
		\shortstack[l]{\rule{0pt}{12pt}Compositional\\Abstraction $\hat{S}_{\Xi}$} & 0.79 & 0.93 & 1.25 & 2.19 & 36.46\\
		\hline
		\shortstack[l]{\rule{0pt}{12pt}Monolithic\\Abstraction $\hat{S}^{c}$} & 0.18 & 1.09 & 23 & $\infty^{*}$ & $\infty^{*}$\\
		\hline
	\end{tabular}
\\$\ ^{*}$Here $\infty$ means that the abstraction algorithm did not converge due to exponential runtime growth with respect to the dimensions of the state and inputs.  
\end{table}

\begin{figure}[t!]
	\centering
	\includegraphics[scale=0.5]{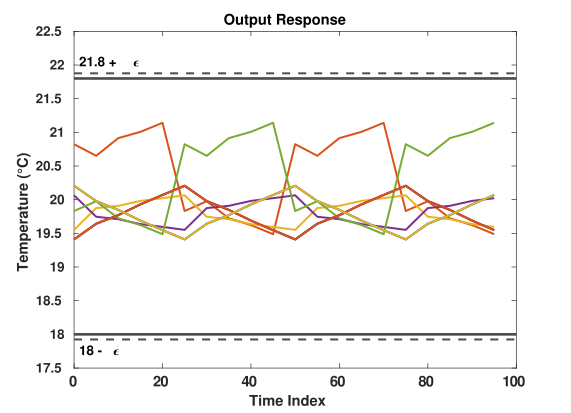}
	\caption{The evolution of mean temperature in a ten-room building.}
	\label{fig:resp_traffic}
\end{figure}

\section{Conclusion}
This paper studied the problem of compositional abstraction of stochastic interconnected nonlinear systems. The compositional framework for constructing abstractions relies on the notion of approximate composition and the $\delta$-ISS-$M_{q}$ property. In particular, given a large system consisting of interconnected components, we provided conditions under which the concept of approximate (alternating) stochastic bisimulation relation is preserved when going from the subsystems to the large interconnected system, with smaller relational conservatism as compared to the one added by $\delta$-ISS-$M_{q}$ property and reduced abstraction computation time with compositional abstraction as compared to monolithic abstraction \hl{only for high-dimensional systems with finer discretization (required for accuracy) and large-scale systems.} Numerical results are presented demonstrating the merits of the theoretical results \hl{with safety in mean.}

\bibliographystyle{plain}
\bibliography{autosam}

\appendix
\section*{Appendix}
\renewcommand{\thesubsection}{A.\arabic{subsection}}

\subsection{$\delta$-ISS-$M_{q}$ for linear DT-SCS}\label{ISSl}

Consider the linear discrete-time stochastic system ($\Sigma$):
\begin{equation}
\label{LinearSys}
    \x(k+1) = A \x(k) + B \uu^{ext}(k) + C \uu^{int}(k) + D \ww(k), 
\end{equation}
where $A \in \mathbb{R}^{n \times n}, B \in \mathbb{R}^{n \times m_{1}}, C \in \mathbb{R}^{n \times m_{2}}$, and $D \in \mathbb{R}^{n \times p}$. The signals $\x, \uu^{ext}, \uu^{int}$ represent the state, external input, and internal input, respectively, while $\ww$ denotes a sequence of $p$-dimensional i.i.d. random variables. 
\begin{thm}\label{THM1}
    Consider the system $\Sigma$ in \eqref{LinearSys}. Let $\varpi_q := (\mathbb{E}[\|\ww(k)\|^q])^{1/q}$. If \hl{$\|A\| < 1$}, then $\Sigma$ is $\delta$-ISS-$M_{q}$ for $q \in \N$. The characterizing functions are defined for $(r,k) \in \mathbb{R}_{0}^+ \times \mathbb{N}_{0}$ by:
    
\begin{align*}
    \beta(r, k) &= 4^{q-1}\big(\|A\|^{k}\big)^{q}\,\hl{r}, & \gamma^{ext}(r) &= 4^{q-1}\Big(\frac{\|B\| r}{1-\|A\|}\Big)^{q},\\
    \gamma^{int}(r) &= 4^{q-1}\Big(\frac{\|C\| r}{1-\|A\|}\Big)^{q}, & \hl{\psi} &\hl{{}= 4^{q-1}\Big(\frac{2\|D\|\varpi_q}{1-\|A\|}\Big)^{q}}.
\end{align*}
\end{thm}
\begin{pf}
    Let $\x^{a}$ and $\x^{b}$ be two trajectories of the system \eqref{LinearSys}. At time $k+1$, the difference between $\x^{a}$ and $\x^{b}$ is
\begin{align*}
    \x^{a}(k+1)-\x^{b}(k+1) &= A (\x^{a}(k)-\x^{b}(k)) + B(\uu^{ext,a}(k)-\uu^{ext,b}(k))\\
    &\quad + C(\uu^{int,a}(k)-\uu^{int,b}(k)) + D(\ww^{a}(k)-\ww^{b}(k)).
\end{align*}
We can rewrite this explicitly as
\begin{align*}
    \x^{a}(k+1)-\x^{b}(k+1) &= A^{k+1} (\x^{a}(0)-\x^{b}(0)) + \sum_{j=0}^{k} A^{k-j}\Big(B (\uu^{ext,a}(j)-\uu^{ext,b}(j))\\
    &\quad + C(\uu^{int,a}(j)-\uu^{int,b}(j)) + D(\ww^{a}(j)-\ww^{b}(j))\Big),
\end{align*}
and, using the generalized Minkowski inequality,
\begin{align*}
    \mathbb{E}\big(\|\x^{a}(k+1)-\x^{b}(k+1)\|^{q}\big)
    &\leqslant 4^{q-1}\Bigg[ \mathbb{E}\big(\|A^{k+1} (\x^{a}(0)-\x^{b}(0))\|^{q}\big)\\
    &\qquad\quad + \mathbb{E}\Big(\Big\|\sum_{j=0}^{k} A^{k-j}B ( \uu^{ext,a}(j)-\uu^{ext,b}(j))\Big\|^{q}\Big)\\
    &\qquad\quad +\mathbb{E}\Big(\Big\|\sum_{j=0}^{k} A^{k-j}C(\uu^{int,a}(j) - \uu^{int,b}(j))\Big\|^{q}\Big)\\
    &\qquad\quad + \mathbb{E}\Big( \Big\|\sum_{j=0}^{k} A^{k-j}D(\ww^{a}(j)-\ww^{b}(j))  \Big\|^{q}\Big)\Bigg].
\end{align*}
Using the Minkowski inequality for the noise term, it follows that
\begin{align*}
    \mathbb{E}\big(\|\x^{a}(k+1)-\x^{b}(k+1)\|^{q}\big)
    &\leqslant 4^{q-1}\Bigg( \|A^{k+1} (\x^{a}(0)-\x^{b}(0))\|^{q} + \Big\|\sum_{j=0}^{k} A^{k-j}B ( \uu^{ext,a}(j)-\uu^{ext,b}(j))\Big\|^{q}\\
    &\qquad\quad +\Big\|\sum_{j=0}^{k} A^{k-j}C(\uu^{int,a}(j) - \uu^{int,b}(j))\Big\|^{q}\\
    &\qquad\quad \hl{{}+\Big\|\sum_{j=0}^{k} A^{k-j}D (2\varpi_q) \Big\|^{q}}\Bigg).
\end{align*}
Hence, if $\|A\| < 1$, system \eqref{LinearSys} satisfies the $\delta$-ISS-$M_{q}$ property with $\beta(r, k), \gamma^{ext}(r), \gamma^{int}(r)$, $\psi$ as defined in Theorem \ref{THM1}. \hfill $\qed$
\end{pf}

\subsection{$\delta$-ISS-$M_{q}$ for Lipschitz nonlinear DT-SCS}\label{ISSnl}

Consider the discrete-time stochastic nonlinear system \eqref{dt_sys}. We posit the following Lipschitz continuity assumption on the nonlinear maps $f$ and $g$.

\begin{assum}
    \label{assum:4}
    There exist positive constants $L_{f}^{x}, L^{u^{ext}}_f,$ $ L^{u^{int}}_f, L^{x}_g \in \mathbb{R}^{+}$ such that for all $x^{a}, x^{b} \in \mathcal{X}$, $u^{int,a}, u^{int,b} \in \mathcal{U}^{int}$, and $u^{ext,a}, u^{ext,b} \in \mathcal{U}^{ext}$, we have
\begin{align}
    \|f(x^{a}, u^{ext,a},u^{int,a})-f(x^{b}, u^{ext,b},u^{int,b})\| &\leqslant  L_{f}^{x}\|x^{a}-x^{b}\| + L_{f}^{u^{ext}}\|u^{ext,a}-u^{ext, b}\| + L_{f}^{u^{int}}\|u^{int,a}-u^{int,b}\|, \label{Eq28}\\
    \|g(x^{a})-g(x^{b})\| &\leqslant L_{g}^{x}\|x^{a}-x^{b}\|.
\end{align}
\end{assum}

\begin{thm}\label{THM2}
    Let $\Sigma$ be a nonlinear system as in \eqref{dt_sys} satisfying Assumption \ref{assum:4}. \hl{Let $\varpi_q := (\mathbb{E}[\|\ww(k)\|^q])^{1/q}, G := \sup_{x \in \mathcal X}\|g(x)\| < \infty$.} If the contraction condition $L_{f}^{x} + L_{g}^{x}\varpi_q < 1$ holds, then $\Sigma$ is $\delta$-ISS-$M_{q}$ for $q \in \N$. The characterizing functions are defined for $(r,k) \in \mathbb{R}_{0}^+ \times \mathbb{N}_{0}$ as
    
\begin{align*}
    \beta(r, k) &= 4^{q-1}\Big( (L_{f}^{x}+L_{g}^{x}\varpi_q)^{k}  \Big)^{q}\,\hl{r},\\
    \gamma^{ext}(r) &= 4^{q-1}\Big(\frac{L_{f}^{u^{ext}} r}{1-(L_{f}^{x}+L_{g}^{x}\varpi_q)}\Big)^{q},\\
    \gamma^{int}(r) &= 4^{q-1}\Big(\frac{L_{f}^{u^{int}} r}{1-(L_{f}^{x}+L_{g}^{x}\varpi_q)}\Big)^{q},\\
    \hl{\psi} &\hl{{}= 4^{q-1}\Big(\frac{ 2G\varpi_q}{1-(L_{f}^{x}+L_{g}^{x}\varpi_q)}\Big)^{q}}.
\end{align*}
\end{thm}
\begin{pf}
    Let $\x^{a}$ and $\x^{b}$ be two trajectories of system \eqref{dt_sys}, satisfying Assumption \ref{assum:4}. At time $k+1$, we have
\begin{align*}
    \mathbb{E}\big[\|\x^{a}(k+1)-\x^{b}(k+1)\|^q\big]^{\frac{1}{q}}
    = \mathbb{E}\Big[\Big\|&f(\x^{a}(k),\uu^{ext,a}(k),\uu^{int,a}(k)) - f(\x^{b}(k),\uu^{ext,b}(k),\uu^{int,b}(k))\\
    &+ g(\x^{a}(k))\ww^{a}(k)-g(\x^{b}(k))\ww^{b}(k) \Big\|^q \Big]^{\frac{1}{q}}.
\end{align*}
Using Assumption \ref{assum:4} and applying the Minkowski inequality, if $\rho \triangleq (L_{f}^{x}+L_{g}^{x}\varpi_q)<1$ is satisfied, the recurrence is stable. Iterating for $q \in \N$, we obtain
\begin{align*}
    \mathbb{E}\big[\|\x^{a}(k+1)-\x^{b}(k+1)\|^{q}\big]
    \leqslant 4^{q-1}\Bigg[ &\big(\rho^{k+1}\|\x^{a}(0)-\x^{b}(0)\|\big)^q
    + \Big(\sum_{j=0}^{k}\rho^{k-j} L_{f}^{u^{ext}}\|\uu^{ext,a}(j)-\uu^{ext,b}(j)\|\Big)^q\\
    &+ \Big(\sum_{j=0}^{k}\rho^{k-j} L_{f}^{u^{int}}\|\uu^{int,a}(j)-\uu^{int,b}(j)\|\Big)^q
    \hl{{}+\Big(\sum_{j=0}^{k}\rho^{k-j}\,2G\varpi_q\Big)^q}\Bigg].
\end{align*}
If $(L_{f}^{x}+L_{g}^{x}\varpi_q) < 1$, the $\delta$-ISS-$M_{q}$ property is satisfied and $\beta(r, k), \gamma^{ext}(r)$, $\gamma^{int}(r), \psi,$ are defined as in Theorem \ref{THM2}. \hfill $\qed$
\end{pf}

\begin{rem}
\hl{    
The characterizing functions of the $\delta$-ISS-$M_q$ property in
Theorems~\ref{THM1} and~\ref{THM2} are derived under the condition that
both compared trajectories are driven by different noise
realizations $\ww^a(k), \ww^b(k), k \in \mathbb{N}_0$. Under the assumption that the compared trajectories are driven by the same noise realization, independent of initial conditions, the $\delta$-ISS-$M_q$ upper bound on their distance is obtained with the same $\beta, \gamma$ functions but
$\psi=0$.}
\end{rem}

\subsection{Verification procedure for Assumption \ref{Assum8}}\label{sec: example}
\label{app:verification-assum10}
{
Assumption \ref{Assum8} requires the existence of a scalar $\lambda \in (0,1)$ and a
function $\bar\alpha(k)$.
It can be checked through the following five steps,
applicable to any DT-SCS $\Sigma$ satisfying the $\delta$-ISS-$M_q$ property of
Definition \ref{DefD-ISS-M-q} with known characterizing functions $\beta$, $\psi$
(e.g.\ as given in Theorem \ref{THM1} for linear systems or Theorem \ref{THM2} for
Lipschitz nonlinear systems).}
\begin{enumerate}{\color{hlcolor}
    \item {Choose $\lambda$ so that the target bound is strictly positive.}
    The right-hand side of \eqref{mainkey} contains the term
    $(4^{-(q-1)} - \lambda)\psi$. Since $4^{-(q-1)} > 0$ for any $q \in \N$ and
    $\lambda$ is a design parameter, selecting
    $\lambda \in \big(0,\, 4^{-(q-1)}\big)$
    ensures this term is strictly positive whenever $\psi > 0$.
    \item {Bound the noise-combination term.}
    Bound
    $r_k := \mathbb{E}\big[\|(1-\lambda)g_k(\x_1)\ww_k + \lambda g_k(\x_2)\ww'_k\|^q\big]$
    depending on the model class:
    \begin{itemize}
        \item  {Linear systems} \eqref{LinearSys}: using $g(x) \equiv D$,
        $ r_k \;\leq\; \|D\|^q\, \mathbb{E}\big[\|(1-\lambda)\ww_k + \lambda \ww'_k\|^q\big]. $
        \item  {Lipschitz nonlinear systems} (Assumption \ref{assum:4}): using
        $\|g(x)\| \leq L_g^x \|x\| + \|g(0)\|$ for $x \in \mathcal{X}$ (compact) and
        the noise moment $\varpi_q := (\mathbb{E}[\|\ww(k)\|^q])^{1/q}$,
        $ r_k \;\leq\; \big(L_g^x \sup_{x \in \mathcal{X}} \|x\| + \|g(0)\|\big)^q (\varpi_q)^q. $
    \end{itemize}
    In both cases the bound is independent of $k$, owing to the stationarity of the
    noise sequence and the compactness of $\mathcal{X}$.
    \item {Compute the envelope constant $R$.}
    Set $R$ equal to the (finite, $k$-independent) bound obtained in Step 2, i.e.
    $R := \sup_{k \in \mathbb{N}_0} r_k$. 
    Finiteness of $R$ is guaranteed by
    compactness of $\mathcal{X}$ and boundedness of the noise moment $\varpi_q$, so this
    step never fails for a well-posed DT-SCS.
    \item {Check feasibility and fix the concave majorant.}
    Consider $L_f^x \leq 0.5, (L_f^x)^q\bar{S}_{\alpha} \leq \bar \alpha(R)$. Verify that there exists $a \in \mathbb R^{+}$ such that
    
\[
    R \leq \exp\!\left(\frac{T_1 + T_2 + (4^{-(q-1)} - \lambda)\psi}{2a}\right) - 1, \qquad a \geq \frac{R}{\log(1+R)},
\]
    where $T_1 = (4^{-(q-1)} - \lambda)\beta(\|x_1 - x_2\|^q, k), T_2 = (\lambda - 1)\big(\beta(\|x_1 - x_2\|, k)\big)^q$. 
    In
    particular, choosing the same $a$
    guarantees $\bar\alpha(r_k) := a\log(1+r_k)$ satisfies $\bar\alpha(r_k) \geq r_k$ for all
    $r \in [0, R]$, by concavity of $\log(1+\cdot)$. This defines the required
    concave function $\bar\alpha$.
    \item {Confirm Assumption \ref{Assum8}.}
    With $\lambda$ from Step 1 and $\bar\alpha$ from Step 4, confirm that
    
\begin{align*}
    2\bar\alpha(r_k) \leq{}& (4^{-(q-1)} - \lambda)\beta(\|x_1 - x_2\|^q, k)\\
    &+ (\lambda - 1)\big(\beta(\|x_1 - x_2\|, k)\big)^q + (4^{-(q-1)} - \lambda)\psi
\end{align*}
    holds for all $k \in \mathbb{N}_0$. If the inequality fails for the chosen
    $\lambda$, $R$, and $a$, tighten $\lambda$ (moving it closer to $0$) or reduce
    $R$ (e.g.\ by shrinking $\mathcal{X}$ or improving the Lipschitz/noise-moment
    estimates in Step 2), and repeat from Step 1.}
\end{enumerate}
\begin{rem} {\color{hlcolor}The
right-hand side of \eqref{mainkey} can be made strictly positive
(whenever $\psi \in \mathbb R^{+}$) by an appropriate choice of $\lambda$, and that a valid
majorant $\bar\alpha$ can be constructed once the state space is compact
and the noise moment is finite, conditions already assumed throughout the paper.
}
\end{rem}

\end{document}